\pdfoutput=1
\newif\ifdraft\draftfalse % DISABLE THIS BEFORE SUBMITTING!!
\newif\iffull\fulltrue % DISABLE THIS BEFORE SUBMITTING!!

\newcommand{\docroot}{.}

\documentclass[preprint, nocopyrightspace]{sigplanconf}

\usepackage{local}

\usepackage[firstpage]{draftwatermark}
\SetWatermarkText{
    \hspace*{7in}\raisebox{6.4in}{\includegraphics[scale=0.1]{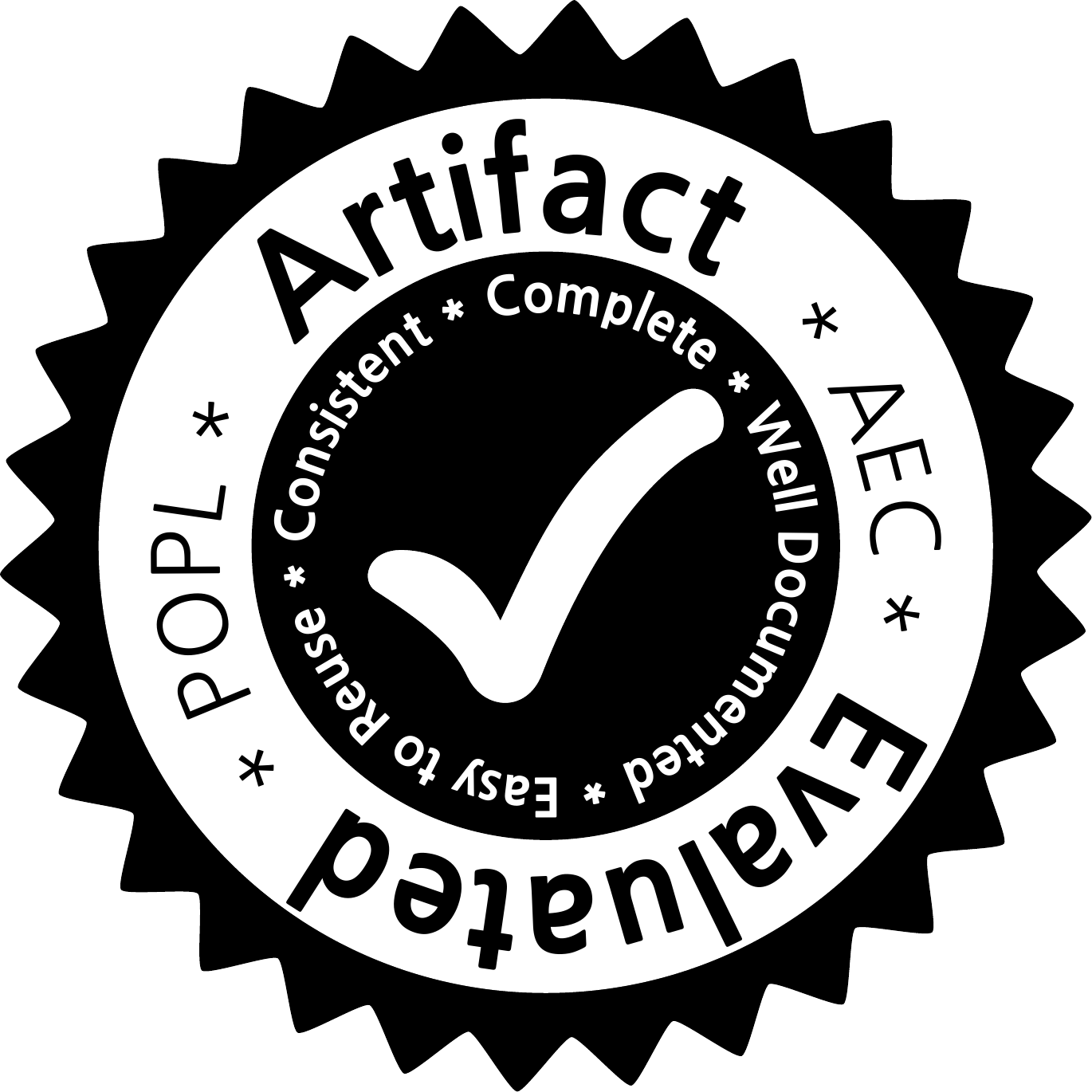}}}
\SetWatermarkAngle{0}
\begin{document}
\toappear{}

\title{Cantor Meets Scott: Semantic \\[.15em] Foundations for Probabilistic Networks}
% \title{Cantor meets Scott: Domain-Theoretic \\[.15em] Foundations for Probabilistic Network Programming}

\authorinfo{Steffen Smolka}{Cornell University, USA}{}
\authorinfo{Praveen Kumar}{Cornell University, USA}{}
\authorinfo{Nate Foster}{Cornell University, USA}{}
\authorinfo{Dexter Kozen}{Cornell University, USA}{}
\authorinfo{Alexandra Silva}{University College London, UK}{}

\maketitle

\begin{abstract}
ProbNetKAT is a probabilistic extension of NetKAT with a denotational
semantics based on Markov kernels. The language is expressive enough
to generate continuous distributions, which raises the question of how
to compute effectively in the language. This paper gives an new
characterization of ProbNetKAT's semantics using domain theory, which
provides the foundation needed to build a practical implementation. We
show how to use the semantics to approximate the behavior of arbitrary
ProbNetKAT programs using distributions with finite support. We
develop a prototype implementation and show how to use it to solve a
variety of problems including characterizing the expected congestion
induced by different routing schemes and reasoning probabilistically
about reachability in a network.
\end{abstract}



%%%%%%%%%%%%%%%%%%%%%%%%%%%%%%%%%%%%%%%%%%%%%%%%%%%%%%%%%%%%%%%%%%%%%%%%%%%%%%%%
%%%%%%%%%%%%%%%%%%%%%%%%%%%%%%%%%%%%%%%%%%%%%%%%%%%%%%%%%%%%%%%%%%%%%%%%%%%%%%%%

\section{Introduction}
\label{sec:intro}

The recent emergence of software-defined networking (SDN) has led to
the development of a number of domain-specific programming
languages~\cite{frenetic-icfp11,composing-nsdi13,maple,nelson:flowlog}
and reasoning tools~\cite{hsa,veriflow,\nkpapers} for networks. But
there is still a large gap between the models provided by these
languages and the realities of modern networks. In particular, most
existing SDN languages have semantics based on deterministic
packet-processing functions, which makes it impossible to encode
probabilistic behaviors. This is unfortunate because in the real
world, network operators often use randomized protocols and
probabilistic reasoning to achieve good performance.

Previous work on ProbNetKAT~\cite{\pnkpaper} proposed an extension to
the \NK\ language~\cite{\nkpapers} with a random choice operator that
can be used to express a variety of probabilistic behaviors.
ProbNetKAT has a compositional semantics based on Markov kernels that
conservatively extends the deterministic \NK\ semantics and has been
used to reason about various aspects of network performance including
congestion, fault tolerance, and latency. However, although the
language enjoys a number of attractive theoretical properties, there
are some major impediments to building a practical implementation: (i)
the semantics of iteration is formulated as an infinite process rather
than a fixpoint in a suitable order, and (ii) some programs generate
continuous distributions. These factors make it difficult to determine
when a computation has converged to its final value, and there are
also challenges related to representing and analyzing distributions
with infinite support.

This paper introduces a new semantics for \PNK, following the approach
pioneered by Saheb-Djahromi, Jones, and
Plotkin~\cite{Saheb-Djahromi80,SahebDjahromi78,Jones89,Plotkin82,JonesPlotkin89}.
Whereas the original semantics of ProbNetKAT was somewhat imperative
in nature, being based on stochastic processes, the semantics
introduced in this paper is purely functional. Nevertheless, the two
semantics are closely related---we give a precise, technical
characterization of the relationship between them. The new semantics
provides a suitable foundation for building a practical
implementation, it provides new insights into the nature of
probabilistic behavior in networks, and it opens up several
interesting theoretical questions for future work.

Our new semantics follows the order-theoretic tradition established in
previous work on Scott-style domain
theory~\cite{Scott72,Abramsky94domaintheory}. In particular,
Scott-continuous maps on algebraic and continuous DCPOs both play a
key role in our development. However, there is an interesting
twist: \NK\ and \PNK\ are not \emph{state-based} as with most other
probabilistic systems, but are rather \emph{throughput-based}. A \PNK\
program can be thought of as a filter that takes an input set of
packet histories and generates an output randomly distributed on the
measurable space $\pH$ of sets of packet histories. The closest thing
to a ``state'' is a set of packet histories, and the structure of
these sets (e.g., the lengths of the histories they contain and the
standard subset relation) are important considerations. Hence, the
fundamental domains are not flat domains as in traditional domain
theory, but are instead the DCPO of sets of packet histories ordered
by the subset relation. Another point of departure from prior work is
that the structures used in the semantics are not subprobability
distributions, but genuine probability distributions: with probability
$1$, some set of packets is output, although it may be the empty set.

It is not obvious that such an order-theoretic semantics should exist
at all. Traditional probability theory does not take order and
compositionality as fundamental structuring principles, but prefers to
work in monolithic sample spaces with strong topological properties
such as Polish spaces. Prototypical examples of such spaces are the
real line, Cantor space, and Baire space. The space of sets of packet
histories $\pH$ is homeomorphic to the Cantor space, and this was the
guiding principle in the development of the original ProbNetKAT
semantics. Although the Cantor topology enjoys a number of attractive
properties (compactness, metrizability, strong separation) that are
lost when shifting to the Scott topology, the sacrifice is compensated
by a more compelling least-fixpoint characterization of iteration that
aligns better with the traditional domain-theoretic
treatment. Intuitively, the key insight that underpins our development
is the observation that ProbNetKAT programs are monotone: if a larger
set of packet histories is provided as input, then the likelihood of
seeing any particular set of packets as a subset of the output set can
only increase. From this germ of an idea, we formulate an
order-theoretic semantics for ProbNetKAT.

In addition to the strong theoretical motivation for this work, our
new semantics also provides a source of practical useful reasoning
techniques, notably in the treatment of iteration and
approximation. The original paper on ProbNetKAT showed that the Kleene
star operator satisfies the usual fixpoint equation $P\star
= \skp \pcomp P\cmp P\star$, and that its finite approximants $\pp n$
converge weakly (but not pointwise) to it. However, it was not
characterized as a least fixpoint in any order or as a canonical
solution in any sense. This was a bit unsettling and raised questions
as to whether it was the ``right'' definition---questions for which
there was no obvious answer. This paper characterizes $P\star$ as the
least fixpoint of the Scott-continuous map $X\mapsto\skp \pcomp P\cmp
X$ on a continuous DCPO of Scott-continuous Markov kernels. This not
only corroborates the original definition as the ``right'' one, but
provides a powerful tool for monotone approximation. Indeed, this
result implies the correctness of our prototype implementation, which
we have used to build and evaluate several applications inspired by
common real-world scenarios.

\paragraph*{Contributions.}
This main contributions of this paper are as follows: (i) we develop a
domain-theoretic foundation for probabilistic network programming,
(ii) using this semantics, we build a prototype implementation of the
ProbNetKAT language, and (iii) we evaluate the applicability of the
language on several case studies.

\paragraph*{Outline.}
The paper is structured as follows. In \S\ref{sec:overview} we give a
high-level overview of our technical development using a simple
running example. In \S\ref{sec:primer} we review basic definitions from
domain theory and measure theory. In \S\ref{sec:syntax} we formalize the
syntax and semantics of \PNK\ abstractly in terms of a
monad. In \S\ref{sec:cantor} we prove a general theorem relating the Scott
and Cantor topologies on $\pH$. Although the Scott topology is much
weaker, the two topologies generate the same Borel sets, so the
probability measures are the same in both. We also show that the bases
of the two topologies are related by a countably infinite-dimensional
triangular linear system, which can be viewed as an infinite analog of
the inclusion-exclusion principle. The cornerstone of this result is
an extension theorem (Theorem \ref{thm:extension}) that determines
when a function on the basic Scott-open sets extends to a
measure. In \S\ref{sec:DCPO} we give the new domain-theoretic semantics
for \PNK\ in which programs are characterized as Markov kernels that
are Scott-continuous in their first argument. We show that this class
of kernels forms a continuous DCPO, the basis elements being those
kernels that drop all but fixed finite sets of input and output
packets. In \S\ref{sec:continuity} we show that ProbNetKAT's primitives are
(Scott-)continuous and its program operators preserve
continuity. Other operations such as product and Lebesgue integration
are also treated in this framework. In proving these results, we
attempt to reuse general results from domain theory whenever possible,
relying on the specific properties of $\pH$ only when necessary. We
supply complete proofs for folklore results and in cases where we
could not find an appropriate original source. We also show that the
two definitions of the Kleene star operator---one in terms of an
infinite stochastic process and one as the least fixpoint of a
Scott-continuous map---coincide. In \S\ref{sec:approx} we apply the
continuity results from \S\ref{sec:continuity} to derive monotone
convergence theorems.
In \S\ref{sec:case-study} we describe a prototype implementation based on \S\ref{sec:approx}
and several applications. In \S\ref{sec:rel-work} we review related
work. We conclude in \S\ref{sec:concl} by discussing open problems and
future directions.

%%%%%%%%%%%%%%%%%%%%%%%%%%%%%%%%%%%%%%%%%%%%%%%%%%%%%%%%%%%%%%%%%%%%%%%%%%%%%%%%
%%%%%%%%%%%%%%%%%%%%%%%%%%%%%%%%%%%%%%%%%%%%%%%%%%%%%%%%%%%%%%%%%%%%%%%%%%%%%%%%

\section{Overview}
\label{sec:overview}

This section provides motivation for the ProbNetKAT language and
summarizes our main results using a simple example.

\paragraph*{Example.}
Consider the topology shown in Figure~\ref{fig:4cycle} and suppose we
are asked to implement a routing application that forwards all traffic
to its destination while minimizing congestion, gracefully adapting to
shifts in load, and also handling unexpected failures. This problem is
known as traffic engineering in the networking literature and has been
extensively
studied~\cite{fortz02,he2008toward,b4,applegate2003making,racke2008optimal}.
Note that standard shortest-path routing (SPF) does not solve the
problem as stated---in general, it can lead to bottlenecks and also
makes the network vulnerable to failures. For example, consider
sending a large amount of traffic from host $h_1$ to host $h_3$: there
are two paths in the topology, one via switch $S_2$ and one via switch
$S_4$, but if we only use a single path we sacrifice half of the
available capacity. The most widely-deployed approaches to traffic
engineering today are based on using multiple paths and
randomization. For example, Equal Cost Multipath Routing (ECMP), which
is widely supported on commodity routers, selects a least-cost path
for each traffic flow uniformly at random. The intention is to spread
the offered load across a large set of paths, thereby reducing
congestion without increasing latency.

\paragraph*{ProbNetKAT Language.}
Using ProbNetKAT, it is straightforward to write a program that
captures the essential behavior of ECMP. We first construct programs
that model the routing tables and topology, and build a program that
models the behavior of the entire network.

\smallskip
\noindent\textit{Routing:} We model the routing tables for the 
switches using simple ProbNetKAT programs that match on destination
addresses and forward packets on the next hop toward their
destination. To randomly map packets to least-cost paths, we use the
choice operator ($\oplus$). For example, the program for switch
$\textsf{S1}$ in Figure~\ref{fig:4cycle} is as follows:
\[
\begin{array}{@{~}r@{\,}l@{~}}
p_1 & \defeq (\pseq{\match{\dst}{h_1}}{\modify{\pt}{1}})\\
  &\pcomp~ (\pseq{\match{\dst}{h_2}}{\modify{\pt}{2}})\\
  &\pcomp~(\pseq{\match{\dst}{h_3}}{(\modify{\pt}{2} \oplus \modify{\pt}{4})})\\
  &\pcomp~(\pseq{\match{\dst}{h_4}}{\modify{\pt}{4}})
\end{array}
\]
The programs for other switches are similar. To a first approximation,
this program can be read as a routing table, whose entries are
separated by the parallel composition operator ($\pcomp$). The first
entry states that packets whose destination is $h_1$ should be
forwarded out on port $1$ (which is directly connected to
$h_1$). Likewise, the second entry states that packets whose
destination is host $h_2$ should be forwarded out on port $2$, which
is the next hop on the unique shortest path to $h_2$. The third entry,
however, is different: it states that packets whose destination is
$h_3$ should be forwarded out on ports $2$ and $4$ with equal
probability. This divides traffic going to $h_3$ among the clockwise
path via $S_2$ and the counter-clockwise path via $S_4$. The final
entry states that packets whose destination is $h_4$ should be
forwarded out on port $4$, which is again the next hop on the unique
shortest path to $h_4$. The routing program for the network is the
parallel composition of the programs for each switch:
\[
p \defeq (\pseq{\match{\sw}{S_1}}{p_1}) \pcomp
(\pseq{\match{\sw}{S_2}}{p_2}) \pcomp
(\pseq{\match{\sw}{S_3}}{p_3}) \pcomp
(\pseq{\match{\sw}{S_4}}{p_4})
\]

\smallskip
\noindent\textit{Topology:} We model a directed link as a program
that matches on the switch and port at one end of the link and
modifies the switch and port to the other end of the link. We model an
undirected link $l$ as a parallel composition of directed links in
each direction. For example, the link between switches $S_1$ and $S_2$
is modeled as follows:
\[
\begin{aligned}
l_{1,2} &\triangleq
  (\pseq{{\pseq{\pseq{\match{\sw}{S_1}}{\match{\pt}{2}}}{\pdup}}} {\pseq{\pseq{\modify{\sw}{S_2}}{\modify{\pt}{1}}}{\pdup}}) \\
&\pcomp
  (\pseq{{\pseq{\pseq{\match{\sw}{S_2}}{\match{\pt}{1}}}{\pdup}}} {\pseq{\pseq{\modify{\sw}{S_1}}{\modify{\pt}{2}}}{\pdup}})
   \\
\end{aligned}
\]
Note that at each hop we use ProbNetKAT's $\pdup$ operator to store
the headers in the packet's history, which records the trajectory of
the packet as it goes through the network. Histories are useful for
tasks such as measuring path length and analyzing link congestion. We
model the topology as a parallel composition of individual links:
\[
  t ~\triangleq ~l_{1,2} ~\pcomp ~l_{2,3} ~\pcomp ~l_{3,4} ~\pcomp ~l_{1,4}
\]
To delimit the network edge, we define ingress and egress predicates:
\[
\def\arraycolsep{2pt}
\begin{array}{rcl}
\mathit{in} & \defeq & (\pseq{\match{\sw}{1}}{\match{\pt}{1}}) \pcomp
  (\pseq{\match{\sw}{2}}{\match{\pt}{2}}) \pcomp \dots\\
\mathit{out} & \defeq & (\pseq{\match{\sw}{1}}{\match{\pt}{1}}) \pcomp
  (\pseq{\match{\sw}{2}}{\match{\pt}{2}}) \pcomp \dots\\
\end{array}
\]
Here, since every ingress is an egress, the predicates are identical.

\smallskip
\noindent\textit{Network:}
We model the end-to-end behavior of the entire network by combining
$p$, $t$, $\mathit{in}$ and $\mathit{out}$ into a single program:
\[
  \mathit{net} \defeq \pseq{\mathit{in}}{\pseq{\pseq{\pstar{(\pseq{p}{t})}}{p}}}{\mathit{out}}
\]
This program models processing each input from ingress to egress
across a series of switches and links. Formally it denotes a Markov
kernel that, when supplied with an input distribution on packet
histories $\mu$ produces an output distribution $\nu$.

\smallskip
\noindent\textit{Queries:}
Having constructed a probabilistic model of the network, we can use
standard tools from measure theory to reason about performance. For
example, to compute the expected congestion on a given link $l$, we
would introduce a function $Q$ from sets of packets to
$\R \cup \{\infty\}$ (formally a random variable):
\[
Q(a) \defeq \sum_{h \in a}~\#_l(h)
\]
where $\#_l(h)$ is the function on packet histories that returns the
number of times that link $l$ occurs in $h$, and then compute the
expected value of $Q$ using integration:
\[
\ex Q \nu = \int Q\,d \nu
\]
We can compute queries that capture other aspects of network
performance such as latency, reliability, etc. in similar fashion.

\paragraph*{Limitations.}
Unfortunately there are several issues with the approach just
described:
\begin{itemize}
\item One problem is that computing the results of a query can
require complicated measure theory since a ProbNetKAT program may
generate a continuous distribution in general
(Lemma~\ref{thm:pnk-may-be-continuous}). Formally, instead of summing
over the support of the distribution, we have to use Lebesgue
integration in an appropriate measurable space. Of course, there are also
challenges in representing infinite distributions in an implementation.
\item Another issue is that the semantics of iteration is modeled
in terms of an infinite stochastic process rather than a standard
fixpoint. The original ProbNetKAT paper showed that it is possible to
approximate a program using a series of star-free programs that weakly
converge to the correct result, but the approximations need not
converge monotonically, which makes this result difficult to apply in
practice.
\item Even worse, many of the queries that we would like to answer are
not actually continuous in the Cantor topology, meaning that the weak
convergence result does not even apply! The notion of distance on sets
of packet histories is $d(a,b) = 2^{-n}$ where $n$ is the length of
the smallest history in $a$ but not in $b$, or vice versa. It is easy
to construct a sequence of histories $h_n$ of length $n$ such that
$\lim_{n\to\infty} d(\{h_n\}, \{\}) = 0$ but $\lim_{n\to\infty}
Q(\{h_n\}) = \infty$ which is not equal to $Q(\{\}) = 0$.
\end{itemize}
Together, these issues are significant impediments that make it
difficult to apply ProbNetKAT in many scenarios.

\begin{figure}[t!]
\tikzset{
  switch/.style={
    circle,
    fill=black!25,
    draw=black!50,
    very thick,
    minimum size=20pt,
    inner sep=0pt},
  host/.style={
    rectangle,
    fill=black!25,
    draw=black!50,
    very thick,
    minimum size=14pt,
    inner sep=0pt},
}
\smallskip
\noindent
\centerline{\begin{tikzpicture}
    \node[switch] (S1) at (0,0) {\textsf{S1}};
    \node[switch] (S2) at (1.5,0)  {\textsf{S2}};
    \node[switch] (S3) at (1.5,-1.5) {\textsf{S3}};
    \node[switch] (S4) at (0,-1.5) {\textsf{S4}};
    \node[host] (H1) at (-1,.25) {\textsf{h1}};
    \node[host] (H2) at (2.5,.25) {\textsf{h2}};
    \node[host] (H3) at (2.5,-1.75) {\textsf{h3}};
    \node[host] (H4) at (-1,-1.75) {\textsf{h4}};
    \path[draw]
    (S1) edge node [above,pos=0.1] {\small \sf 1} (H1)
    (S2) edge node [above,pos=0.1] {\small \sf 2} (H2)
    (S3) edge node [below,pos=0.1] {\small \sf 3} (H3)
    (S4) edge node [below,pos=0.1] {\small \sf 4} (H4)
    (S1) edge node [above,pos=0.1] {\small \sf 2} node [above,pos=0.9] {\small \sf 1} (S2)
    (S2) edge node [right,pos=0.1] {\small \sf 3} node [right,pos=0.9] {\small \sf 2} (S3)
    (S3) edge node [below,pos=0.1] {\small \sf 4} node [below,pos=0.9] {\small \sf 3} (S4)
    (S4) edge node [left ,pos=0.1] {\small \sf 1} node [left ,pos=0.9] {\small \sf 4} (S1);
\end{tikzpicture}}
\centerline{\small (a)}
\\[1ex]
\begin{minipage}{.53\columnwidth}
  \includegraphics[width=\columnwidth,clip,trim={0 22px 0 0}]{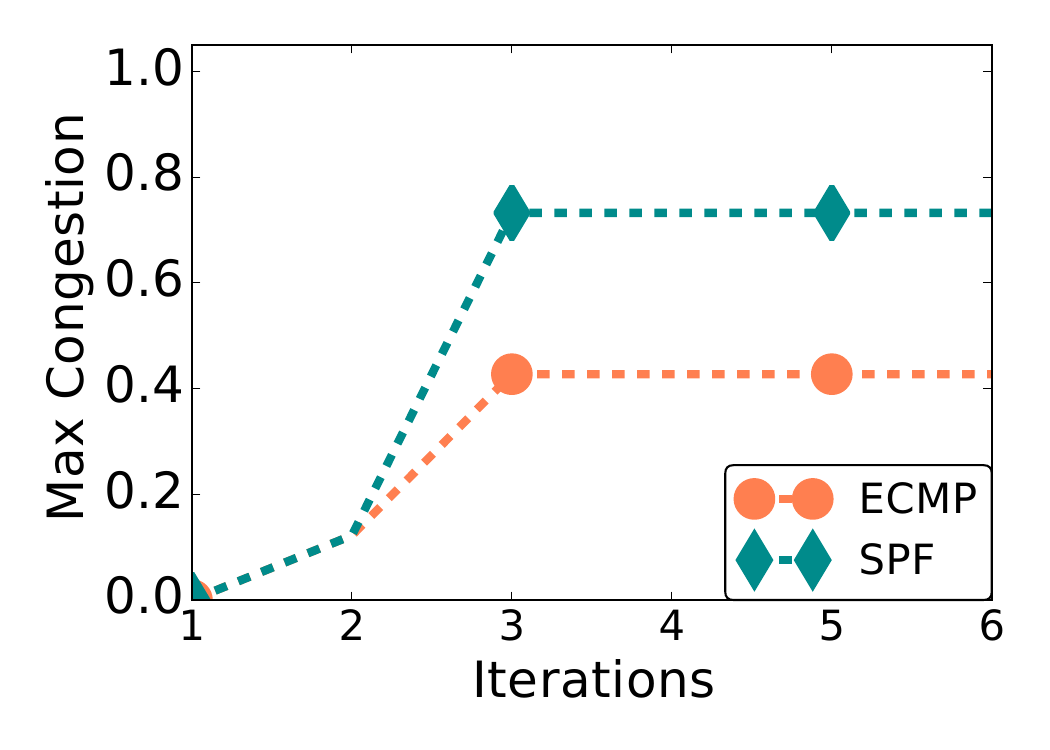}
  \centerline{~~\quad\small(b)}
\end{minipage}\hfill
\begin{minipage}{.43\columnwidth}
  \includegraphics[width=\columnwidth]{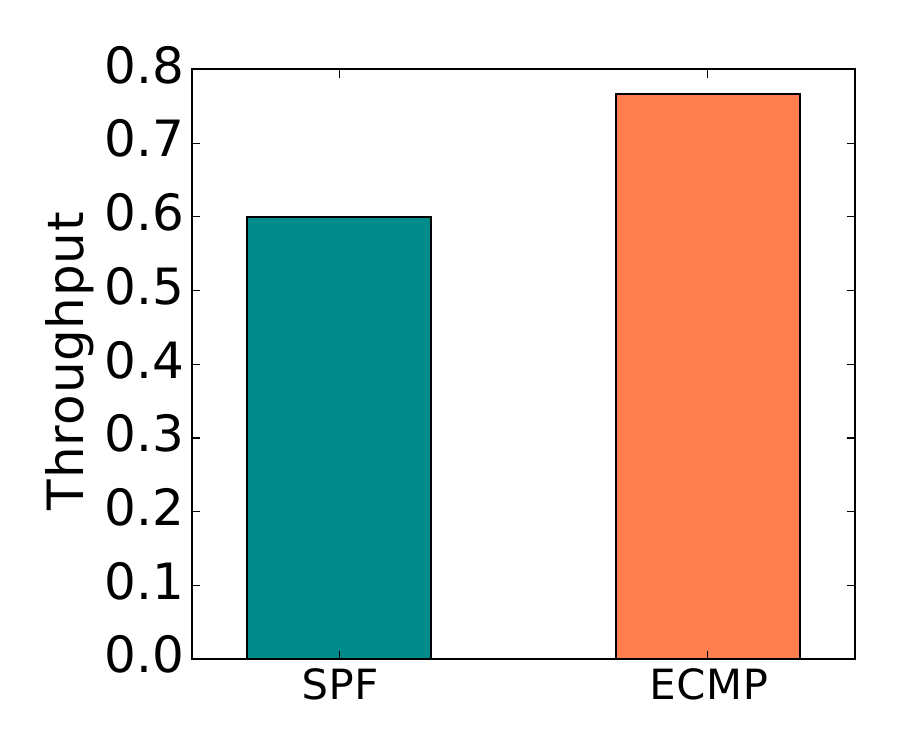}
  \centerline{\qquad\small(c)}
\end{minipage}
\caption{(a) topology, (b) congestion, (c) failure throughput.}
\label{fig:4cycle}
\end{figure}

\paragraph*{Domain-Theoretic Semantics.}
This paper develops a new semantics for ProbNetKAT that overcomes
these problems and provides the key building blocks needed to engineer
a practical implementation. The main insight is that we can formulate
the semantics in terms of the Scott topology rather than the Cantor
topology. It turns out that these two topologies generate the same
Borel sets, and the relationship between them can be characterized
using an extension theorem that captures when functions on the basic
Scott-open sets extend to a measure. We show how to construct a DCPO
equipped with a natural partial order that also lifts to a partial
order on Markov kernels. We prove that standard program operators are
continuous, which allows us to formulate the semantics of the
language---in particular Kleene star---using standard tools from
domain theory, such as least fixpoints. Finally, we formalize a
notion of approximation and prove a monotone convergence theorem.

The problems with the original ProbNetKAT semantics identified above
are all solved using the new semantics. Because the new semantics
models iteration as a least fixpoint, we can work with finite
distributions and star-free approximations that are guaranteed to monotonically
converge to the analytical solution (Corollary~\ref{thm:expectation-approx}).
Moreover, whereas our query $Q$ was
not Cantor continuous, it is straightforward to show that it is Scott
continuous. Let $A$ be an increasing chain $a_0 \subseteq
a_1 \subseteq a_2 \subseteq \dots$ ordered by inclusion. Scott
continuity requires
\(
\bigsqcup_{a \in A} Q(a) = Q(\bigsqcup A\big)
\)
which is easy to prove. Hence, the convergence theorem applies and we
can compute a monotonically increasing chain of approximations that
converge to $\ex Q \nu$.

\paragraph*{Implementation and Applications.}
We developed the first implementation of ProbNetKAT using the new
semantics. We built an interpreter for the language and implemented a
variety of traffic engineering schemes including ECMP, $k$-shortest
path routing, and oblivious routing~\cite{racke2008optimal}. We
analyzed the performance of each scheme in terms of congestion and
latency on real-world demands drawn from Internet2's Abilene backbone,
and in the presence of link failures. We showed how to use the
language to reason probabilistically about reachability properties
such as loops and black holes. Figures~\ref{fig:4cycle}~(b-c) depict
the expected throughput and maximum congestion when using shortest
paths (SPF) and ECMP on the 4-node topology as computed by our
ProbNetKAT implementation. We set the demand from $h_1$ to $h_3$ to be
$\frac{1}{2}$ units of traffic, and the demand between all other pairs
of hosts to be $\frac{1}{8}$ units. The first graph depicts the
maximum congestion induced under successive approximations of the
Kleene star, and shows that ECMP achieves much better congestion than
SPF. With SPF, the most congested link (from $S_1$ to $S_2$) carries
traffic from $h_1$ to $h_2$, from $h_4$ to $h_2$, and from $h_1$ to
$h_3$, resulting in $\frac{3}{4}$ total traffic. With ECMP, the same
link carries traffic from $h_1$ to $h_2$, half of the traffic from
$h_2$ to $h_4$, half of the traffic from $h_1$ to $h_3$, resulting in
$\frac{7}{16}$ total traffic. The second graph depicts the loss of
throughput when the same link fails. The total aggregate demand is
$1 \frac{7}{8}$. With SPF, $\frac{3}{4}$ units of traffic are dropped
leaving $1 \frac{1}{8}$ units, which is 60\% of the demand, whereas
with ECMP only $\frac{7}{16}$ units of traffic are dropped leaving
$1 \frac{7}{16}$ units, which is 77\% of the demand.

%%%%%%%%%%%%%%%%%%%%%%%%%%%%%%%%%%%%%%%%%%%%%%%%%%%%%%%%%%%%%%%%%%%%%%%%%%%%%%%%
%%%%%%%%%%%%%%%%%%%%%%%%%%%%%%%%%%%%%%%%%%%%%%%%%%%%%%%%%%%%%%%%%%%%%%%%%%%%%%%%
\section{Preliminaries}
\label{sec:primer}
\newcommand{\salg}{\FF}

This section briefly reviews basic concepts from topology, measure
theory, and domain theory, and defines \emph{Markov kernels}, the
objects on which \probnetkat's semantics is based. For a more detailed
account, the reader is invited to consult standard
texts~\cite{durrett2010probability,Abramsky94domaintheory}.

\paragraph*{Topology.}%
A \emph{topology} $\SO \subseteq \pset X$ on a set $X$ is a collection
of subsets including $X$ and $\emptyset$ that is closed under finite
intersection and arbitrary union. A pair $(X,\SO)$ is called a
\emph{topological space} and the sets $U,V \in \SO$ are called
the \emph{open sets} of $(X,\SO)$. A function $f:X\to Y$ between
topological spaces $(X,\SO_X)$ and $(Y,\SO_Y)$ is \emph{continuous} if
the preimage of any open set in $Y$ is open in $X$, \ie if
\[
\inv f (U) = \set{x\in X}{f(x) \in U} \in \SO_X
\]
for any $U \in \SO_Y$.

\paragraph*{Measure Theory.}
A \emph{$\sigma$-algebra} $\salg \subseteq \pset X$ on a set $X$ is a
collection of subsets including $X$ that is closed under complement,
countable union, and countable intersection. A \emph{measurable space}
is a pair $(X,\salg)$. A \emph{probability measure} $\mu$ over such a
space is a function $\mu : \salg \to [0,1]$ that assigns probabilities
$\mu(A) \in [0,1]$ to the
\emph{measurable sets} $A \in \salg$, and satisfies the following conditions:
\begin{itemize}
\item $\mu(X) = 1$
\item $\mu(\bigcup_{i \in I} A_i) = \sum_{i \in I} \mu(A_i)$ whenever $\{A_i\}_{i \in I}$
is a countable\\[.5ex] collection of disjoint measurable sets. 
\end{itemize}
Note that these conditions already imply that $\mu(\emptyset)=0$. Elements
$a,b \in X$ are called \emph{points} or \emph{outcomes}, and measurable sets
$A,B \in \salg$ are also called \emph{events}.
The $\sigma$-algebra $\sigma(U)$ generated by a set $U\subseteq X$ is the smallest
$\sigma$-algebra containing $U$:
\[
\sigma(U) \defeq \bigcap \set{\salg \subseteq \pset X}{\salg \text{ is a $\sigma$-algebra and } U \subseteq \salg}.
\]
Note that it is well-defined because the intersection is not empty
($\pset X$ is trivially a $\sigma$-algebra containing $U$) and intersections
of $\sigma$-algebras are again $\sigma$-algebras. If
$\mathcal{O} \subs \pset X$ are the open sets of $X$, then the
smallest $\sigma$-algebra containing the open sets $\BB
= \sigma(\mathcal{O})$ is the \emph{Borel algebra}, and the measurable
sets $A,B \in \BB$ are the \emph{Borel sets} of $X$.

Let $P_\mu \defeq \set{ a \in X}{ \mu(\sset{a}) > 0 }$ denote the
points (not events!) with non-zero probability. It can be shown that
$P_\mu$ is countable. A probability measure is called \emph{discrete}
if $\mu(P_\mu) = 1$. Such a measure can simply be represented by a
function $Pr : X \to [0,1]$ with $Pr(a) = \mu(\sset{a})$. If $|P_\mu|
< \infty$, the measure is called \emph{finite} and can be represented
by a finite map $Pr : P_\mu \to [0,1]$.  In contrast, measures for
which $\mu(P_\mu) = 0$ are called \emph{continuous}, and measures for
which $0 < \mu(P_\mu) < 1$ are called \emph{mixed}.  The \emph{Dirac
measure} or \emph{point mass} puts all probability on a single point
$a \in X$: $\dirac a (A) = 1$ if $a \in A$ and $0$ otherwise.
%% \begin{align*}
%%   \dirac a (A) \defeq \begin{cases}
%%     1 &\text{if } a \in A\\
%%     0 &\text{otherwise}
%%   \end{cases}
%% \end{align*}
The uniform distribution on $[0,1]$ is a continuous measure.

A function $f : X \to Y$ between measurable spaces $(X, \salg_X)$ and
$(Y, \salg_Y)$ is called \emph{measurable} if the preimage of any
measurable set in $Y$ is measurable in $X$, \ie if \[ \inv f
(A) \defeq \set{x\in X}{f(x) \in A} \in \salg_X
\]
for all $A \in \salg_Y$. If $Y=\R \cup \sset{-\infty,+\infty}$, then $f$ is
called a \emph{random variable} and its \emph{expected value} with respect to a
measure $\mu$ on $X$ is given by the Lebesgue integral
\begin{align*}
  \ex{f}{\mu} &\defeq \int f d\mu = \int_{x \in X} f(x) \cdot \mu(dx)
\end{align*}
If $\mu$ is discrete, the integral simplifies to the sum
\begin{align*}
  \ex{f}{\mu} &= \sum_{x \in X} f(x) \cdot \mu(\sset{x})
  = \sum_{x \in P_\mu} f(x) \cdot Pr(x)
\end{align*}

\paragraph*{Markov Kernels.}%
Consider a probabilistic transition system with states $X$ that makes a
random transition between states at each step. If $X$ is finite, the
system can be captured by a transition matrix $T \in [0,1]^{X\times
X}$, where the matrix entry $T_{xy}$ gives the probability that the
system transitions from state $x$ to state $y$. Each row $T_x$
describes the transition function of a state $x$ and must sum to
$1$. Suppose that the start state is initially distributed according
to the row vector $V \in [0,1]^X$, \ie the system starts in state
$x\in X$ with probability $V_x$.  Then, the state distribution is
given by the matrix product $VT \in [0,1]^{X}$ after one step and by
$VT^n$ after $n$ steps.

Markov kernels generalize this idea to infinite state systems. Given
measurable spaces $(X,\salg_X)$ and $(Y,\salg_Y)$, a Markov kernel
with source $X$ and target $Y$ is a function $P : X \times \salg_Y \to
[0,1]$ (or equivalently, $X \to \salg_Y \to [0,1]$) that maps each
source state $x\in X$ to a distribution over target states $P(x,-)
: \salg_Y \to [0,1]$. If the initial distribution is given by a
measure $\nu$ on $X$, then the target distribution $\mu$ after one
step is given by Lebesgue integration:
\begin{align}
\label{eq:kernl-dist-comp}
  \mu(A) &\defeq \int_{x \in X} P(x,A) \cdot \nu(dx) && (A \in \salg_Y)
\end{align}
If $\nu$ and $P(x,-)$ are discrete, the integral simplifies to the sum
\begin{align*}
  \mu(\sset{y}) &= \sum_{x \in X} P(x,\sset{y}) \cdot \nu({\sset{x}})
  &&(y \in Y)
\end{align*}
which is just the familiar vector-matrix-product $VT$. Similarly, two
kernels $P,Q$ from $X$ to $Y$ and from $Y$ to $Z$, respectively, can
be sequentially composed to a kernel $P \cmp Q$ from $X$ to $Z$:
\begin{align}
\label{eq:kernl-kernel-comp}
  (P \cmp Q)(x,A) &\defeq \int_{y \in Y} P(x, dy) \cdot Q(y, A)
\end{align}
This is the continuous analog of the matrix product $TT$. A Markov
kernel $P$ must satisfy two conditions:
\begin{enumerate}[(i)]
\item For each source state $x \in X$, the map $A \mapsto P(x,A)$
  must be a probability measure on the target space.
\item \label{kernel-measurable} For each event $A \in \salg_Y$
  in the target space, the map $x \mapsto P(x,A)$ must be a measurable
  function.
\end{enumerate}
Condition (\ref{kernel-measurable}) is required to ensure that
integration is well-defined.
A kernel $P$ is called \emph{deterministic} if $P(a,-)$ is a dirac measure
for each $a$.

\paragraph*{Domain Theory.}
A \emph{partial order} (PO) is a pair $(D,\sqleq)$ where $D$ is a set
and $\sqleq$ is a reflexive, transitive, and antisymmetric relation on
$D$.  For two elements $x,y \in D$ we let $x \sqcup y$ denote their
$\sqleq$-least upper bound (\ie, their supremum), provided it
exists. Analogously, the least upper bound of a subset $C\subs D$ is
denoted $\bigsqcup C$, provided it exists.  A non-empty subset $C\subs D$
is \emph{directed} if for any two $x,y\in C$ there exists
\emph{some} upper bound $x,y \sqleq z$ in $C$.
A \emph{directed complete partial order} (DCPO) is a PO for which any
directed subset $C\subseteq D$ has a supremum $\bigsqcup C$ in $D$.
If a PO has a least element it is denoted by $\bot$, and if it has a
greatest element it is denoted by $\top$. For example, the nonnegative
real numbers with infinity $\R_+ \defeq [0,
\infty]$ form a DCPO under the natural order $\leq$ with suprema $\bigsqcup C =
\sup C$, least element $\bot = 0$, and greatest element $\top = \infty$.
The unit interval is a DCPO under the same order, but with $\top =
1$. Any powerset $\pset X$ is a DCPO under the subset order, with
suprema given by union.

A function $f$ from $D$ to $E$ is
called \emph{{\upshape (}Scott-{\upshape)}continuous} if
\begin{enumerate}[(i)]
  \item it is monotone, \ie $x \sqleq y$ implies $f(x) \sqleq f(y)$,
  and
\item it preserves suprema, \ie $f(\bigsqcup C)
  = \bigsqcup_{x\in C} f(x)$ for any directed set $C$ in $D$.
\label{itm:preserve-suprema}
\end{enumerate}
Equivalently, $f$ is continuous with respect to the \emph {Scott
topologies} on $D$ and $E$ \cite[Proposition
2.3.4]{Abramsky94domaintheory}, which we define next.
(Note how condition \eqref{itm:preserve-suprema} looks like the classical
definition of continuity of a function $f$, but with suprema taking the role of
limits).
The set of all
continuous functions $f : D \to E$ is denoted $[D \to E]$.

A subset $A\subs D$ is called \emph{up-closed} (or an \emph{upper
set}) if $a\in A$ and $a\sqleq b$ implies $b \in A$. The smallest
up-closed superset of $A$ is called its \emph{up-closure} and is
denoted $\up A$.  $A$ is called (\emph{Scott-})\emph{open} if it is
up-closed and intersects every directed subset $C\subs D$ that
satisfies $\bigsqcup C\in A$.  For example, the Scott-open sets of
$\R_+$ are the upper semi-infinite intervals $(r,\infty]$,
$r\in\R_+$. The Scott-open sets form a topology on $D$ called
the \emph{Scott topology}.

DCPOs enjoy many useful closure properties:
\begin{enumerate}[(i)]
\item The cartesian product of any collection of DCPOs is a DCPO with
  componentwise order and suprema.
\item If $E$ is a DCPO and $D$ any set, the function space $D \to E$ is a
  DCPO with pointwise order and suprema.
  %(In fact, this is a special %case of (i)).
\item The continuous functions $[D \to E]$ between DCPOs $D$ and $E$ form a
  DCPO with pointwise order and suprema.
\end{enumerate}

If $D$ is a DCPO with least element $\bot$, then any Scott-continuous
self-map $f \in [D \to D]$ has a $\sqleq$-least fixpoint, and it is
given by the supremum of the chain $\bot \sqleq f(\bot) \sqleq
f(f(\bot)) \sqleq \dots$:
\[
\lfp(f) = \bigsqcup_{n \geq 0} f^n(\bot)
\]
Moreover, the least fixpoint operator, $\lfp \in [[D \to D] \to D]$
is itself continuous, that is:
\(
\lfp(\bigsqcup C) = \bigsqcup_{f \in C} \lfp(f)
\),
for any directed set of functions $C \subseteq [D \to D]$.

An element $a$ of a DCPO is called \emph{finite} (\citet{Abramsky94domaintheory}
use the term \emph{compact}~) if for any
directed set $A$, if $a\sqleq\tbigsqcup A$, then there exists $b\in A$
such that $a\sqleq b$. Equivalently, $a$ is finite if its up-closure
$\up{\{a\}}$ is Scott-open. A DCPO is called \emph{algebraic} if for
every element $b$, the finite elements $\sqleq$-below $b$ form a
directed set and $b$ is the supremum of this set.  An element $a$ of a
DCPO \emph{approximates} another element $b$, written $a\ll b$, if for
any directed set $A$, $a\sqleq c$ for some $c\in A$ whenever
$b\sqleq\bigsqcup A$.  A DCPO is called \emph{continuous} if for every
element $b$, the elements $\ll$-below $b$ form a directed set and $b$
is the supremum of this set.  Every algebraic DCPO is continuous.  A
set in a topological space is \emph{compact-open} if it is compact
(every open cover has a finite subcover) and open.

Here we recall some basic facts about DCPOs. These are all well-known,
but we state them as a lemma for future reference.
\begin{lemma}[DCPO Basic Facts]
\label{lem:basic}\ 
\begin{enumerate}[{\upshape(i)}]

\item
Let $E$ be a DCPO and $D_1,D_2$ sets. There is a homeomorphism
(bicontinuous bijection) $\curry$
%% \begin{align*}
%% \curry : (D_1\times D_2 \to E)\to(D_1\to D_2\to E)
%% \end{align*}
between the DCPOs $D_1\times D_2\to E$ and $D_1\to D_2\to E$, where
the function spaces are ordered pointwise. The inverse of $\curry$ is
$\uncurry$.

\item\label{lem:basic:finitebase}
In an algebraic DCPO, the open sets $\up{\{a\}}$ for finite $a$ form a
base for the Scott topology.

\item A subset of an algebraic DCPO is compact-open iff it is a finite
union of basic open sets $\up{\{a\}}$.
\end{enumerate}
\end{lemma}

%%%%%%%%%%%%%%%%%%%%%%%%%%%%%%%%%%%%%%%%%%%%%%%%%%%%%%%%%%%%%%%%%%%%%%%%%%%%%%%%
%%%%%%%%%%%%%%%%%%%%%%%%%%%%%%%%%%%%%%%%%%%%%%%%%%%%%%%%%%%%%%%%%%%%%%%%%%%%%%%%
\begin{figure*}[t!]
\begin{minipage}{.475\textwidth}
\textbf{Syntax}
\[
\begin{array}{r@{~~~}r@{~}c@{~}l@{\qquad}l}
\textrm{Naturals} & n & ::=  & \mathrlap{0 \mid 1 \mid 2 \mid \ldots}\\
\textrm{Fields} & \field & ::=  & \mathrlap{\field_1 \mid \ldots \mid \field_k} \\
\textrm{Packets} & \Pk \ni \pk & ::= & \mathrlap{\sset{\field_1=n_1, \dots , \field_k = n_k}} \\
\textrm{Histories} & \Hist \ni \h & ::= & \mathrlap{\hcons{\pk}{\hbar}}\\
                   & \hbar & ::= & \mathrlap{\hempty \mid \hcons{\pk}{\hbar}} \\
\textrm{Probabilities} & [0,1] \ni r\\
\textrm{Predicates} & \preds &
   ::= & \pfalse                       & \textit{False/Drop} \\
    & & \mid & \ptrue                  & \textit{True/Skip} \\
    & & \mid & \match{\field}{n}       & \textit{Test} \\
    & & \mid & \punion{\preda}{\predb} & \textit{Disjunction} \\
    & & \mid & \pseq{\preda}{\predb}   & \textit{Conjunction} \\
    & & \mid & \pnot{\preda}           & \textit{Negation} \\
\textrm{Programs} & \pols &
  ::= & \preda                         & \textit{Filter} \\
    & & \mid & \modify{\field}{n}      & \textit{Modification} \\
    & & \mid & \pdup                   & \textit{Duplication} \\
    & & \mid & \punion{\polp}{\polq}   & \textit{Parallel Composition} \\
    & & \mid & \pseq{\polp}{\polq}     & \textit{Sequential Composition} \\
    & & \mid & \polp \opr \polq        & \textit{Choice} \\
    & & \mid & \pstar{\polp}           & \textit{Iteration}
\end{array}\]
\end{minipage}\hfill\vrule\hfill\begin{minipage}{.475\textwidth}
\textbf{Semantics}\quad\(\den{\polp} \in \pset{\Hist} \to \Mon(\pset{\Hist})\)
\[
\def\arraystretch{1.2}
\begin{array}{r@{~~}c@{~~}l}
\den{\pfalse}(a) & \defeq &
  \unit{\emptyset}\\
\den{\ptrue}(a) & \defeq &
  \unit{a}\\
\den{\match{\field}{n}}(a) & \defeq &
  \unit{\set{\hcons{\pk}{\hbar} \in a}{\pk.f = n}} \\
\den{\pnot{\preda}}(a) & \defeq &
  \den{\preda}(a) \bind \lambda b.
  \unit{a-b}\\
\den{\modify{\field}{n}}(a) & \defeq & 
  \unit{\set{\hcons{\upd{\pk}{\field}{n}}{\hbar}}{\hcons{\pk}{\hbar} \in a }} \\
\den{\pdup}(a) & \defeq & 
  \unit{\set{\hcons{\pk}{\hcons{\pk}{\hbar}}}{\hcons{\pk}{\hbar} \in a }} \\
\den{\punion{\polp}{\polq}}(a) & \defeq & 
  \den{\polp}(a) \bind \lambda b_1.
  \den{\polq}(a) \bind \lambda b_2.
  \unit{b_1 \cup b_2}\\
\den{\pseq{\polp}{\polq}}(a) & \defeq &
  \den{\polp}(a) \bind \den{\polq}\\
\den{\polp \opr \polq}(a) & \defeq &
  r \cdot \den{\polp}(a) + (1-r) \cdot \den{\polq}(a)\\
\den{\pstar\polp}(a) & \defeq & \displaystyle\bigsqcup_{n \in \N} \den{p^{(n)}}(a)\\
\multicolumn{3}{l}{\text{where } ~ p^{(0)} \defeq \ptrue ~ \text{ and } ~
  p^{(n+1)} \defeq \punion{\ptrue}{\pseq{p}{p^{(n)}}}}
\end{array}
\]
%% \hrule\vspace{1ex}
%
%% \textbf{Identity Monad}
%% \[
%% \def\arraystretch{1.25}
%% \begin{array}{l}
%% \Mon(X) \defeq X \qquad
%% \unit{a} \defeq a \qquad
%% a \bind f \defeq f(a)
%% \end{array}
%% \]
\hrule\vspace{1ex}
\textbf{Probability Monad}
\[
\def\arraystretch{1.25}
\begin{array}{l}
\Mon(X) \defeq \set{\mu : \BB \to [0,1]}{\mu \text{ is a probability measure}}
\\
\unit{a} \defeq \dirac a \qquad
\displaystyle \mu \bind P \defeq \lambda A. \int_{a \in X} P(a)(A) \cdot \mu(da) 
\end{array}
\]
\end{minipage}
\caption{\probnetkat: syntax and semantics.}
\label{fig:probnetkat}
\end{figure*}

\section{ProbNetKAT}\label{sec:syntax}
This section defines the syntax and semantics of \PNK\ formally (see
Figure~\ref{fig:probnetkat}) and establishes some basic
properties. \probnetkat is a core calculus designed to capture the
essential forwarding behavior of probabilistic network programs. In
particular, the language includes primitives that model fundamental
constructs such as parallel and sequential composition, iteration, and
random choice. It does not model features such as mutable state,
asynchrony, and dynamic updates, although extensions to \netkat-like
languages with several of these features have been studied in previous
work~\cite{reitblatt12,event-driven-pldi16}.

\paragraph*{Syntax.}%
A packet $\pk$ is a record mapping a finite set of fields
$\field_1, \field_2, \dots, \field_k$ to bounded integers $n$. Fields
include standard header fields such as the source (\src) and
destination (\dst) of the packet, and two logical fields (\sw for
switch and \pt for port) that record the current location of the
packet in the network. The logical fields are not present in a
physical network packet, but it is convenient to model them as proper
header fields. We write $\pk.\field$ to denote the value of field
$\field$ of $\pi$ and $\upd{\pk}{\field}{n}$ for the packet obtained
from $\pi$ by updating field $\field$ to $n$. We let $\Pk$ denote the
(finite) set of all packets.
 
A history $\h=\hcons{\pk}{\hbar}$ is a non-empty list of packets with
\emph{head packet} $\pi$ and (possibly empty) \emph{tail} $\hbar$.
The head packet models the packet's current state and the tail
contains its prior states, which capture the trajectory of the packet
through the network. Operationally, only the head packet exists, but
it is useful to discriminate between identical packets with different
histories. We write $\Hist$ to denote the (countable) set of all
histories.

We differentiate between \emph{predicates} ($\preda,\predb$)
and \emph{programs} ($\polp,\polq$).  The predicates form a Boolean
algebra and include the primitives
\emph{false} ($\pfalse$), \emph{true} ($\ptrue$), and \emph{tests} ($\match{\field}{n}$),
as well as the standard Boolean operators disjunction
($\punion{\preda}{\predb}$), conjunction ($\pseq{\preda}{\predb}$),
and negation ($\pnot{\preda}$). Programs include \emph{predicates}
($\preda$) and \emph{modifications} ($\modify{\field}{n}$) as
primitives, and the operators \emph{parallel composition}
($\punion{\polp}{\polq}$), \emph{sequential composition}
($\pseq{\polp}{\polq}$), and \emph{iteration} ($\pstar{\polp}$). The
primitive $\pdup$ records the current state of the packet by extending
the tail with the head packet. Intuitively, we may think of a history as a
log of a packet's activity, and of $\pdup$ as the logging command.
Finally, \emph{choice}
$\polp \opr \polq$ executes $\polp$ with probability $r$ or $\polq$
with probability $1-r$. We write $p \oplus q$ when $r=0.5$.

Predicate conjunction and sequential composition use the same syntax
($\pseq{\preda}{\predb}$) as their semantics coincide (as we will see
shortly). The same is true for disjunction of predicates and parallel
composition ($\punion{\preda}{\predb}$).  The distinction
between \emph{predicates} and \emph{programs} is merely to restrict
negation to predicates and rule out programs like
$\pnot{(\pstar\polp)}$.

\paragraph{Syntactic Sugar.}%
The language as presented in Figure~\ref{fig:probnetkat} is reduced to its core
primitives. It is worth noting that many useful constructs can be derived from
this core. In particular, it is straightforward to encode conditionals and while
loops:
\begin{align*}
\ite{\preda}{\polp}{\polq}&~\defeq~
  \punion{\pseq{\preda}{\polp}}{\pseq{\pnot \preda}{\polq}} \\
\while{\preda}{\polp} &~\defeq~
  \pseq{(\pseq{\preda}{\polp})\star}{\pnot \preda}
\end{align*}
These encodings are well-known from \kat \cite{K97c}. While loops are useful for implementing
higher level abstractions such as network virtualization in \netkat \cite{compilekat}.

\paragraph*{Example.}%
Consider the programs
\begin{align*}
  p_1 &\defeq \match{\pt}{1} \cmp (\modify{\pt}{2} \pcomp \modify{\pt}{3}) \\
  p_2 &\defeq (\match{\pt}{2} \pcomp \match{\pt}{3}) 
    \cmp \modify{\dst}{10.0.0.1}
    \cmp \modify{\pt}{1}
\end{align*}
The first program forwards packets entering at port $1$ out of ports
$2$ and $3$---a simple form of multicast---and drops all other
packets. The second program matches on packets coming in on ports
$2$ \emph{or} $3$,
%
%JNF: we said this above!
%\footnote{Note that $\pcomp$ acts like boolean \emph{disjunction} on predicates!},
%
modifies their destination to the IP address $10.0.0.1$, and sends
them out through port $1$.  The program $\punion{p_1}{p_2}$ acts like
$p_1$ for packets entering at port $1$, and like $p_2$ for packets
entering at ports $2$ or $3$.

\paragraph*{Monads.}%
We define the semantics of \netkat programs parametrically over a
monad $\Mon$.  This allows us to give two concrete semantics at once:
the classical deterministic semantics (using the identity monad), and
the new probabilistic semantics (using the probability monad).  For
simplicity, we refrain from giving a categorical treatment and simply
model a monad in terms of three components:
\begin{itemize}
\item a constructor $\Mon$ that lifts $X$ to a domain $\Mon(X)$;
\item an operator $\unitop : X \to \Mon(X)$ that lifts objects into the domain $\Mon(X)$; and
\item an infix operator \[
  \bind : \Mon(X) \to (X \to \Mon(X)) \to \Mon(X)
\]
that lifts a function $f : X \to \Mon(X)$ to a function
\[(- \bind f) : \Mon(X) \to \Mon(X)\]
\end{itemize}
These components must satisfy three axioms:
\begin{align}
  \unit a \bind f &~=~ f(a) \tag{\bf M1}\label{eq:monad-left-id}\\
  m \bind \unitop &~=~ m    \tag{\bf M2}\label{eq:monad-right-id}\\
  (m \bind f) \bind g &~=~ m \bind (\lambda x. f(x) \bind g)
    \tag{\bf M3}\label{eq:monad-assoc}
\end{align}
The semantics of deterministic programs (not containing probabilistic
choices $p \opr q$) uses as underlying objects the set of packet
histories $\pH$ and the identity monad $\Mon(X)=X$: $\unitop$ is the
identify function and $x \bind f$ is simply function application
$f(x)$.  The identity monad trivially satisfies the three axioms.

The semantics of probabilistic programs uses the probability (or Giry)
monad \cite{giry1982categorical,JonesPlotkin89,ramsey2002stochastic}
that maps a measurable space to the domain of probability measures
over that space. The operator $\unitop $ maps $a$ to the point mass
(or Dirac measure) $\dirac a$ on $a$. Composition $\mu \bind (\lambda
a. \nu_a)$ can be thought of as a two-stage probabilistic experiment
where the second experiment $\nu_a$ depends on the outcome $a$ of the
first experiment $\mu$. Operationally, we first sample from $\mu$ to
obtain a random outcome $a$; then, we sample from $\nu_a$ to obtain
the final outcome $b$. What is the distribution over final outcomes?
It can be obtained by observing that $\lambda a. \nu_a$ is a Markov
kernel (\S\ref{sec:primer}), and so composition with $\mu$ is given by
the familiar integral \[ \mu \bind (\lambda a.\nu_a) = \lambda
A. \int_{a \in X} \nu_a(A) \cdot \mu(da)
\]
introduced in \eqref{eq:kernl-dist-comp}.  It is well known that these
definitions satisfy the monad axioms
\cite{K81c,giry1982categorical,JonesPlotkin89}.
\eqref{eq:monad-left-id} and \eqref{eq:monad-right-id} are trivial properties of
the Lebesgue Integral.
\eqref{eq:monad-assoc} is essentially Fubini's
theorem, which permits changing the order of integration in a double
integral.

\paragraph*{Deterministic Semantics.}
In deterministic \netkat (without $\polp \opr \polq$), a program $p$
denotes a function $\den{p} \in \pH \to \pH$ mapping a set of input
histories $a \in \pH$ to a set of output histories $\den{p}(a)$. Note
that the input and output sets do \emph{not} encode non-determinism
but represent sets of ``in-flight'' packets in the network. Histories
record the processing done to each packet as it traverses the
network. In particular, histories enable reasoning about path
properties and determining which outputs were generated from common
inputs.

Formally, a predicate $t$ maps the input set $a$ to the subset
$b \subseteq a$ of histories satisfying the predicate. In particular,
the false primitive $\pfalse$ denotes the function mapping any input
to the empty set; the true primitive $\ptrue$ is the identity
function; the test $\match{\field}{n}$ retains those histories with
field $f$ of the head packet equal to $n$; and negation $\pnot \preda$
returns only those histories not satisfying $\preda$.  Modification
$\modify{\field}{n}$ sets the \field-field of all head-packets to the
value $n$. Duplication $\pdup$ extends the tails of all input
histories with their head packets, thus permanently recording the
current state of the packets.

Parallel composition $\punion{\polp}{\polq}$ feeds the input to both
$\polp$ and $\polq$ and takes the union of their outputs. If $\polp$
and $\polq$ are predicates, a history is thus in the output iff it
satisfies at least one of $\polp$ or $\polq$, so that union acts like
logical disjunction on predicates.  Sequential composition
$\pseq{\polp}{\polq}$ feeds the input to $\polp$ and then feeds
$\polp$'s output to $\polq$ to produce the final result.  If $\polp$
and $\polq$ are predicates, a history is thus in the output iff it
satisfies both $\polp$ and $\polq$, acting like logical
conjunction. Iteration $\pstar\polp$ behaves like the parallel
composition of $\polp$ sequentially composed with itself zero or more
times (because $\bigsqcup$ is union in $\pH$).

\paragraph{Probabilistic Semantics.}%
The semantics of \probnetkat is given using the probability monad
applied to the set of history sets $\pH$ (seen as a measurable space). A
program $\polp$ denotes a function
\[
\den{p} \in \pH \to \set{\mu : \BB \to [0,1]}{\mu \text{ is a probability measure}}
\]
mapping a set of input histories $a$ to a \emph{distribution} over
output sets $\den{\polp}(a)$. Here, $\BB$ denotes the Borel sets of
$\pH$ (\S\ref{sec:cantor}).  Equivalently, $\den{p}$ is a Markov
kernel with source and destination $(\pH, \BB)$.  The semantics of all
primitive programs is identical to the deterministic case, except that
they now return point masses on output sets (rather than just output
sets). In fact, it follows from \eqref{eq:monad-left-id} that all
programs without choices and iteration are point masses.

Parallel composition $\polp \pcomp \polq$ feeds the input $a$ to
$\polp$ and $\polq$, samples $b_1$ and $b_2$ from the output
distributions $\den{\polp}(a)$ and $\den{\polq}(a)$, and returns the
union of the samples $b_1 \cup b_2$.  Probabilistic choice
$\polp \opr \polq$ feeds the input to both $p$ and $q$ and returns a
convex combination of the output distributions according to $r$.
Sequential composition $p \cmp q$ is just sequential composition of
Markov kernels. Operationally, it feeds the input to $p$, obtains a
sample $b$ from $p$'s output distribution, and feeds the sample to $q$
to obtain the final distribution. Iteration $\pstar\polp$ is defined
as the least fixpoint of the map on Markov kernels $X \mapsto
1 \pcomp \den{p}; X$, which is continuous in a DCPO that we will
develop in the following sections. We will show that this definition,
which is simple and is based on standard techniques from domain
theory, coincides with the semantics proposed in previous
work~\cite{\pnkpaper}.

\paragraph{Basic Properties.}%
To clarify the nature of predicates and other primitives, we establish
two intuitive properties:
\begin{lemma}\label{lem:pred-charact}
Any predicate $t$ satisfies $\den{t}(a) = \unit{a \cap b_t}$,
where $b_t \defeq \den{t}(\Hist)$ in the identity monad.
\end{lemma}
\begin{proof*}
By induction on $t$, using \eqref{eq:monad-left-id} in the induction
step.
\end{proof*}
\begin{lemma}\label{lem:nk-prim-charact}
All atomic programs $p$ (i.e., predicates, $\pdup$, and modifications) satisfy
\[
\den{p}(a) = \unit{\set{f_p(h)}{h \in a}}
\]
for some partial function $f_p : \Hist \pfun \Hist$.
\end{lemma}
\begin{proof*}
Immediate from Figure~\ref{fig:probnetkat} and Lemma~\ref{lem:pred-charact}.
\end{proof*}
Lemma~\ref{lem:pred-charact} captures the intuition that predicates
act like packet \emph{filters}. Lemma~\ref{lem:nk-prim-charact}
establishes that the behavior of atomic programs is captured by their
behavior on individual histories.

Note however that \probnetkat's semantic domain is rich enough to
model interactions between packets. For example, it would be
straightforward to extend the language with new primitives whose
behavior depends on properties of the input set of packet
histories---e.g., a rate-limiting construct $\mathord{@}n$ that
selects at most $n$ packets uniformly at random from the input and
drops all other packets. Our results continue to hold when the
language is extended with arbitrary continuous Markov kernels of
appropriate type, or continuous operations on such kernels.

Another important observation is that although \probnetkat does not
include continuous distributions as primitives, there are programs
that generate continuous distributions by combining choice and
iteration:
\begin{lemma}[Theorem 3 in \citet{FKMRS15a}]
\label{thm:pnk-may-be-continuous}
Let $\pi_0, \pi_1$ denote distinct packets. Let $\polp$ denote the
program that changes the head packet of all inputs to either $\pi_0$
or $\pi_1$ with equal
probability. Then
\[ \den{\pseq{\polp}{(\pseq{\pdup}{\polp})\star}}(\sset{\pi},-) \]
is a continuous distribution.
\end{lemma}
Hence, \probnetkat programs cannot be modeled by functions of
type $\pH \to (\pH \to [0,1])$ in general. We need to define a measure
space over $\pH$ and consider general probability measures.

%%%%%%%%%%%%%%%%%%%%%%%%%%%%%%%%%%%%%%%%%%%%%%%%%%%%%%%%%%%%%%%%%%%%%%%%%%%%%%%%
%%%%%%%%%%%%%%%%%%%%%%%%%%%%%%%%%%%%%%%%%%%%%%%%%%%%%%%%%%%%%%%%%%%%%%%%%%%%%%%%
\section{Cantor Meets Scott}
\label{sec:cantor}

To define continuous probability measures on an infinite set $X$, one
first needs to endow $X$ with a topology---some additional structure
that, intuitively, captures which elements of $X$ are close to each
other or approximate each other. Although the choice of topology is
arbitrary in principle, different topologies induce different notions
of continuity and limits, thus profoundly impacting the concepts
derived from these primitives. Which topology is the ``right'' one for
$\pH$? A fundamental contribution of this paper is to show that there
are (at least) two answers to this question:
\begin{itemize}
\item The initial work on ProbNetKAT \cite{FKMRS15a} uses the
\emph{\textbf{Cantor topology}}. This makes $\pH$ a \emph{standard Borel space},
which is well-studied and known to enjoy many desirable properties.

\item This paper is based on the \emph{\textbf{Scott topology}}, the
standard choice of domain theorists. Although this topology is weaker
in the sense that it lacks much of the useful structure and properties
of a standard Borel space, it leads to a simpler and more
computational account of ProbNetKAT's semantics.
\end{itemize}
Despite this, one view is not better than the other. The main 
advantage of the Cantor topology is that it allows us to reason in terms of a metric.
With the Scott topology, we sacrifice this metric, but in return we are able
to interpret all program operators and programs as continuous functions.
The two views yield different convergence theorem, both of which are useful.
Remarkably,
we can have the best of both worlds: it turns out that the two
topologies generate the same Borel sets, so the probability measures
are the same regardless. We will prove (Theorem \ref{thm:star}) that
the semantics in Figure~\ref{fig:probnetkat} coincides with the
original semantics \cite{FKMRS15a}, recovering all the results from
previous work.  This allows us to freely switch between the two views
as convenient. The rest of this section illustrates the difference
between the two topologies intuitively, defines the topologies
formally and endows $\pH$ with Borel sets, and proves a
general theorem relating the two.

\paragraph*{Cantor and Scott, Intuitively.}%
The Cantor topology is best understood in terms of a distance $d(a,b)$
of history sets $a,b$, formally known as a \emph{metric}.  Define this
metric as $d(a,b) = 2^{-n}$, where $n$ is the length of the shortest
packet history in the symmetric difference of $a$ and $b$ if $a \neq
b$, or $d(a,b)=0$ if $a = b$.  Intuitively, history sets are close if
they differ only in very long histories.  This gives the following
notions of limit and continuity:
\begin{itemize}
  \item $a$ is the limit of a sequence $a_1, a_2, \dots$ iff
  the distance $d(a,a_n)$ approaches $0$ as $n \to \infty$.
  \item a function $f : \pH \to [0,\infty]$ is continuous at point $a$ iff
  $f(a_n)$ approaches $f(a)$ whenever $a_n$ approaches $a$.
\end{itemize}

The Scott topology cannot be described in terms of a metric. It is captured by
a complete partial order $(\pH,\sqleq)$ on history sets.
If we choose the subset order (with suprema given by union) we obtain the following
notions:
\begin{itemize}
  \item $a$ is the limit of a sequence $a_1 \subseteq a_2 \subseteq \dots$ iff
  $a = \bigcup_{n \in \N} a_n$.
  \item a function $f : \pH \to [0,\infty]$ is continuous at point $a$ iff
  $f(a) = \sup_{n \in \N} f(a_n)$ whenever $a$ is the limit of $a_1 \subseteq
  a_2 \subseteq \dots$.
\end{itemize}

% The order is typically meant to capture some notion of definedness, so that
% $a \sqleq b$ means that $a$ is ``more defined'' than $b$. Computationally speaking,
% this gives a way of defining a (possibly infinitary) object $a$ as the supremum
% of a sequence of partial definitions $a_n$.

\paragraph*{Example.}%
To illustrate the difference between Cantor-continuity and
Scott-continuity, consider the function $f(a) \defeq \len a$ that maps
a history set to its (possibly infinite) cardinality. The function is
not Cantor-continuous.  To see this, let $h_n$ denote a history of
length $n$ and consider the sequence of singleton sets
$a_n \defeq \sset{h_n}$. Then $d(a_n,\emptyset) = 2^{-n}$,
\ie the sequence approaches the empty set as $n$ approaches infinity. But the
cardinality $|a_n|=1$ does \emph{not} approach $|\emptyset|=0$. In contrast, the
function is easily seen to be Scott-continuous.

As a second example, consider the function $f(a) \defeq 2^{-k}$, where $k$
is the length of the smallest history \emph{not} in $a$. This function is
Cantor-continuous: if $d(a_n,a)=2^{-n}$, then \[
|f(a_n) - f(a)| \leq 2^{-(n-1)} - 2^{-n} \leq 2^{-n}
\]
Therefore $f(a_n)$ approaches $f(a)$ as the distance $d(a_n,a)$ approaches $0$.
However, the function is not Scott-continuous\footnote{with respect to the orders
$\subseteq$ on $\pH$ and $\leq$ on $\R$}, as all Scott-continuous functions
are monotone.

% \paragraph*{A subtle point.}%
% Recall from Section~\ref{sec:primer} that the notion of continuity of a
% function $f: \pH \to [0,\infty]$ depends on the topologies of both its domain 
% (here $\pH$) and codomain (here $[0,\infty]$). When we call $f$ \emph{Cantor-continuous},
% we consider both $\pH$ and $[0,\infty]$ with their ``standard'' topologies: the Cantor 
% topology on $\pH$ and the Euclidean topology on $[0,\infty]$.
% When we call $f$ \emph{Scott-continuous}, we consider both spaces with their respective
% Scott topologies, \ie as complete partial orders $(\pH,\subseteq)$ and $([0,\infty], \leq)$.

\paragraph*{Approximation.}
The computational importance of limits and continuity comes from the
following idea. Assume $a$ is some complicated (say infinite)
mathematical object. If $a_1,a_2, \dots$ is a sequence of simple (say
finite) objects with limit $a$, then it may be possible to approximate
$a$ using the sequence $(a_n)$. This gives us a computational way of
working with infinite objects, even though the available resources may
be fundamentally finite.  Continuity captures precisely when this is
possible: we can perform a computation $f$ on $a$ if $f$ is
continuous in $a$, for then we can compute the sequence $f(a_1),
f(a_2), \dots$ which (by continuity) converges to $f(a)$.

We will show later that any measure $\mu$ can be approximated by a
sequence of finite measures $\mu_1, \mu_2, \dots$, and that the
expected value $\ex{f}{\mu}$ of a Scott-continuous random variable $f$
is continuous with respect to the measure. Our implementation exploits
this to compute a monotonically improving sequence of approximations
for performance metrics such as latency and congestion
(\S\ref{sec:case-study}).

\paragraph*{Notation.}
We use lower case letters
$a,b,c \subseteq \Hist$ to denote history sets, uppercase letters
$A,B,C \subseteq \pH$ to denote measurable sets (\ie, sets of
history sets), and calligraphic letters
$\BB, \SO, \dots \subseteq \pset{\pH}\vphantom{2^{2^\Hist}}$ % SJS: hack to ensure enough line-spacing
to denote sets of measurable
sets. For a set $X$, we let $\pfin X \defeq \set{Y \subseteq
X}{\len{Y} < \infty}$ denote the finite subsets of $X$ and $\chrf X$
the characteristic function of X. For a statement $\phi$, such as $a \subseteq b$,
we let $[\phi]$ denote $1$ if $\phi$ is true and $0$ otherwise.
% For a
% function $f: X \to Y$, we let $\inv f : \pset Y \to \pset X$ denote
% the preimage function:
% \begin{align*}
%   \inv f(A) &\defeq \set{x \in X}{f(x) \in A} & \inv f(y) &\defeq \inv f(\sset{y}) 
% \end{align*}

\paragraph{Cantor and Scott, Formally.}%
For $\h\in \Hist$ and $b\in\pH$, define
\begin{align}
B_\h &\defeq \set c{\h\in c} & B_b &\defeq \bigcap_{\h\in b} B_\h = \set c{b\subs c}.\label{eq:Bb}
\end{align}
The Cantor space topology, denoted $\CC$, can be generated by closing
$\set{B_\h, \setcompl B_\h}{\h \in \Hist}$ under finite intersection
and arbitrary union.  The Scott topology of the DCPO
$(\pH, \subseteq)$, denoted $\SO$, can be generated by closing
$\set{B_\h}{\h \in \Hist}$ under the same operations and adding the empty set.
The Borel algebra $\BB$ is the smallest
$\sigma$-algebra containing the Cantor-open sets, \ie
$\BB \defeq \sigma(\CC)$.  We write $\BB_b$ for the Boolean subalgebra
of $\BB$ generated by $\set{B_\h}{\h\in b}$.

\begin{lemma}
\ \\[-1em]
\begin{enumerate}[{\upshape(i)}]
\item
$b\subs c \Iff B_c \subs B_b$
\item
$B_b\cap B_c = B_{b\cup c}$
\item
$B_\emptyset = \pH$
\item
$\BBo=\bigcup_{b\in\subpfin H}\BB_b$.
\end{enumerate}
\end{lemma}
Note that if $b$ is finite, then so is $\BB_b$. Moreover, the atoms of
$\BB_b$ are in one-to-one correspondence with the subsets $a\subs
b$. The subsets $a$ determine which of the $B_\h$ occur positively in the
construction of the atom,
\begin{align}
\begin{split}
\atm ab &\defeq \bigcap_{\h\in a}B_\h\cap\bigcap_{\h\in b-a}\setcompl B_\h\\
&= B_a - \bigcup_{a\subset c\subs b} B_c
= \set{c\in\pH}{c\cap b=a},\label{eq:atoms}
\end{split}
\end{align}
where $\subset$ denotes proper subset. The atoms $\atm ab$ are the basic open sets of the Cantor space.
The notation $A_{ab}$ is reserved for such sets.

\begin{lemma}[Figure~\ref{fig:venn}]
\label{lem:sumofatoms}
For $b$ finite and $a\subs b$, $B_a = \bigcup_{a\subs c\subs b} \atm cb$.
\end{lemma}
\begin{proof*}
By \eqref{eq:atoms},
\begin{align*}
\bigcup_{a\subs c\subs b} \atm cb
&= \bigcup_{a\subs c\subs b} \set{d\in\pH}{d\cap b=c}\\
&= \set{d\in\pH}{a\subs d} = B_a.\qedhere
\end{align*}
\end{proof*}

\paragraph*{Scott Topology Properties.}
Let $\SO$ denote the family of Scott-open sets of $(\pH, \subseteq)$. Following are
some facts about this topology.
\begin{itemize}
\item
The DCPO $(\pH, \subseteq)$ is algebraic. The finite elements of $\pH$ are the finite subsets 
$a \in \pfin\Hist$, and their up-closures are $\up{\{a\}} = B_a$.
\item
By Lemma \ref{lem:basic}\eqref{lem:basic:finitebase}, the up-closures $\up{\{a\}}=B_a$
form a base for the Scott topology.
The sets $B_\h$ for $\h\in \Hist$ are therefore a subbase.
\item
Thus, a subset $B\subs\pH$ is Scott-open iff there exists $F\subs\pfin \Hist$ such that $B=\bigcup_{a\in F} B_a$.
\item
The Scott topology is weaker than the Cantor space topology, \eg, $\setcompl B_\h$ is Cantor-open but not Scott-open. However, the Borel sets of the topologies are the same, as $\setcompl B_\h$ is a $\Pi^0_1$ Borel set.\footnote{References to the Borel hierarchy $\Sigma^0_n$ and $\Pi^0_n$ refer to the Scott topology. The Cantor and Scott topologies have different Borel hierarchies.}
\item Although any Scott-open set in $\pH$ is also Cantor-open, a Scott-continuous
function $f : \pH \to \R_+$ is not necessarily Cantor-continuous. This is
because for Scott-continuity we consider $\R_+$ (ordered by $\leq$) with the Scott topology,
but for Cantor-continuity we consider $\R_+$ with the standard Euclidean
topology.
\item Any Scott-continuous function $f : \pH \to \R_+$ is measurable,
because the Scott-open sets of $(\R_+, \leq)$ (\ie, the upper semi-infinite intervals
$(r,\infty] = \up{\sset{r}}$ for $r\geq0$) generate the Borell sets on $\R_+$.
\item
The open sets $\SO$ ordered by the subset relation forms an $\omega$-complete lattice with bottom $\emptyset$ and top $B_\emptyset=\pH$.
\item
The finite sets $a\in\pfin \Hist$ are dense and countable, thus the space is separable.
\item
The Scott topology is not Hausdorff, metrizable, or compact. It is not Hausdorff, as any nonempty open set contains $\Hist$, but it satisfies the weaker $T_0$ separation property: for any pair of points $a,b$ with $a\not\subs b$, $a\in B_a$ but $b\not\in B_a$.
\item
There is an up-closed $\Pi^0_2$ Borel set with an uncountable set of minimal elements.
\item
There are up-closed Borel sets with no minimal elements; for example, the family of cofinite subsets of $H$, a $\Sigma^0_3$ Borel set.
\item
The compact-open sets are those of the form $\up F$, where $F$ is a finite set of finite sets. There are plenty of open sets that are not compact-open, e.g.\ $B_\emptyset - \{\emptyset\} = \bigcup_{\h\in \Hist} B_\h$.
\end{itemize}

\begin{lemma}[{see \citet[Theorem III.13.A]{Halmos50}}]
\label{lem:inclexcl}
Any probability measure is uniquely determined by its values on $B_b$ for $b$ finite.
\end{lemma}
\begin{proof*}
For $b$ finite, the atoms of $\BB_b$ are of the form \eqref{eq:atoms}. By the inclusion-exclusion principle (see Figure~\ref{fig:venn}),
\begin{align}
\mu(\atm ab)
&= \mu(B_a-\bigcup_{a\subset c\subs b} B_c)
= \sum_{a\subs c\subs b}(-1)^{\len{c-a}}\mu(B_c).\label{eq:inclexcl}
\end{align}
Thus $\mu$ is uniquely determined on the atoms of $\BB_b$ and
therefore on $\BB_b$. As $\BBo$ is the union of the $\BB_b$ for finite
$b$, $\mu$ is uniquely determined on $\BBo$. By the monotone class
theorem, the Borel sets $\BB$ are the smallest monotone class
containing $\BBo$, and since $\mu(\bigcup_n A_n)=\sup_n\mu(A_n)$ and
$\mu(\bigcap_n A_n)=\inf_n\mu(A_n)$, we have that $\mu$ is determined
on all Borel sets.
\end{proof*}

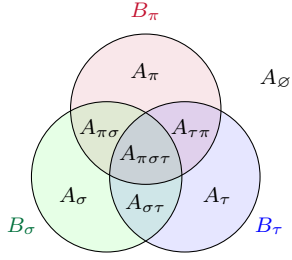
\begin{figure}[t]
\begin{center}
\begin{tikzpicture}
\small
\fill[fill=cornellred,fill opacity=0.1] (90:6mm) circle (10mm);
\fill[fill=green,fill opacity=0.1] (210:6mm) circle (10mm);
\fill[fill=blue,fill opacity=0.1] (-30:6mm) circle (10mm);
\draw (90:6mm) circle (10mm);
\draw (210:6mm) circle (10mm);
\draw(-30:6mm) circle (10mm);
\node at (0,0) {$A_{\pi\sigma\tau}$};
\node at (90:11mm) {$A_{\pi}$};
\node at (210:11mm) {$A_{\sigma}$};
\node at (-30:11mm) {$A_{\tau}$};
\node at (270:6.5mm) {$A_{\sigma\tau}$};
\node at (30:7mm) {$A_{\tau\pi}$};
\node at (-210:7mm) {$A_{\pi\sigma}$};
\node at (30:20mm) {$A_{\emptyset}$};
\node at (90:19mm) {\color{cornellred}$B_\pi$};
\node at (210:19mm) {\color{dartmouthgreen}$B_\sigma$};
\node at (-30:19mm) {\color{blue}$B_\tau$};
\end{tikzpicture}
\end{center}
\caption{Relationship of the basic Scott-open sets $B_a$ to the basic Cantor-open sets $A_{ab}$ for $b=\{\pi,\sigma,\tau\}$ and $a\subs b$. The regions labeled $A_\emptyset$, $A_\pi$, $A_{\pi\sigma}$, etc.\ represent the basic Cantor-open sets $A_{\emptyset,b}$, $A_{\{\pi\},b}$, $A_{\{\pi,\sigma\},b}$, etc. These are the atoms of the Boolean algebra $\BB_b$. Several basic Scott-open sets are not shown, e.g.\ $B_{\{\pi,\sigma\}} = B_\pi\cap B_\sigma = A_{\{\pi,\sigma\},b}\cup A_{\{\pi,\sigma,\tau\},b}$.}
\label{fig:venn}
\end{figure}

\paragraph*{Extension Theorem.}
\label{sec:extension}
We now prove a useful extension theorem (Theorem \ref{thm:extension})
that identifies necessary and sufficient conditions for extending a
function $\SO\to[0,1]$ defined on the Scott-open sets of $\pH$ to a
measure $\BB\to[0,1]$. The theorem yields a remarkable linear
correspondence between the Cantor and Scott topologies
(Theorem \ref{thm:CantorScott}). We prove it for $\pH$ only, but
generalizations may be possible.

\begin{theorem}
\label{thm:extension}
A function $\mu:\set{B_b}{\text{$b$ finite}}\to[0,1]$
extends to a measure $\mu:\BB\to[0,1]$ if and only if
for all finite $b$ and all $a\subs b$,
\begin{align*}
\sum_{a\subs c\subs b} (-1)^{\len{c-a}} \mu(B_c) &\geq 0.
\end{align*}
Moreover, the extension to $\BB$ is unique.
\end{theorem}
\begin{proof*}
The condition is clearly necessary by \eqref{eq:inclexcl}. For
sufficiency and uniqueness, we use the Carath{\'e}odory extension
theorem.  For each atom $\atm ab$ of $\BB_b$, $\mu(\atm ab)$ is
already determined uniquely by \eqref{eq:inclexcl} and nonnegative by
assumption. For each $B\in\BB_b$, write $B$ uniquely as a union of
atoms and define $\mu(B)$ to be the sum of the $\mu(\atm ab)$ for all
atoms $\atm ab$ of $\BB_b$ contained in $B$. We must show that
$\mu(B)$ is well-defined. Note that the definition is given in terms
of $b$, and we must show that the definition is independent of the
choice of $b$. It suffices to show that the calculation using atoms of
$b'=b\cup\{\h\}$, $\h\not\in b$, gives the same result. Each atom of
$\BB_b$ is the disjoint union of two atoms of $\BB_{b'}$:
\begin{align*}
\atm ab &= \atm{a\cup\{\h\}}{,b\cup\{\h\}} \cup \atm a{,b\cup\{\h\}}
\end{align*}
It suffices to show the sum of their measures is the measure of
$\atm ab$:
\begin{align*}
\mu(\atm a{,b\cup\{\h\}})
&= \sum_{a\subs c\subs b\cup\{\h\}} (-1)^{\len{c-a}} \mu(B_c)\\
&= \sum_{a\subs c\subs b} (-1)^{\len{c-a}} \mu(B_c)
+
\sum_{\mathclap{a\cup\{\h\}\subs c\subs b\cup\{\h\}}} (-1)^{\len{c-a}} \mu(B_c)\\
&= \mu(\atm ab) - \mu(\atm{a\cup\{\h\}}{,b\cup\{\h\}}).
\end{align*}
To apply the Carath{\'e}odory extension theorem, we must show that
$\mu$ is countably additive, \ie that $\mu(\bigcup_n A_n)
= \sum_n \mu(A_n)$ for any countable sequence $A_n\in\BBo$ of pairwise
disjoint sets whose union is in $\BBo$.
For finite sequences $A_n\in\BBo$, write each $A_n$ uniquely as a
disjoint union of atoms of $\BB_b$ for some sufficiently large $b$
such that all $A_n\in\BB_b$. Then $\bigcup_n A_n\in\BB_b$, the values
of the atoms are given by \eqref{eq:inclexcl}, and the value of
$\mu(\bigcup_n A_n)$ is well-defined and equal to $\sum_n \mu(A_n)$.
We cannot have an infinite set of pairwise disjoint nonempty
$A_n\in\BBo$ whose union is in $\BBo$ by compactness. All elements of
$\BBo$ are clopen in the Cantor topology. If $\bigcup_n A_n=A\in\BBo$,
then $\set{A_n}{n\geq 0}$ would be an open cover of $A$ with no finite
subcover.
\end{proof*}

\paragraph*{Cantor Meets Scott.}
\label{sec:CantorMeetsScott}
We now establish a correspondence between the Cantor and
Scott topologies on $\pH$. Proofs omitted from this section can be
found in \iffull{}Appendix \ref{apx:CantorMeetsScott}\else{}the long version of this paper~\cite{smolka17-long}\fi.
Consider the infinite triangular matrix $E$ and its inverse $E^{-1}$
with rows and columns indexed by the finite subsets of $\Hist$, where
\begin{align*}
E_{ac} &= [a\subs c] & E^{-1}_{ac} &= (-1)^{\len{c-a}}[a\subs c].
\end{align*}
These matrices are indeed inverses: For $a,d\in\pfin \Hist$,
\begin{align*}
(E\cdot E^{-1})_{ad} &= \sum_{c}E_{ac}\cdot E^{-1}_{cd}\\
&= \sum_{c}[a\subs c]\cdot [c\subs d]\cdot(-1)^{\len{d-c}}\\
&= \sum_{a\subs c\subs d}(-1)^{\len{d-c}} = [a=d],
\end{align*}
thus $E\cdot E^{-1} = I$, and similarly $E^{-1}\cdot E = I$.

Recall that the Cantor basic open sets are the elements $\atm ab$ for
$b$ finite and $a\subs b$. Those for fixed finite $b$ are the atoms of
the Boolean algebra $\BB_b$. They form the basis of a $2^{\len
b}$-dimensional linear space. The Scott basic open sets $B_a$ for
$a\subs b$ are another basis for the same space. The two bases are
related by the matrix $E[b]$, the $\pset b\times\pset b$ submatrix of
$E$ with rows and columns indexed by subsets of $b$. One can show that
the finite matrix $E[b]$ is invertible with inverse $E[b]^{-1} =
(E^{-1})[b]$.

\begin{lemma}
\label{lem:CantorScott}
Let $\mu$ be a measure on $\pH$ and $b\in\pfin \Hist$. Let $X,Y$ be
vectors indexed by subsets of $b$ such that $X_a = \mu(B_a)$ and $Y_a
= \mu(\atm ab)$ for $a\subs b$. Let $E[b]$ be the $\pset b\times\pset
b$ submatrix of $E$. Then $X=E[b]\cdot Y$.
\end{lemma}

The matrix-vector equation $X=E[b]\cdot Y$ captures the fact that for
$a\subs b$, $B_a$ is the disjoint union of the atoms $\atm cb$ of
$\BB_b$ for $a\subs c\subs b$ (see Figure~\ref{fig:venn}),
and consequently $\mu(B_a)$ is the sum
of $\mu(\atm cb)$ for these atoms. The inverse equation
$X=E[b]^{-1}\cdot Y$ captures the inclusion-exclusion principle for
$\BB_b$.

In fact, more can be said about the structure of $E$. For any
$b\in\pH$, finite or infinite, let $E[b]$ be the submatrix of $E$ with
rows and columns indexed by the subsets of $b$. If $a\cap
b=\emptyset$, then $E[a\cup b] = E[a]\otimes E[b]$, where $\otimes$
denotes Kronecker product. The formation of the Kronecker product
requires a notion of pairing on indices, which in our case is given by
disjoint set union. For example,
\setlength{\kbcolsep}{0pt}
\setlength{\kbrowsep}{2pt}
\newcommand{\br}{\!\!\!\!\\}
\begin{align*}
E[\{\h_1\}] &= \kbordermatrix{
              & \emptyset & \sset{\h_1} \br
  \emptyset   & 1         & 1           \br
  \sset{\h_1} & 0         & 1           \br
}
&
E[\{\h_2\}] &= \kbordermatrix{
              & \emptyset & \sset{\h_2} \br
  \emptyset   & 1         & 1           \br
  \sset{\h_2} & 0         & 1           \br
}
\end{align*}
\begin{align*}
E[\{\h_1,\h_2\}]
&= E[\{\h_1\}]\otimes E[\{\h_2\}]\\
&= \kbordermatrix{
                & \emptyset & \sset{\h_1} & \sset{\h_2} & \sset{\h_1,\h_2} \br
  \emptyset     & 1         & 1           & 1           & 1                \br
  \sset{\h_1}   & 0         & 1           & 0           & 1                \br
  \sset{\h_2}   & 0         & 0           & 1           & 1                \br
  \{\h_1,\h_2\} & 0         & 0           & 0           & 1                \br
}
\end{align*}
As $(E\otimes F)^{-1} = E^{-1}\otimes F^{-1}$ for Kronecker products of invertible matrices, we also have
\begin{align*}
E[\{\h_1\}]^{-1} &=
\left[\begin{array}{rr}
1 & -1\\
0 & 1
\end{array}\right]
&
E[\{\h_2\}]^{-1} &=
\left[\begin{array}{rr}
1 & -1\\
0 & 1
\end{array}\right]
\end{align*}
\begin{align*}
E[\{\h_1,\h_2\}]^{-1} &= E[\{\h_1\}]^{-1}\otimes E[\{\h_2\}]^{-1}\\[1ex]
&=
\left[\begin{array}{rrrr}
1 & -1 & -1 & 1\\
0 & 1 & 0 & -1\\
0 & 0 & 1 & -1\\
0 & 0 & 0 & 1
\end{array}\right].
\end{align*}
$E$ can be viewed as the infinite Kronecker product $\bigotimes_{\h\in \Hist}E[\{\h\}]$.

\begin{theorem}
\label{thm:CantorScott}
The probability measures on $(\pH,\BB)$ are in one-to-one
correspondence with pairs of matrices
$M,N\in\R^{\pfin \Hist\times\pfin \Hist}$ such that
\begin{enumerate}[{\upshape(i)}]
\item
$M$ is diagonal with entries in $[0,1]$,
\item
$N$ is nonnegative, and
\item
$N=E^{-1}ME$.
\end{enumerate}
The correspondence associates the measure $\mu$ with the matrices
\begin{align}
N_{ab} &= \mu(\atm ab) & M_{ab} &= [a=b]\cdot\mu(B_a).\label{eq:MN}
\end{align}
\end{theorem}

%%%%%%%%%%%%%%%%%%%%%%%%%%%%%%%%%%%%%%%%%%%%%%%%%%%%%%%%%%%%%%%%%%%%%%%%%%%%%%%%
%%%%%%%%%%%%%%%%%%%%%%%%%%%%%%%%%%%%%%%%%%%%%%%%%%%%%%%%%%%%%%%%%%%%%%%%%%%%%%%%
\section{A DCPO on Markov Kernels}
\label{sec:DCPO}

In this section we define a continuous DCPO on Markov kernels. Proofs
omitted from this section can be found
in \iffull{}Appendix \ref{apx:DCPO}\else{}the long version of this
paper~\cite{smolka17-long}\fi.

We will interpret all program operators defined
in Figure~\ref{fig:probnetkat} also as operators on Markov kernels: for an operator
$\den{p \otimes q}$ defined on programs $p$ and $q$, we obtain a definition
of $P \otimes Q$ on Markov kernels $P$ and $Q$ by replacing $\den{p}$ with $P$
and $\den{q}$ with Q in the original definition. Additionally we define
$\pcomp$ on probability measures as follows:
\begin{align*}
  (\mu \pcomp \nu)(A) &\defeq (\mu\times\nu)(\set{(a,b)}{a \cup b \in A})
\end{align*}
The corresponding operation on programs and kernels as defined in
Figure~\ref{fig:probnetkat} can easily be shown to be equivalent to a
pointwise lifting of the definition here.

For measures $\mu,\nu$ on $\pH$, define $\mu\sqleq\nu$ if
$\mu(B)\leq\nu(B)$ for all $B\in\SO$. This order was first defined by
Saheb-Djahromi \cite{Saheb-Djahromi80}.
\begin{theorem}[{\cite{Saheb-Djahromi80}}]
\label{thm:cpo}
The probability measures on the Borel sets generated by the Scott
topology of an algebraic DCPO ordered by $\sqleq$ form a DCPO.
\end{theorem}
Because $(\pH,\subseteq)$ is an algebraic DCPO, Theorem \ref{thm:cpo}
applies.\footnote{A beautiful proof based on
Theorem~\ref{thm:extension} can be found in \iffull{}Appendix \ref{apx:DCPO}\else{}the long version of this
paper~\cite{smolka17-long}\fi.}
In this case, the bottom and top elements are $\dirac\emptyset$ and $\dirac H$ respectively.
\begin{lemma}
\label{lem:upperbound}
$\mu\sqleq\mu\pcomp\nu$ and $\nu\sqleq\mu\pcomp\nu$.
\end{lemma}
Surprisingly, despite Lemma \ref{lem:upperbound}, the probability
measures do not form an upper semilattice under $\sqleq$, although
counterexamples are somewhat difficult to construct. See
\iffull{}Appendix \ref{apx:nosemilattice}\else{}the long version of this
paper~\cite{smolka17-long}\fi{} for an example.

Next we lift the order $\sqleq$ to Markov kernels
$P:\pH\times\BB\to[0,1]$. The order is defined pointwise on kernels
regarded as functions $\pH\times\SO\to[0,1]$; that is,
\begin{align*}
P\sqleq Q \defiff \forall a\in\pH.\;
\forall B\in \SO.\; P(a,B)\leq Q(a,B).
\end{align*}
There are several ways of viewing the lifted order $\sqleq$, as shown
in the next lemma.
\begin{lemma}
\label{lem:kernelorder}
The following are equivalent:
\begin{enumerate}[{\upshape(i)}]
\item
$P\sqleq Q$, i.e., $\forall a\in\pH$ and $B\in \SO$, $P(a,B)\leq Q(a,B)$;
\item
$\forall a\in\pH$, $P(a,-)\sqleq Q(a,-)$ in the DCPO $\meas(\pH)$;
\item
$\forall B\in \SO$, $P(-,B)\sqleq Q(-,B)$ in the DCPO $\pH\to[0,1]$;
\item
$\curry P\sqleq\curry Q$ in the DCPO $\pH\to\meas(\pH)$.
\end{enumerate}
\end{lemma}
A Markov kernel $P : \pH \times \BB \to [0,1]$ is \emph{continuous} if
it is Scott-continuous in its first argument; i.e., for any fixed
$A\in \SO$, $P(a,A)\leq P(b,A)$ whenever $a\subs b$, and for any
directed set $D\subs\pH$ we have
\(
P(\bigcup D,A) = \sup_{a\in D} P(a,A)
\).
This is equivalent to saying that $\curry P:\pH\to\meas(\pH)$ is
Scott-continuous as a function from the DCPO $\pH$ ordered by $\subs$
to the DCPO of probability measures ordered by $\sqleq$. We will show
later that all \PNK\ programs give rise to continuous kernels.

\begin{theorem}
\label{thm:kernelDCPO}
The continuous kernels $P:\pH\times\BB\to[0,1]$ ordered by $\sqleq$
form a continuous DCPO with basis consisting of kernels of the form
$\bPd$ for $P$ an arbitrary continuous kernel and $b,d$ filters on
finite sets $b$ and $d$; that is, kernels that drop all input packets
except for those in $b$ and all output packets except those in $d$.
\end{theorem}

It is not true that the space of continuous kernels is algebraic with
finite elements
$\bPd$. See \iffull{}Appendix \ref{apx:algebraic}\else{}the long
version of this paper~\cite{smolka17-long}\fi{} for a counterexample.

\section{Continuity and Semantics of Iteration}
\label{sec:continuity}

%% \begin{figure}[t]
%% \[
%% \small
%% \begin{aligned}
%% \bigsqcup_{n\geq0} \unit{a_n} &~=~
%%   \unitop \Big(\bigcup_{n\geq0} a_n \Big) \\
%% \Big(\bigsqcup_{n\geq0} \mu_n \Big) \bind P &~=~
%%   \bigsqcup_{n\geq0} \big( \mu_n \bind P \big) \\
%% \mu \bind \Big(\bigsqcup_{n\geq0} P_n \Big) &~=~
%%   \bigsqcup_{n\geq0} \big( \mu \bind P_n \big) \\
%% \Big( \bigsqcup_{n\geq0} \mu_n \Big)(O) &~=~
%%   \lim_{n \to \infty} \mu_n(O), \quad O \in \mathcal{O}\\
%% \mu\Big( \bigcup_{n\geq0} A_n \Big) &~=~
%%   \lim_{n \to \infty} \mu(A_n)\\
%% \end{aligned}
%% \]
%% \caption{Scott-Continuity of important functions.}
%% \label{fig:basic-continuity}
%% \end{figure}

This section develops the technology needed to establish that
all \PNK\ programs give continuous Markov kernels and that all program
operators are themselves continuous. These results are needed for the
least fixpoint characterization of iteration and also pave the way for
our approximation results (\S\ref{sec:approx}).

The key fact that underpins these results is that
Lebesgue integration respects the orders on measures and on functions:
\begin{theorem}
\label{thm:int-continuous}
Integration is Scott-continuous in both arguments:
\vspace{-0.5ex}
\begin{enumerate}[{\upshape(i)}]
\item
For any Scott-continuous function $f:\pH\to[0,\infty]$, the map
\begin{align}
\mu \mapsto \int f\,d\mu\label{eq:int-continuous}
\end{align}
is Scott-continuous with respect to the order $\sqleq$ on $\MM(\pH)$.
\item
For any probability measure $\mu$, the map
\begin{align}
f \mapsto \int f\,d\mu\label{eq:int-continuous2}
\end{align}
is Scott-continuous with respect to the order on $[\pH\to[0,\infty]]$.
\end{enumerate}
\end{theorem}

The proofs of the remaining results in this section are somewhat long
and mostly routine, but can be found
in \iffull{}Appendix \ref{apx:continuity}\else{}the long version of
this paper~\cite{smolka17-long}\fi.

\begin{theorem}
\label{thm:continuouskernels}
The deterministic kernels associated with any Scott-continuous
function $f:D\to E$ are continuous, and the following operations on
kernels preserve continuity: product, integration, sequential
composition, parallel composition, choice, iteration.
\end{theorem}

The above theorem implies that $Q \mapsto \ptrue \pcomp P \cmp Q$ is a continuous
map on the DCPO of continuous Markov kernels. Hence $P\star = \bigsqcup_n P^{(n)}$
is well-defined as the least fixed point of that map.

\begin{corollary}
\label{cor:continuouskernels}
Every \PNK\ program denotes a continuous Markov kernel.
\end{corollary}
The next theorem is the key result that enables a practical implementation:
\begin{theorem}
\label{thm:continuouspnkops}
The following semantic operations are continuous functions of the DCPO of continuous kernels: product, parallel composition, $\curry$, sequential composition, choice, iteration. (Figure~\ref{fig:program-continuity}.)
\end{theorem}

\begin{figure}[t]
\[
\begin{aligned}
\punion{\Big( \bigsqcup_{n\geq0} P_n \Big)}{Q} &~=~
  \bigsqcup_{n\geq0} \Big(\punion{P_n}{Q}\Big)\\
%
%\punion{Q}{\Big( \bigsqcup_{n\geq0} P_n \Big)} &~=~
%  \bigsqcup_{n\geq0} \Big(\punion{Q}{P_n}\Big)\\
%
\Big( \bigsqcup_{n\geq0} P_n \Big) \opr Q &~=~
  \bigsqcup_{n\geq0} \Big(P_n \opr Q\Big)\\
%
%Q \opr \Big(\bigsqcup_{n\geq0} P_n \Big) &~=~
%  \bigsqcup_{n\geq0} \Big(Q \opr P_n \Big)\\
%
\pseq{\Big( \bigsqcup_{n\geq0} P_n \Big)}{Q} &~=~
  \bigsqcup_{n\geq0} \Big(\pseq{P_n}{Q}\Big)\\
\pseq{Q}{\Big( \bigsqcup_{n\geq0} P_n \Big)} &~=~
  \bigsqcup_{n\geq0} \Big(\pseq{Q}{P_n}\Big)\\
\pstar{\Big( \bigsqcup_{n\geq0} P_n \Big)} &~=~
  \bigsqcup_{n\geq0} \Big( \pstar{P_n} \Big)\\
\end{aligned}
\]
\caption{Scott-Continuity of program operators (Theorem \ref{thm:continuouspnkops}).}
\label{fig:program-continuity}
\end{figure}

The semantics of iteration presented in \cite{\pnkpaper}, defined in
terms of an infinite process, coincides with the least fixpoint
semantics presented here. The key observation is the relationship
between weak convergence in the Cantor topology and fixpoint convergence 
in the Scott topology:
\begin{theorem}
\label{thm:approx}
Let $A$ be a directed set of probability measures with respect to
$\sqleq$ and let $f:\pH\to[0,1]$ be a Cantor-continuous function. Then
\begin{align*}
\lim_{\mu\in A}\int_{c\in\pH} f(c)\cdot d\mu &= \int_{c\in\pH} f(c)\cdot d(\tbigsqcup A).
\end{align*}
\end{theorem}
This theorem implies that $P^{(n)}$ weakly converges to $P\star$ in the
Cantor topology.
\cite{\pnkpaper} showed that $P^{(n)}$ also weakly converges to $P\oldstar$ in
the Cantor topology, where we let $P\oldstar$ denote the iterate of $P$ as defined in \cite{\pnkpaper}.
But since $(\pH,\CC)$ is a Polish space, this implies that $P\star = P\oldstar$.
\begin{lemma}
\label{lem:separation}
In a Polish space $D$, the values of
\begin{align*}
\int_{a\in D} f(a)\cdot\mu(da)
\end{align*}
for continuous $f:D\to[0,1]$ determine $\mu$ uniquely.
\end{lemma}
\begin{corollary}
\label{thm:star}
$P\oldstar = \bigsqcup_n\pp n = P\star$.
\end{corollary}

\section{Approximation}
\label{sec:approx}

\newcommand{\aprx}[2]{[#1]_{#2}}

We now formalize a notion of approximation for ProbNetKAT
programs. Given a program $p$, we define the $n$-th approximant
$\aprx{p}{n}$ inductively as
\begin{align*}
  \aprx{p}{n} &\defeq p
    \quad(\text{for $p$ primitive})\\
  \aprx{q \opr r}{n} &\defeq \aprx{q}{n} \opr \aprx{r}{n}\\
    % &&(\text{for $p = q \opr r$})\\
  \aprx{q \pcomp r}{n} &\defeq \aprx{q}{n} \pcomp \aprx{r}{n}\\
    % &&(\text{for $p = q \pcomp r$})\\
  \aprx{q \cmp r}{n} &\defeq \aprx{q}{n} \cmp \aprx{r}{n}\\
    % &&(\text{for $p = q \cmp r$})\\
  \aprx{q^*}{n} &\defeq \ksn{(\aprx{q}{n})}{n}
    % &&(\text{for $p = q^*$})\\
\end{align*}
Intuitively, $\aprx{p}{n}$ is just $p$ where iteration $-^*$ is
replaced by bounded iteration $\ksn{-}{n}$. Let $\den{p}_n$ denote the
Markov kernel obtained from the $n$-th approximant:
$\den{\aprx{p}{n}}$.

\begin{theorem}
\label{thm:finite-approx}
The approximants of a program $p$ form a $\sqleq$-increasing chain with supremum
$p$, that is \[
\den{p}_1 \sqleq \den{p}_2 \sqleq \dots \qquad\text{and}\qquad
\bigsqcup_{n \geq 0} \den{p}_n = \den{p}
\]
\end{theorem}
\begin{proof*}
By induction on $p$ and continuity of the operators.\qedhere
\end{proof*}
This means that any program can be approximated by a sequence of star-free
programs, which---in contrast to general programs (Lemma~\ref{thm:pnk-may-be-continuous})---can only produce finite distributions.
These finite distributions are sufficient to compute the expected values of
Scott-continuous random variables:
\begin{corollary}
\label{thm:expectation-approx}
Let $\mu \in \Mon(\pH)$ be an input distribution, $p$ be a program, and
$Q : \pH \to [0,\infty]$ be a Scott-continuous random variable.
Let
\begin{align*}
  \nu \defeq \mu \bind \den{p} \qquad\text{and}\qquad
  \nu_n \defeq \mu \bind \den{p}_n
\end{align*}
denote the output distribution and its approximations. Then
\begin{align*}
  \ex Q {~\nu_0} \leq \ex Q {~\nu_1} \leq \dots \qquad\text{and}\qquad
  \sup_{n \in \N} \ex{Q}{~\nu_n} = \ex Q {\nu}
\end{align*}
\end{corollary}
\begin{proof*}
Follows directly from Theorems~\ref{thm:finite-approx} and \ref{thm:int-continuous}.
\end{proof*}
Note that the approximations $\nu_n$ of the output distribution $\nu$ are
always finite, provided the input distribution $\mu$ is finite. Computing an
expected value with respect to $\nu$ thus simply amounts to computing a 
sequence of finite sums $\ex Q {\nu_0}, \ex Q {\nu_1}, \dots$, which is guranteed to converge monotonically to the analytical solution $\ex Q {\nu}$. The approximate semantics $\den{-}_n$ can be thought of
as an executable version of the denotational semantics $\den{-}$. We implement
it in the next section and use it to approximate network metrics based on
the above result.
The rest of this section gives more general approximation results for
measures and kernels on $\pH$,
and shows that we can in fact handle continuous input distributions as well.

A measure is a \emph{finite discrete measure} if it is of the form
$\sum_{a\in F} r_a\dirac a$, where $F\in\pfin{\pfin H}$ is a finite
set of finite subsets of packet histories $H$, $r_a\geq 0$ for all
$a\in F$, $\sum_{a\in F} r_a = 1$. Without loss of generality, we can
write any such measure in the form $\sum_{a\subs b} r_a\dirac a$ for
any $b\in\pfin H$ such that $\tbigcup F\subs b$ by taking $r_a=0$ for
$a\in 2^b - F$.

Saheb-Djahromi \cite[Theorem 3]{Saheb-Djahromi80} shows that every
measure is a supremum of a directed set of finite discrete
measures. This implies that the measures form a continuous DCPO with
basis consisting of the finite discrete measures.  In our model, the
finite discrete measures have a particularly nice characterization:

For $\mu$ a measure and $b\in\pfin H$, define the \emph{restriction of
$\mu$ to $b$} to be the finite discrete measure
\begin{align*}
\rest\mu b &\defeq \sum_{a\subs b} \mu(\atm ab)\dirac a.
\end{align*}

\begin{theorem}\label{thm:directed}
The set $\set{\rest\mu b}{b\in\pfin H}$ is a directed set with
supremum $\mu$. Moreover, the DCPO of measures is continuous with
basis consisting of the finite discrete measures.
\end{theorem}

We can lift the result to continuous kernels, which implies that {\em
every program is approximated arbitrarily closely by programs whose
outputs are finite discrete measures}.

\begin{lemma}
\label{cor:guarded}
Let $b\in\pfin H$. Then $(P\cmp b)(a,-) = \rest{P(a,-)}b$.
\end{lemma}

Now suppose the input distribution $\mu$ in Corollary~\ref{thm:expectation-approx}
is continuous. By Theorem~\ref{thm:directed}, $\mu$ is the supremum of an
increasing chain of finite discrete measures $\mu_1 \sqleq \mu_2 \sqleq \dots$.
If we redefine $\nu_n \defeq \mu_n \bind \den{p}_n$ then by Theorem~\ref{thm:int-continuous} the $\nu_n$ still approximate the output distribution $\nu$
and Corollary~\ref{thm:expectation-approx} continues to hold. Even though the input distribution is now continuous, the output distribution can still be approximated by a chain of finite distributions and hence the expected value
can still be approximated by a chain of finite sums.
\begin{figure}[t]
\centering
\begin{minipage}[b]{.485\columnwidth}
  \includegraphics[width=.95\columnwidth]{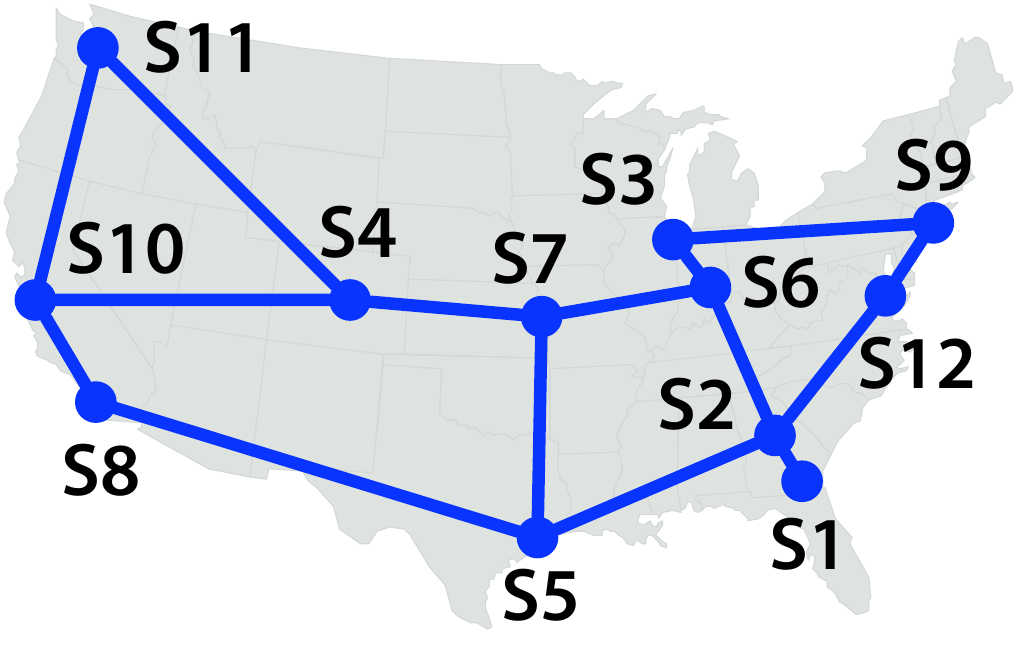}
  \centerline{\small(a) Topology}
\end{minipage}\hfill
\begin{minipage}[b]{.485\columnwidth}
  \includegraphics[width=.95\columnwidth,clip,trim={0 10px 5px 0}]{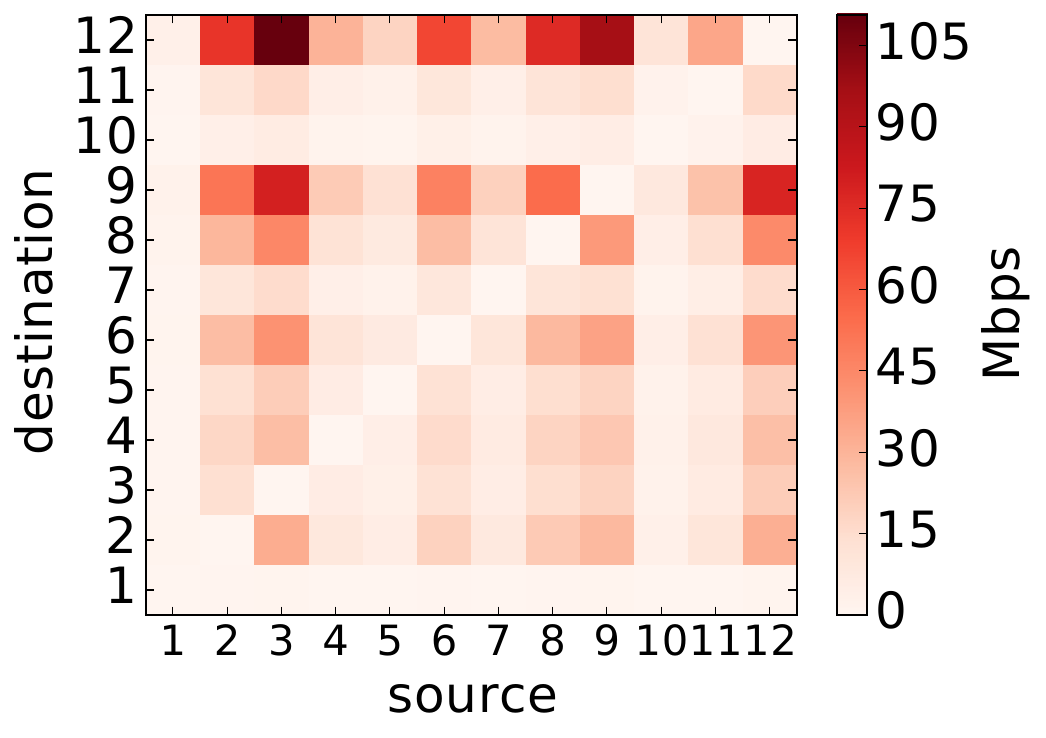}
  \centerline{\small(b) Traffic matrix}
\end{minipage}
\begin{minipage}[b]{.485\columnwidth}
  \includegraphics[width=\columnwidth,clip,trim={0 22px 0 -22px}]{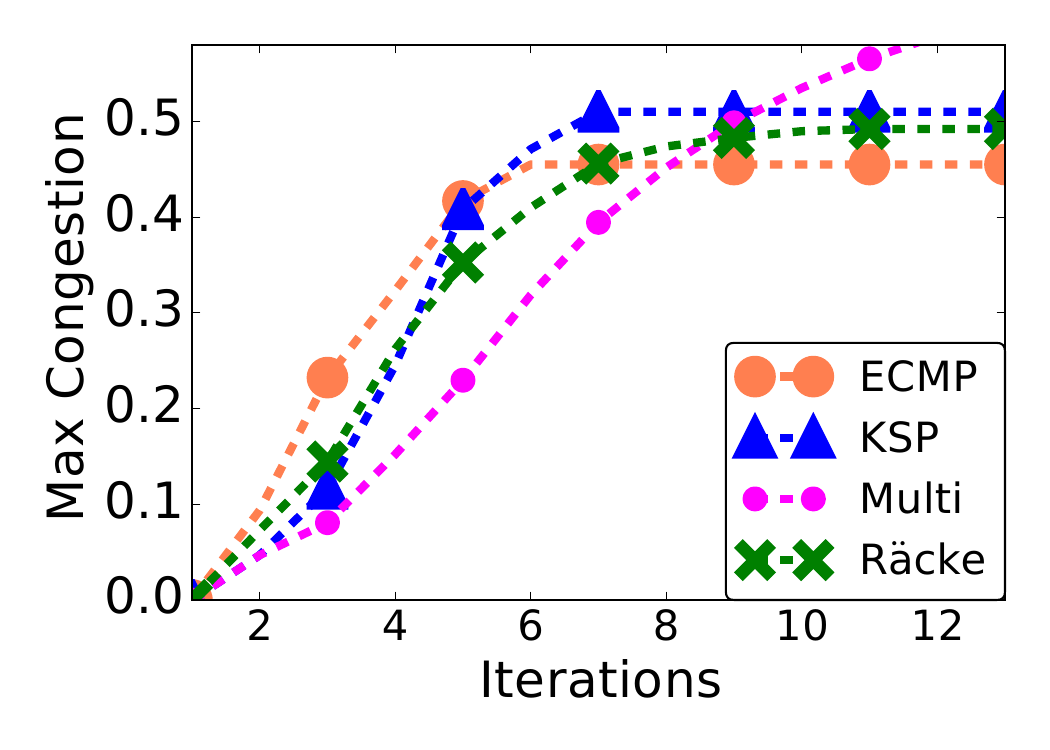}
  \centerline{\small(c) Max congestion}
\end{minipage}\hfill
\begin{minipage}[b]{.485\columnwidth}
  \includegraphics[width=\columnwidth,clip,trim={0 22px 0 -22px}]{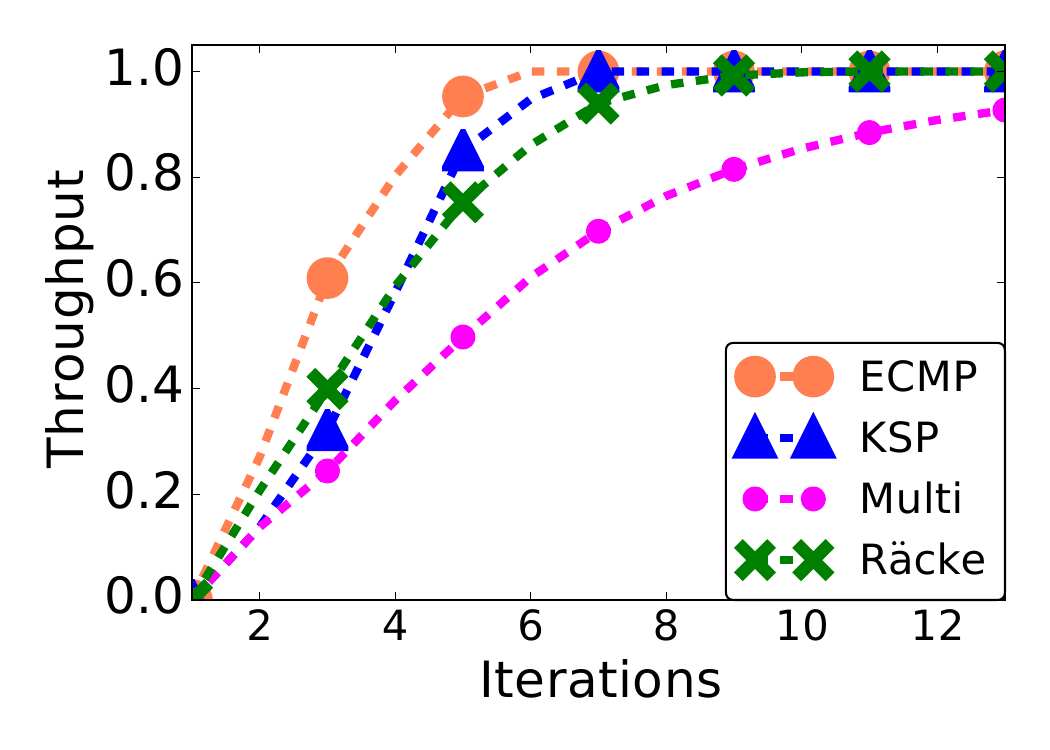}
  \centerline{\small(d) Throughput}
\end{minipage}
\begin{minipage}[b]{.485\columnwidth}
  \includegraphics[width=\columnwidth,clip,trim={0 22px 0 -22px}]{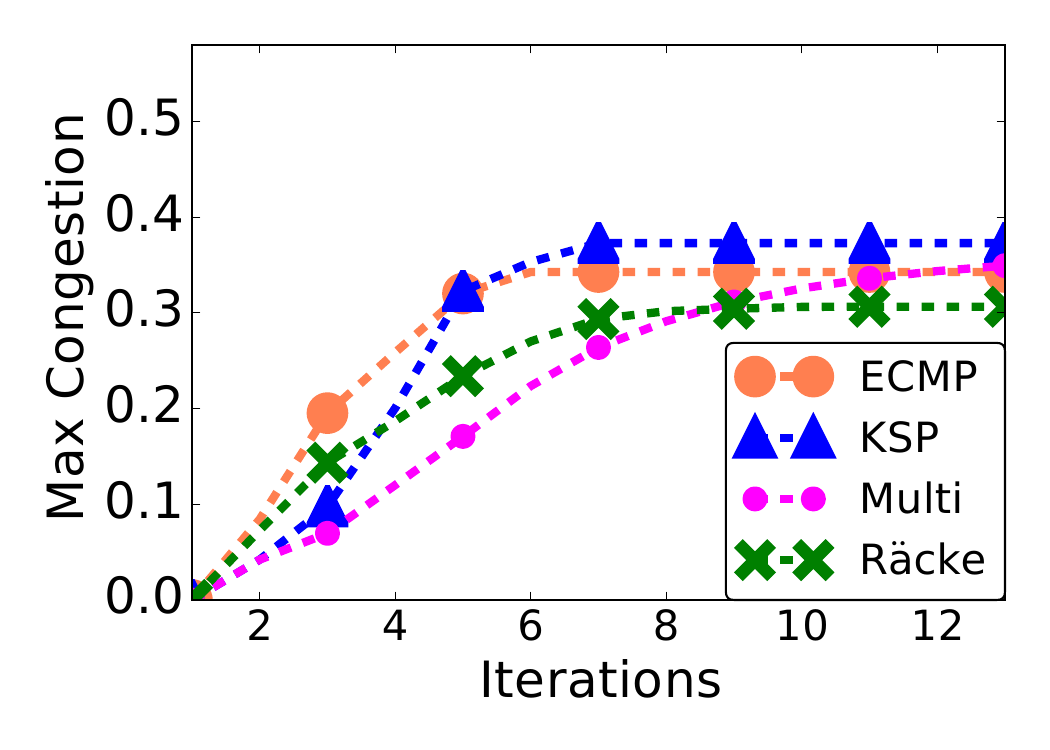}
  \centerline{\small(e) Max congestion}
\end{minipage}\hfill
\begin{minipage}[b]{.485\columnwidth}
  \includegraphics[width=\columnwidth,clip,trim={0 22px 0 -22px}]{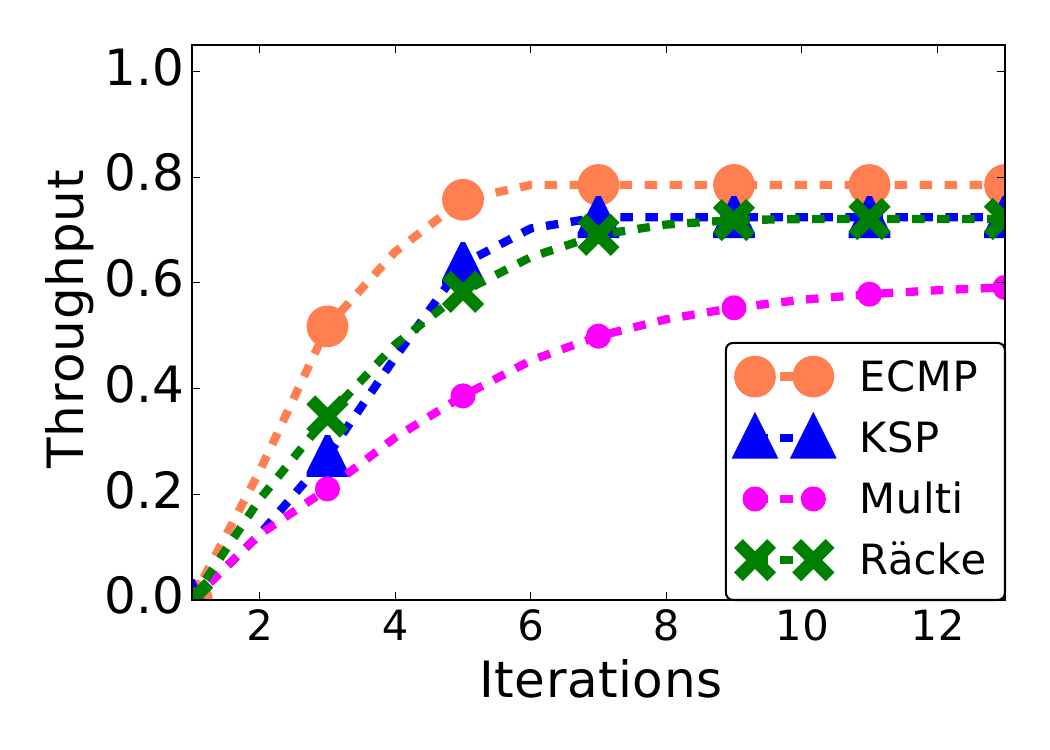}
  \centerline{\small(f) Throughput}
\end{minipage}
\begin{minipage}[b]{.485\columnwidth}
  \includegraphics[width=\columnwidth,clip,trim={0 22px 0 -22px}]{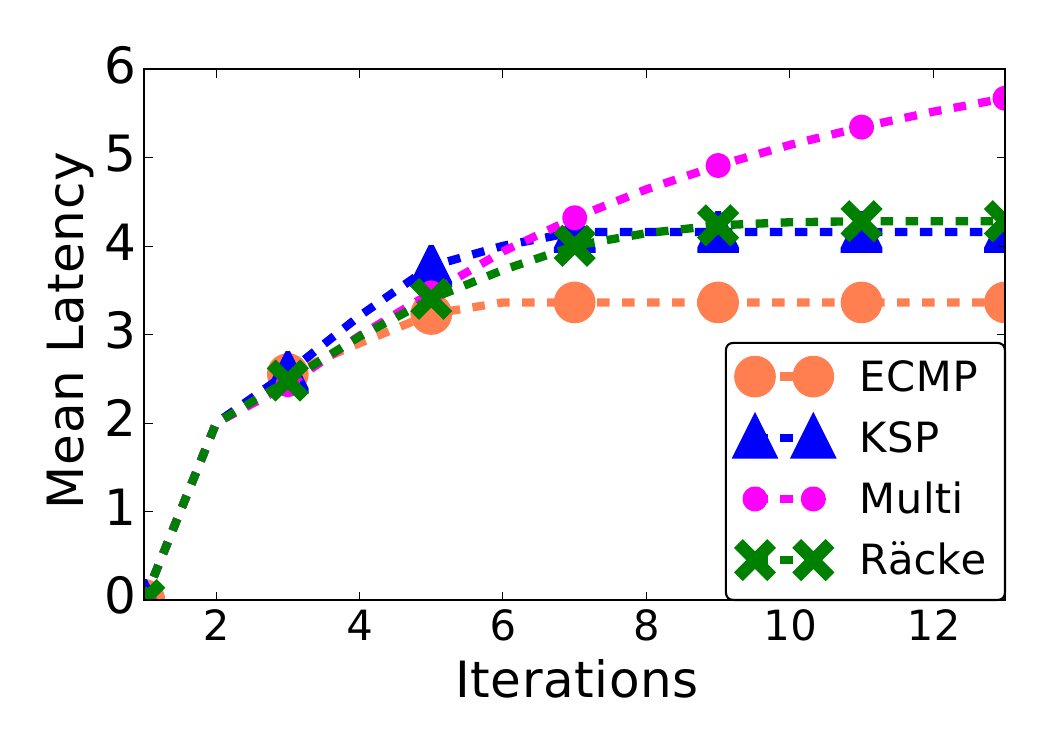}
  \centerline{\small(g) Path length}
\end{minipage}\hfill
\begin{minipage}[b]{.485\columnwidth}
  \includegraphics[width=\columnwidth,clip,trim={0 22px 0 -22px}]{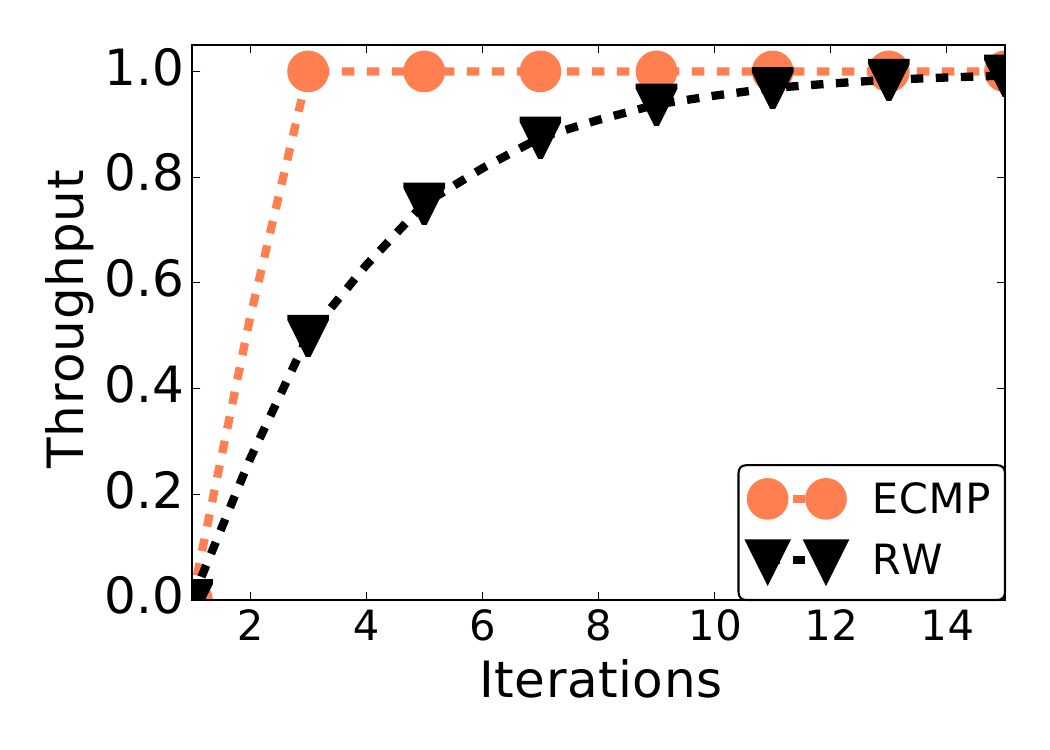}
  \centerline{\small(h) Random walk}
\end{minipage}
\caption{Case study with Abilene: (c, d) without loss. (e, f) with faulty links.
(h) random walk in 4-cycle: all packets are eventually delivered.}
\label{fig:abilene}
\end{figure}

%%%%%%%%%%%%%%%%%%%%%%%%%%%%%%%%%%%%%%%%%%%%%%%%%%%%%%%%%%%%%%%%%%%%%%%%%%%%%%%%
\section{Implementation and Case Studies}
\label{sec:case-study}

We built a simple interpreter for ProbNetKAT in OCaml that implements
the denotational semantics as presented in
Figure~\ref{fig:probnetkat}. Given a query, the interpreter
approximates the answer through a monotonically increasing sequence of
values (Theorems~\ref{thm:finite-approx}
and \ref{thm:expectation-approx}). Although preliminary in
nature---more work on data structures and algroithms for manipulating
distributions would be needed to obtain an efficient
implementation---we were able to use our implementation to conduct
several case studies involving probabilistic reasoning about
properties of a real-world network: Internet2's Abilene backbone. 

\paragraph{Routing.}
In the networking literature, a large number of traffic engineering
(TE) approaches have been explored. We built ProbNetKAT
implementations of each of the following routing schemes:
\begin{itemize}
\item{\textbf{Equal Cost Multipath Routing (ECMP):}} The network
  uses all least-cost paths between each source-destination pair, and
  maps incoming traffic flows onto those paths randomly. ECMP can
  reduce congestion and increase throughput, but can also perform
  poorly when multiple paths traverse the same bottleneck link.
\item{\textbf{$k$-Shortest Paths (KSP):}} The network uses the
  top $k$-shortest paths between each pair of hosts, and again maps
  incoming traffic flows onto those paths randomly. This approach
  inherits the benefits of ECMP and provides improved fault-tolerance
  properties since it always spreads traffic across $k$ distinct
  paths.
\item{\textbf{Multipath Routing (Multi):}} This is similar to KSP, except that it
  makes an independent choice from among the $k$-shortest paths at
  each hop rather than just once at ingress. This approach dynamically
  routes around bottlenecks and failures but can use extremely long
  paths---even ones containing loops.
\item{\textbf{Oblivious Routing (R\"{a}cke):}} The network
  forwards traffic using a pre-computed probability distribution on
  carefully constructed overlays. The distribution is constructed in
  such a way that guarantees worst-case congestion within a
  polylogarithmic factor of the optimal scheme, regardless of the
  demands for traffic.
\end{itemize}
Note that all of these schemes rely on some form of randomization and
hence are probabilistic in nature.

%\paragraph{Network.}%
%After modeling topology ($t$) and programs ($p$) to be executed at switches, we can
%model the entire network as $\pseq{\pstar{(\pseq{p}{t})}}{p}$.  We can
%sequentially compose this with input traffic distribution ($in$) to compute the
%network behavior as $\pseq{in}{\pseq{\pstar{(\pseq{p}{t})}}{p}}$.

\paragraph{Traffic Model.}
Network operators often use traffic models constructed from historical
data to predict future performance. We built a small OCaml tool that
translates traffic models into ProbNetKAT programs using a simple
encoding. Assume that we are given a traffic matrix (TM) that relates
pairs of hosts $(u,v)$ to the amount of traffic that will be sent from
$u$ to $v$. By normalizing each TM entry using the aggregate demand
$\sum_{(u,v)}TM(u,v)$, we get a probability distribution $d$ over
pairs of hosts. For a pair of source and destination $(u,v)$, the
associated probability $d(u,v)$ denotes the amount of traffic from $u$
to $v$ relative to the total traffic. Assuming uniform packet sizes,
this is also the probability that a random packet generated in the
network has source $u$ and destination $v$. So, we can encode a TM as
a program that generates packets according to $d$:
\[
  \begin{aligned}
    inp &\triangleq ~\scalebox{2}{$\oplus$}_{d(u,v)} \pi_{(u,v)}! \\
    \mbox{where, }\pi_{(u,v)}! &\triangleq \pseq{\pseq{\modify{\src}{u}}{\modify{\dst}{v}}}{\modify{\sw}{u}}
  \end{aligned}
\]
$\pi_{(u,v)}!$ generates a packet at $u$ with source $u$ and
destination $v$. For any (non-empty) input, $inp$ generates a
distribution $\mu$ on packet histories which can be fed to the network
program. For instance, consider a uniform traffic distribution for our
4-switch example (see Figure~\ref{fig:4cycle}) where each node sends equal traffic to every other
node. There are twelve $(u,v)$ pairs with $u \neq v$. So,
$d(u,v)_{u \neq v} = \frac{1}{12}$ and $d(u,u)=0$. We also store the
aggregate demand as it is needed to model queries such as expected
link congestion, throughput etc.

\paragraph{Queries.}%
Our implementation can be used to answer probabilistic queries about a
variety of network performance properties. \S\ref{sec:overview} showed
an example of using a query to compute expected congestion. We can
also measure expected mean latency in terms of path length:

\begin{lstlisting}[language=ML]
let path_length (h:Hist.t) : Real.t =
  Real.of_int ((Hist.length h)/2 + 1)
let lift_query_avg
  (q:Hist.t -> Real.t) : (HSet.t -> Real.t) =
  fun hset ->
    let n = HSet.length hset in
    if n = 0 then Real.zero else
    let sum = HSet.fold hset ~init:Real.zero
      ~f:(fun acc h -> Real.(acc + q h)) in
    Real.(sum / of_int n)
\end{lstlisting}

\noindent The latency function (\texttt{path\_length}) counts the number of
switches in a history. We lift this function to sets and compute the
expectation (\texttt{lift\_query\_avg}) by computing the average over
all histories in the set (after discarding empty sets).

\paragraph*{Case Study: Abilene.}
To demonstrate the applicability of ProbNetKAT for reasoning about a
real network, we performed a case study based on the topology and
traffic demands from Internet2's Abilene backbone network as shown in
Figure~\ref{fig:abilene}~(a). We evaluate the traffic engineering
approaches discussed above by modeling traffic matrices based on
NetFlow traces gathered from the production network. A sample TM is
depicted in Figure~\ref{fig:abilene}~(b).

%% If we had empirical information about latency at each hop, we could
%% easily modify the query to compute latency as time delay. %% \pr{What
%% if we have a latency
%% distribution associated with each hop? Can we replace probability with
%% a distribution? How would convolution work for composing programs?}
Figures~\ref{fig:abilene}~(c,d,g) show the expected maximum
congestion, throughput and mean latency. Because we model a network
using the Kleene star operator, we see that the values converge
monotonically as the number of iterations used to approximate Kleene
star increases, as guaranteed by Corollary~\ref{thm:expectation-approx}.

\paragraph{Failures.}%
Network failures such as a faulty router or a link going down are
common in large networks~\cite{gill11}. Hence, it is important to be
able to understand the behavior and performance of a network in the
presence of failures. We can model failures by assigning
empirically measured probabilities to various components---e.g., we
can modify our encoding of the topology so that every link in the
network drops packets with probability $\frac{1}{10}$:
\[
\def\arraycolsep{2pt}
\begin{array}{rcl}
\ell_{1,2} &\defeq&
\pseq{{\pseq{\pseq{\match{\sw}{S_1}}{\match{\pt}{2}}}{\pdup}}}
(({\pseq{\pseq{\modify{\sw}{S_2}}{\modify{\pt}{1}}}{\pdup}}) \oplus_{0.9}
\drop)\\
&\pcomp& \pseq{{\pseq{\pseq{\match{\sw}{S_2}}{\match{\pt}{1}}}{\pdup}}} (({\pseq{\pseq{\modify{\sw}{S_1}}{\modify{\pt}{2}}}{\pdup}}) \oplus_{0.9} \drop)
\end{array}
\]
Figures~\ref{fig:abilene}~(e-f) show the network performance for
Abilene under this failure model. As expected, congestion and
throughput decrease as more packets are dropped. As every link drops
packets probabilistically, the relative fraction of packets delivered
using longer links decreases---hence, there is a decrease in mean
latency. 

\paragraph{Loop detection.}%
Forwarding loops in a network are extremely undesirable as they
increase congestion and can even lead to black holes. With
probabilistic routing, not all loops will necessarily result in a
black hole---if there is a non-zero probability of exiting a loop,
every packet entering it will eventually exit. Consider the example of
random walk routing in the four-node topology from
Figure~\ref{fig:4cycle}. In a random walk, a switch either forwards
traffic directly to its destination or to a random neighbor. As
packets are never duplicated and only exit the network when they reach
their destination, the total throughput is equivalent to the fraction
of packets that exit the network. Figure~\ref{fig:abilene}~(h) shows
that the fraction of packets exiting increases monotonically with
number of iterations and converges to $1$. Moreover, histories can be
queried to test if it encountered a topological loop by checking for
duplicate locations. Hence, given a model that computes all possible
history prefixes that appear in the network, we can query it for
presence of loops. We do this by removing $\mathit{out}$ from our
standard network model and using
$\pseq{\mathit{in}}{\pseq{\pstar{(\pseq{p}{t})}}{p}}$ instead. This
program generates the required distribution on history prefixes.
Moreover, if we generalize packets with wildcard fields, similar to
HSA~\cite{hsa}, we can check for loops symbolically. We have extended
our implementation in this way, and used it to check whether the
network exhibits loops on a number of routing schemes based on
probabilistic forwarding.

%%%%%%%%%%%%%%%%%%%%%%%%%%%%%%%%%%%%%%%%%%%%%%%%%%%%%%%%%%%%%%%%%%%%%%%%%%%%%%%%
%%%%%%%%%%%%%%%%%%%%%%%%%%%%%%%%%%%%%%%%%%%%%%%%%%%%%%%%%%%%%%%%%%%%%%%%%%%%%%%%
\section{Related Work}
\label{sec:rel-work}

This paper builds on previous work on \NK\ \cite{AFGJKSW13a,FKMST15a}
and \PNK\ \cite{FKMRS15a}, but develops a semantics based on ordered
domains as well as new applications to traffic engineering.

\paragraph*{Domain Theory.}
The domain-theoretic treatment of probability measures goes back to
the seminal work of Saheb-Djahromi \cite{Saheb-Djahromi80}, who was
the first to identify and study the CPO of probability measures.
Jones and Plotkin \cite{JonesPlotkin89,Jones89} generalized and
extended this work by giving a category-theoretical treatment and
proving that the probabilistic powerdomain is a monad. It is an open
problem if there exists a cartesian-closed category of continuous
DCPOs that is closed under the probabilistic powerdomain;
see \cite{JungTix98} for a discussion. This is an issue for
higher-order probabilistic languages, but not for \probnetkat, which
is strictly first-order.
Edalat \cite{edalat1994domain,edalat1996scott,edalat1998computational}
gives a computational account of measure theory and integration for
general metric spaces based on domain theory.  More recent papers on
probabilistic powerdomains are \cite{JungTix98,Heckmann94,Graham88}.
All this work is ultimately based on Scott's pioneering
work \cite{Scott72}.

\paragraph*{Probabilistic Logic and Semantics.}%
Computational models and logics for probabilistic programming have
been extensively studied. Denotational and operational semantics for
probabilistic while programs were first studied by
Kozen \cite{K81c}. Early logical systems for reasoning about
probabilistic programs were proposed
in \cite{K85a,Ramshaw79,SahebDjahromi78}. There are also numerous
recent
efforts \cite{GHNR14,DBLP:journals/corr/GretzJKKMO15,KMP13a,LarsenMardarePanangaden12,MorganMcIverSeidel96}.
Sankaranarayanan et al. \cite{sankaranarayanan2013static} propose
static analysis to bound the the value of probability queries.
Probabilistic programming in the context of artificial intelligence
has also been extensively studied in recent
years \cite{Gordon11,Roy11}.  Probabilistic automata in several forms
have been a popular model going back to the early work of
Paz \cite{Paz71}, as well as more recent
efforts \cite{McIverCohenMorgan08,Segala06,Segala95}. Denotational
models combining probability and nondeterminism have been proposed by
several
authors \cite{McIverMorgan04,TixKeimelPlotkin09,VaraccaWinskel06}, and
general models for labeled Markov processes, primarily based on Markov
kernels, have been studied
extensively \cite{Doberkat07,Panangaden98probabilisticrelations,Panangaden09}.
  
Our semantics is also related to the work on event
structures~\cite{NielsenPW79,VaraccaVW06}. A (Prob)NetKAT program
denotes a simple (probabilistic) event structure: packet histories are
events with causal dependency given by extension and with all finite
subsets consistent. We have to yet explore whether the event structure
perspective on our semantics could lead to further applications and
connections to e.g. (concurrent) games.

\paragraph{Networking.}
Network calculus is a general framework for analyzing network behavior
using tools from queuing theory \cite{cruz91}. It has been used to
reason about quantitative properties such as latency, bandwidth, and
congestion. The stochastic branch of network calculus provides tools
for reasoning about the probabilistic behavior, especially in the
presence of statistical multiplexing, but is often considered
difficult to use. In contrast, \PNK\ is a self-contained framework
based on a precise denotational semantics.

Traffic engineering has been extensively studied and a wide variety of
approaches have been proposed for data-center
networks~\cite{al2010hedera,jeyakumar2013eyeq,perry14,zhangVlb05,shieh2010seawall}
and wide-area
networks~\cite{swan,b4,fortz02,applegate2003making,racke2008optimal,kandula2005walking,suchara2011network,he2008toward}.
These approaches optimize for metrics such as congestion, throughput,
latency, fault tolerance, fairness etc. Optimal techniques typically
have high overheads~\cite{danna2012practical}, but
oblivious~\cite{kodialam2009oblivious,applegate2003making} and hybrid
approaches with near-optimal performance~\cite{swan,b4} have recently
been adopted.

%%%%%%%%%%%%%%%%%%%%%%%%%%%%%%%%%%%%%%%%%%%%%%%%%%%%%%%%%%%%%%%%%%%%%%%%%%%%%%%%
%%%%%%%%%%%%%%%%%%%%%%%%%%%%%%%%%%%%%%%%%%%%%%%%%%%%%%%%%%%%%%%%%%%%%%%%%%%%%%%%
\section{Conclusion}
\label{sec:concl}
This paper presents a new order-theoretic semantics for ProbNetKAT in
the style of classical domain theory. The semantics allows a standard
least-fixpoint treatment of iteration, and enables new modes of
reasoning about the probabilistic network behavior. We have used these
theoretical tools to analyze several randomized routing protocols on
real-world data.

The main technical insight is that all programs and the operators defined on
them are continuous, provided we consider the right notion of continuity: that
induced by the Scott topology. Continuity enables precise approximation,
and we exploited this to build an implementation. But continuity is also a powerful
tool for reasoning that we expect to be very helpful in the future development
of \probnetkat's meta theory.
To establish continuity we had to switch from
the Cantor to the Scott topology, and give up reasoning in terms of a metric.
Luckily we were able to show a strong correspondence between the two topologies
and that the Cantor-perspective and the Scott-perspective lead to equivalent
definitions of the semantics.
This allows us to choose whichever perspective is best-suited for the task at hand.

\paragraph*{Future Work.}%
The results of this paper are general enough to accommodate arbitrary extensions
of \probnetkat with continuous Markov kernels or continuous operators on such
kernels. An obvious next step is therefore to investigate extension of the language
that would enable richer network models.
Previous work on deterministic NetKAT included a decision procedure and a sound and
complete axiomatization. In the presence of probabilities we expect a
decision procedure will be hard to devise, as witnessed by several
undecidability results on probabilistic automata. We intend to explore
decision procedures for restricted fragments of the language.  Another
interesting direction is to compile ProbNetKAT programs into suitable
automata that can then be analyzed by a probabilistic model checker
such as PRISM~\cite{KwiatkowskaNP11}. A sound and complete
axiomatization remains subject of further investigation, we can draw
inspiration from recent work \cite{KMP13a,MardarePP16}. Another
opportunity is to investigate a weighted version
of NetKAT, where instead of probabilities we consider weights from an
arbitrary semiring, opening up several other applications---e.g. in
cost analysis. Finally, we would like to explore efficient
implementation techniques including compilation, as well as approaches
based on sampling, following several other probabilistic
languages~\cite{park08,Gordon11}.

%%%%%%%%%%%%%%%%%%%%%%%%%%%%%%%%%%%%%%%%%%%%%%%%%%%%%%%%%%%%%%%%%%%%%%%%%%%%%%%%
%%%%%%%%%%%%%%%%%%%%%%%%%%%%%%%%%%%%%%%%%%%%%%%%%%%%%%%%%%%%%%%%%%%%%%%%%%%%%%%%

\acks

The authors wish to thank Arthur Azevedo de Amorim, David Kahn,
Anirudh Sivaraman, Hongseok Yang, the Cornell PLDG, and the
Barbados Crew for insightful discussions and helpful comments. Our
work is supported by the National Security Agency; the National
Science Foundation under grants CNS-1111698, CNS-1413972, CCF-1422046,
CCF-1253165, and CCF-1535952; the Office of Naval Research under grant
N00014-15-1-2177; the European Research Council under starting grant ProFoundNet (679127); a Leverhulme Prize (PLP-2016-129); and gifts from Cisco, Facebook,
Google, and Fujitsu.

\balance
\bibliographystyle{abbrvnat}
\bibliography{\docroot/bib/prob,\docroot/bib/dk,\docroot/bib/main}

\iffull

\newpage

\appendix

\section{$(\MM,\sqleq)$ is not a Semilattice}
\label{apx:nosemilattice}

Despite the fact that $(\MM,\sqleq)$ is a directed set (Lemma \ref{lem:upperbound}), it is not a semilattice. Here is a counterexample.

Let $b=\{\pi,\sigma,\tau\}$, where $\pi$, $\sigma$, $\tau$ are distinct packets. Let
\begin{gather*}
\mu_1 = \tfrac 12\delta_{\{\pi\}} + \tfrac 12\delta_{\{\sigma\}} \qquad
\mu_2 = \tfrac 12\delta_{\{\sigma\}} + \tfrac 12\delta_{\{\tau\}}\\
\mu_3 = \tfrac 12\delta_{\{\tau\}} + \tfrac 12\delta_{\{\pi\}}.
\end{gather*}
The measures $\mu_1$, $\mu_2$, $\mu_3$ would be the output measures of the programs
$\pi! \oplus \sigma!$, $\sigma! \oplus \tau!$, $\tau! \oplus \pi!$, respectively.

We claim that $\mu_1 \sqcup \mu_2$ does not exist. To see this, define
\begin{gather*}
\nu_1 = \tfrac 12\delta_{\{\tau\}} + \tfrac 12\delta_{\{\pi,\sigma\}} \qquad
\nu_2 = \tfrac 12\delta_{\{\pi\}} + \tfrac 12\delta_{\{\sigma,\tau\}}\\
\nu_3 = \tfrac 12\delta_{\{\sigma\}} + \tfrac 12\delta_{\{\tau,\pi\}}.
\end{gather*}
All $\nu_i$ are $\sqleq$-upper bounds for all $\mu_j$.
(In fact, any convex combination $r\nu_1+s\nu_2+t\nu_3$
for $0\leq r,s,t$ and $r+s+t=1$ is an upper bound for any convex combination
$u\mu_1+v\mu_2+w\mu_3$ for $0\leq u,v,w$ and $u+v+w=1$.)
But we show by contradiction that there cannot exist a measure that is both $\sqleq$-above $\mu_1$ and $\mu_2$
and $\sqleq$-below $\nu_1$ and $\nu_2$. Suppose $\rho$ was such a measure. Since
$\rho\sqleq\nu_1$ and $\rho\sqleq\nu_2$, we have
\begin{gather*}
\rho(B_{\sigma\tau}) \leq \nu_1(B_{\sigma\tau}) = 0 \qquad
\rho(B_{\tau\pi}) \leq \nu_1(B_{\tau\pi}) = 0\\
\rho(B_{\pi\sigma}) \leq \nu_2(B_{\pi\sigma}) = 0.
\end{gather*}
Since $\mu_1\sqleq\rho$ and $\mu_2\sqleq\rho$, we have
\begin{gather*}
\rho(B_\pi) \geq \mu_1(B_\pi) = \tfrac 12 \qquad
\rho(B_\sigma) \geq \mu_1(B_\sigma) = \tfrac 12\\
\rho(B_\tau) \geq \mu_2(B_\tau) = \tfrac 12.
\end{gather*}
But then
\begin{align*}
\rho(\atm\pi b) &= \rho(B_\pi)-\rho(B_{\pi\sigma}\cup B_{\tau\pi}) \geq \tfrac 12\\
\rho(\atm\sigma b) &= \rho(B_\sigma)-\rho(B_{\sigma\tau}\cup B_{\pi\sigma}) \geq \tfrac 12\\
\rho(\atm\tau b) &= \rho(B_\tau)-\rho(B_{\tau\pi}\cup B_{\sigma\tau}) \geq \tfrac 12,
\end{align*}
which is impossible, because $\rho$ would have total weight at least $\tfrac 32$.

\section{Non-Algebraicity}
\label{apx:algebraic}

Here is a counterexample to the conjecture that the elements continuous DCPO of continuous kernels is algebraic with finite elements $\bPd$. Let $\sigma,\tau$ be packets and let $\sigma!$ and $\tau!$ be the programs that set the current packet to $\sigma$ or $\tau$, respectively. For $r\in [\frac 12,1]$, let $P_r = (\sigma!\opr\tau!)\pcomp(\tau!\opr\sigma!)$. On any nonempty input, $P_r$ produces $\{\sigma\}$ with probability $r(1-r)$, $\{\tau\}$ with probability $r(1-r)$, and $\{\sigma,\tau\}$ with probability $r^2+(1-r)^2$. In particular, $P_1$ produces $\{\sigma,\tau\}$ with probability $1$. The kernels $P_r$ for $1/2 \leq r < 1$ form a directed set whose supremum is $P_1$, yet $\{\sigma\}\cmp P_1\cmp\{\sigma,\tau\}$ is not $\sqleq$-bounded by any $P_r$ for $r<1$, therefore the up-closure of $\{\sigma\}\cmp P_1\cmp\{\sigma,\tau\}$ is not an open set.

\section{Cantor Meets Scott}
\label{apx:CantorMeetsScott}

This appendix contains proofs omitted from \S\ref{sec:CantorMeetsScott}.

\begin{proof}[Proof of Lemma \ref{lem:CantorScott}]
For any $a\subs b$,
\begin{align*}
X_a &= \mu(B_a) = \sum_{a\subs c\subs b}\mu(\atm cb)\\
&= \sum_{c}[a\subs c]\cdot[c\subs b]\cdot\mu(\atm cb)\\
&= \sum_{c}E[b]_{ac}\cdot Y_c
= (E[b]\cdot Y)_a.
\end{align*}
\end{proof}

\begin{proof}[Proof of Theorem \ref{thm:CantorScott}]
Given a probability measure $\mu$, certainly (i) and (ii) hold of the matrices $M$ and $N$ formed from $\mu$ by the rule \eqref{eq:MN}. For (iii), we calculate:
\begin{align*}
\lefteqn{(E^{-1}ME)_{ab} = \sum_{c,d} E^{-1}_{ac}M_{cd}E_{db} = \sum_{c,d} E^{-1}_{ac}M_{cd}E_{db}}\qquad\\
&= \sum_{c,d} [a\subs c]\cdot(-1)^{\len{c-a}}\cdot[c=d]\cdot\mu(M_{cd})\cdot[d\subs b]\\
&= \sum_{a\subs c\subs b} (-1)^{\len{c-a}}\cdot\mu(B_c)
= \mu(\atm ab) = N_{ab}.
\end{align*}

That the correspondence is one-to-one is immediate from Theorem \ref{thm:extension}.
\end{proof}

\section{A DCPO on Markov Kernels}
\label{apx:DCPO}

This appendix contains proofs omitted from \S\ref{sec:DCPO}.

\begin{proof}[Proof of Theorem \ref{thm:cpo}]
We prove the theorem for our concrete instance $(\pH,\BB)$.
The relation $\sqleq$ is a partial order. Reflexivity and transitivity are clear, and antisymmetry follows from Lemma \ref{lem:inclexcl}.

To show that suprema of directed sets exist, let $\DD$ be a directed set of measures, and define
\begin{align*}
(\bigsqcup\DD)(B) &\defeq \sup_{\mu\in\DD} \mu(B),\ B\in \SO.
\end{align*}
This is clearly the supremum of $\DD$, provided it defines a valid
measure.\footnote{This is actually quite subtle. One might be tempted to define
\[ (\bigsqcup\DD)(B) \defeq \sup_{\mu\in\DD} \mu(B),\ B\in \BB\]
However, this definition would not give a valid probability measure in
general. In particular, an increasing chain
of measures does not generally converge to its supremum pointwise. However, it
\emph{does} converge pointwise on $\SO$.}
To show this, choose a countable chain $\mu_0 \sqleq \mu_1 \sqleq \cdots$ in $\DD$ such that $\mu_m\sqleq\mu_n$ for all $m<n$ and $(\bigsqcup\DD)(B_c) - \mu_n(B_c) \leq 1/n$ for all $c$ such that $\len c\leq n$. Then for all finite $c\in\pH$, $(\bigsqcup\DD)(B_c)=\sup_n\mu_n(B_c)$. 

Then $\bigsqcup\DD$ is a measure by Theorem \ref{thm:extension} because for all finite $b$ and $a\subs b$,
\begin{align*}
\sum_{a\subs c\subs b} (-1)^{\len{c-a}}(\bigsqcup\DD)(B_c)
&= \sum_{a\subs c\subs b} (-1)^{\len{c-a}}\sup_n\mu_n(B_c)\\
&= \lim_n\sum_{a\subs c\subs b} (-1)^{\len{c-a}}\mu_n(B_c)\\
&\geq 0.
\end{align*}

To show that $\dirac\emptyset$ is $\sqleq$-minimum, observe that for all $B\in \SO$,
\begin{align*}
\dirac\emptyset(B) &= [\emptyset\in B] = [B=B_\emptyset=\pH]
\end{align*}
as $B_\emptyset=\pH$ is the only up-closed set containing $\emptyset$. Thus for all measures $\mu$, $\dirac\emptyset(\pH) = 1 = \mu(\pH)$, and for all $B\in \SO$, $B\neq\pH$, $\dirac\emptyset(B) = 0 \leq \mu(B)$.

Finally, to show that $\dirac H$ is $\sqleq$-maximum, observe that every nonempty $B \in \SO$
contains $H$ because it is up-closed.
Therefore, $\dirac H$ is the constant function $1$ on $\SO-\{\emptyset\}$, making it
$\sqleq$-maximum.
\end{proof}

\begin{proof}[Proof of Lemma \ref{lem:upperbound}]
For any up-closed measurable set $B$,
\begin{align*}
\mu(B) &= \mu(B)\cdot\nu(\pH)
= (\mu\times\nu)(B\times \pH)\\
&= (\mu\times\nu)(\set{(b,c)}{b\in B})\\
&\leq (\mu\times\nu)(\set{(b,c)}{b\cup c\in B})
= (\mu\pcomp\nu)(B).
\end{align*}
and similarly for $\nu$.
\end{proof}

\begin{proof}[Proof of Lemma \ref{lem:kernelorder}]
To show that (i), (ii), and (iv) are equivalent,
\begin{align*}
\lefteqn{\forall a\in\pH\ \forall B\in\SO\ P(a,B)\leq Q(a,B)}\qquad\\
&\Iff \forall a\in\pH\ (\forall B\in\SO\ P(a,B)\leq Q(a,B))\\
&\Iff \forall a\in\pH\ P(a,-)\sqleq Q(a,-)\\
&\Iff \forall a\in\pH\ (\curry P)(a)\sqleq(\curry Q)(a)\\
&\Iff \curry P\sqleq\curry Q.
\end{align*}
To show that (i) and (iii) are equivalent,
\begin{align*}
\lefteqn{\forall a\in\pH\ \forall B\in\SO\ P(a,B)\leq Q(a,B)}\qquad\\
&\Iff \forall B\in\SO\ (\forall a\in\pH\ P(a,B)\leq Q(a,B))\\
&\Iff \forall B\in\SO\ P(-,B)\sqleq Q(-,B).
\end{align*}
\end{proof}

\begin{proof}[Proof of Theorem \ref{thm:kernelDCPO}]
We must show that the supremum of any directed set of continuous Markov kernels is a continuous Markov kernel. In general, the supremum of a directed set of continuous functions between DCPOs is continuous. Given a directed set $\DD$ of continuous kernels, we apply this to the directed set $\set{\curry P:\pH\to\MM(\pH)}{P\in\DD}$ to derive that $\tbigsqcup_{P\in\DD}\curry P$ is continuous, then use the fact that $\curry$ is continuous to infer that $\tbigsqcup_{P\in\DD}\curry P=\curry\tbigsqcup\DD$, therefore $\curry\tbigsqcup\DD$ is continuous. This says that the function $P:\pH\times\BB\to[0,1]$ is continuous in its first argument.

We must still argue that the supremum $\tbigsqcup\DD$ is a Markov kernel, that is, a measurable function in its first argument and a probability measure in its second argument. The first statement follows from the fact that any continuous function is measurable with respect to the Borel sets generated by the topologies of the two spaces. For the second statement, we appeal to Theorem \ref{thm:cpo} and the continuity of $\curry$:
\begin{align*}
(\curry\tbigsqcup\DD)(a)
&= (\tbigsqcup_{P\in\DD}\curry P)(a)
= \tbigsqcup_{P\in\DD}(\curry P)(a),
\end{align*}
which is a supremum of a directed set of probability measures, therefore by Theorem \ref{thm:cpo} is itself a probability measure.

To show that it is a continuous DCPO with basis of the indicated form, we note that for any $a\in\pH$ and $B\in\SO$,
\begin{align}
(\bPd)(a,B) = P(a\cap b,\set c{c\cap d\in B}).\label{eq:kernelDCPO}
\end{align}
Every element of the space is the supremum of a directed set of such elements. Given a continuous kernel $P$, consider the directed set $\DD$ of all elements $\bPd$ for $b,d$ finite. Then for any $a\in\pH$ and $B\in\SO$,
\begin{align}
(\tbigsqcup\DD)(a,B) &= \sup_{b,d\in\pfin \Hist}P(a\cap b,\set c{c\cap d\in B})\label{eq:kernelDCPO0}\\
&= \sup_{d\in\pfin \Hist}P(a,\set c{c\cap d\in B})\label{eq:kernelDCPO1}\\
&= P(a,B),\label{eq:kernelDCPO2}
\end{align}
the inference \eqref{eq:kernelDCPO0} from \eqref{eq:kernelDCPO},
the inference \eqref{eq:kernelDCPO1} from the fact that $P$ is continuous in its first argument, and
the inference \eqref{eq:kernelDCPO1} from the fact that the sets $\set c{c\cap d\in B}$ for $d\in\pfin \Hist$ form a directed set of Scott-open sets whose union is $B$ and that $P$ is a measure in its second argument.
\end{proof}

\section{Continuity of Kernels and Program Operators and a Least-Fixpoint Characterization of Iteration}
\label{apx:continuity}

This appendix contains lemmas and proofs omitted from \S\ref{sec:continuity}.

\subsection{Products and Integration}
\label{apx:productsintegration}

This section develops some properties of products and integration needed for
from the point of view of Scott topology.

As pointed out by Jones \cite[\S3.6]{Jones89}, the product $\sigma$-algebra
of the Borel sets of two topological spaces $X,Y$ is in general not the same as
the Borel sets of the topological product $X\times Y$, although this property does hold for
the Cantor space, as its basic open sets are clopen. More importantly, as also observed in \cite[\S3.6]{Jones89},
the Scott topology on the product of DCPOs with the componentwise order is not necessarily
the same as the product topology. However, in our case, the two topologies coincide.

\begin{theorem}
\label{thm:product}
Let $D_\alpha$, $\alpha<\kappa$, be a collection of algebraic DCPOs with
$F_\alpha$ the finite elements of $D_\alpha$. Then the product $\prod_{\alpha<\kappa} D_\alpha$ with the componentwise order is an algebraic DCPO with finite elements
\begin{align*}
F &= \set{c\in\tprod_\alpha F_\alpha}{\text{$\pi_\alpha(c)=\bot$ for all but finitely many $\alpha$}}.
\end{align*}
\end{theorem}
\begin{proof}
The projections $\pi_\beta:\prod_\alpha D_\alpha\to D_\beta$ are
easily shown to be continuous with respect to the componentwise order. For any $d\in\prod_{\alpha<\kappa} D_\alpha$, the set $\down{\{d\}}\cap F$ is directed, and $d = \bigsqcup (\down{\{d\}}\cap F)$: for any $\alpha$, the set $\pi_\alpha(\down{\{d\}}\cap F) = \down{\{\pi_\alpha(d)\}}\cap F_\alpha$ is directed, thus 
\begin{align*}
\pi_\alpha(d) &= \bigsqcup(\down{\{\pi_\alpha(d)\}}\cap F_\alpha) = \bigsqcup(\pi_\alpha(\down{\{d\}}\cap F))\\
&= \pi_\alpha(\bigsqcup(\down{\{d\}}\cap F)), 
\end{align*}
and as $\alpha$ was arbitrary, $d = \bigsqcup (\down{\{d\}}\cap F)$.

It remains to show that $\up{\{c\}} = \prod_{\alpha<\kappa} \up{\{\pi_\alpha(c)\}}$ is open for $c\in F$. Let $A$ be a directed set with $\bigsqcup A\in\up{\{c\}}$. For each $\alpha$, $\set{\pi_\alpha(a)}{a\in A}$ is directed, and
\begin{align*}
\bigsqcup_{a\in A}\pi_\alpha(a) &= \pi_\alpha(\bigsqcup A) \in \pi_\alpha(\up{\{c\}}) = \up{\{\pi_\alpha(c)\}},
\end{align*}
so there exists $a_\alpha\in A$ such that $\pi_\alpha(a_\alpha)\in\up{\{\pi_\alpha(c)\}}$. Since $A$ is directed, there is a single $a\in A$ that majorizes the finitely many $a_\alpha$ such that $\pi_\alpha(c)\neq\bot$. Then $\pi_\alpha(a)\in\up{\{\pi_\alpha(c)\}}$ for all $\alpha$, thus $a\in\up{\{c\}}$.
\end{proof}

\begin{corollary}
\label{cor:product}
The Scott topology on a product of algebraic DCPOs with respect to the componentwise order coincides with the product topology induced by the Scott topology on each component.
\end{corollary}
\begin{proof}
Let $\prod_{\alpha<\kappa} D_\alpha$ be a product of algebraic DCPOs with $\SO_0$ the product topology and $\SO_1$ the Scott topology. As noted in the proof of Theorem \ref{thm:product}, the projections $\pi_\beta:\prod_\alpha D_\alpha\to D_\beta$ are continuous with respect to $\SO_1$. By definition, $\SO_0$ is the weakest topology on the product such that the projections are continuous, so $\SO_0\subs \SO_1$.

For the reverse inclusion, we use the observation that the sets $\up{\{c\}}$ for finite elements $c\in F$ as defined in Theorem \ref{thm:product} form a base for the topology $\SO_1$. These sets are also open in $\SO_0$, since they are finite intersections of sets of the form $\pi_\alpha^{-1}(\up{\{\pi_\alpha(c)\}})$, and $\up{\{\pi_\alpha(c)\}}$ is open in $D_\alpha$ since $\pi_\alpha(c)\in F_\alpha$. As $\SO_1$ is the smallest topology containing its basic open sets, $\SO_1\subs \SO_0$.
\end{proof}

A function $g:\pH\to\R_+$ is \emph{$\SO$-simple} if it is a finite linear combination of the form $\sum_{A\in F} r_A\chrf A$, where $F$ is a finite subset of $\SO$. Let $S_\SO$ denote the set of $\SO$-simple functions.

\begin{theorem}
\label{thm:integration}
Let $f$ be a bounded Scott-continuous function $f:\pH\to\R_+$. Then
\begin{align*}
\sup_{\substack{g\in S_\SO\\g\leq f}} \int g\,d\mu &= \int f\,d\mu = \inf_{\substack{g\in S_\SO\\f\leq g}} \int g\,d\mu
\end{align*}
under Lebesgue integration.
\end{theorem}
\begin{proof}
Let $\eps>0$ and $r_N=\sup_{a\in\pH} f(a)$. Let
\begin{align*}
0 = r_0 < r_1 < \cdots < r_N
\end{align*}
such that $r_{i+1}-r_i<\eps$, $0\leq i\leq N-1$, and set
\begin{align*}
A_i = \set a{f(a) > r_i} = f^{-1}((r_i,\infty)) \in \SO,\ \ 0\leq i\leq N.
\end{align*}
Then $A_{i+1}\subs A_i$ and
\begin{align*}
A_i-A_{i+1} = \set a{r_i < f(a) \leq r_{i+1}} = f^{-1}((r_i,r_{i+1}]).
\end{align*}
Let 
\begin{align*}
\lb f &= \sum_{i=0}^{N-1} r_i\chrf{A_i-A_{i+1}} &
\ub f &= \sum_{i=0}^{N-1} r_{i+1}\chrf{A_i-A_{i+1}}.
\end{align*}
For $a\in A_i-A_{i+1}$,
\begin{align*}
\lb f(a) &= \sum_{i=0}^{N-1} r_i\chrf{A_i-A_{i+1}}(a) = r_i < f(a)\\
&\leq r_{i+1} = \sum_{i=0}^{N-1} r_{i+1}\chrf{A_i-A_{i+1}}(a) = \ub f(a),
\end{align*}
and as $a$ was arbitrary, $\lb f \leq f \leq \ub f$ pointwise. Thus
\begin{align*}
\int \lb f\,d\mu &\leq \int f\,d\mu \leq \int \ub f\,d\mu.
\end{align*}
Moreover,
\begin{align*}
\int \ub f\,d\mu - \int \lb f\,d\mu
&= \sum_{i=0}^{N-1} r_{i+1}\mu(A_i-A_{i+1})\\
&\qquad\qquad - \sum_{i=0}^{N-1} r_i\mu(A_i-A_{i+1})\\
&= \sum_{i=0}^{N-1} (r_{i+1}-r_i)\mu(A_i-A_{i+1})\\
&< \eps\cdot \sum_{i=0}^{N-1}\mu(A_i-A_{i+1}) = \eps\cdot\mu(\pH) = \eps,
\end{align*}
so the integral is approximated arbitrarily closely from above and below by the $\ub f$ and $\lb f$.
Finally, we argue that $\lb f$ and $\ub f$ are $\SO$-simple. Using the fact that $r_0=0$ and $A_N=\emptyset$ to reindex,
\begin{align*}
\lb f &= \sum_{i=0}^{N-1} r_i\chrf{A_i-A_{i+1}}
= \sum_{i=0}^{N-1} r_i\chrf{A_i} - \sum_{i=0}^{N-1} r_i\chrf{A_{i+1}}\\
&= \sum_{i=0}^{N-1} r_{i+1}\chrf{A_{i+1}} - \sum_{i=0}^{N-1} r_i\chrf{A_{i+1}}
= \sum_{i=0}^{N-1} (r_{i+1}-r_i)\chrf{A_{i+1}},\\
\ub f &= \sum_{i=0}^{N-1} r_{i+1}\chrf{A_i-A_{i+1}}
= \sum_{i=0}^{N-1} r_{i+1}\chrf{A_i} - \sum_{i=0}^{N-1} r_{i+1}\chrf{A_{i+1}}\\
&= \sum_{i=0}^{N-1} r_{i+1}\chrf{A_i} - \sum_{i=0}^{N-1} r_i\chrf{A_i}
= \sum_{i=0}^{N-1} (r_{i+1}-r_i)\chrf{A_i},
\end{align*}
and both functions are $\SO$-simple since all $A_i$ are in $\SO$.
\end{proof}

We can prove a stronger version of Theorem~\ref{thm:integration} that also works
for functions taking on infinite value.
A function $g$ is simple if it is a finite linear combination of indicator functions
of the form $g = \sum_{i=1}^{k} r_i \chrf{A_i}$, where $k \in \N$ and the $A_i$
are measurable. Let $S$ denote the set of all simple functions.
\begin{theorem}\label{thm:integration-strong}
Let $f : \pH \to [0,\infty]$ be Scott-continuous and let $\mu$ be a probability
measure. Then \[
  \int f\,d\mu = \sup_{\substack{g \in S_{\SO}\\g\leq f}} \int g\,d\mu
\]
\end{theorem}
\begin{proof}
It suffices to show that
\begin{equation}\label{eq:int-proof-1}
  \sup_{\substack{g \in S\\g\leq f}} \int g\,d\mu =
  \sup_{\substack{g \in S_{\SO}\\g\leq f}} \int g\,d\mu
\end{equation}
since the left side of this equation defines the integral of $f$. We trivially
have
\begin{equation}\label{eq:int-proof-2}
  \sup_{\substack{g \in S\\g\leq f}} \int g\,d\mu \geq
  \sup_{\substack{g \in S_{\SO}\\g\leq f}} \int g\,d\mu
\end{equation}
because $S_{\SO} \subseteq S$. To show the reverse inequality, let $g \in S$
with $g \leq f$ be arbitrary. We will show that there exists a family of
functions $g_\epsilon \in S_{\SO}$, $\epsilon>0$ with $g_\epsilon \leq f$ such that 
$\int g\, d\mu - \int g_{\epsilon}\,d\mu \leq \epsilon$.
Together with \eqref{eq:int-proof-2}, this proves \eqref{eq:int-proof-1} because
it implies that
\begin{align*}
  \sup_{\substack{g \in S\\g\leq f}} \int g\,d\mu \leq
  \sup_{\substack{g \in S\\g\leq f}} \sup_{\epsilon>0}\int g_\epsilon\,d\mu \leq
  \sup_{\substack{g \in S_{\SO}\\g\leq f}} \int g\,d\mu
\end{align*}
Let's turn to constructing the family of functions $g_\epsilon \in S_{\SO}$. Since $g$ is
simple, we may w.l.o.g.\ assume that it has the form $g = \sum_{i=1}^k r_i \chrf{A_i}$
with disjoint $A_i \in \BB$ and $r_1 < r_2 < \cdots < r_k$. Define
\begin{align*}
  r_0 &\defeq \epsilon\\
  B_{i,\epsilon} &\defeq \inv f ((r_i - \epsilon, \infty]) \in \SO\\
  \beta_{i} &\defeq r_i - r_{i-1}\\
  g_{\epsilon} &\defeq \sum_{i=1}^k \beta_i \cdot \chrf{B_{i,\epsilon}} \in S_{\SO}
\end{align*}
Then we have $g_{\epsilon} \leq f$ because for all $a \in \pH$  
\begin{align*}
 (\sum_{i=1}^k \beta_i \cdot \chrf{B_{i,\epsilon}})(a)
 &= \sum_{i=1}^k \beta_i \cdot \ind{a \in B_{i,\epsilon}} \\
 &= \sum_{i=1}^k (r_i - r_{i-1}) \cdot \ind{f(a) > r_i - \epsilon} \\
 &= \max \set{r_i}{ 1 \leq i \leq k \text{ and } f(a) > r_i - \epsilon} - r_0\\
 &< f(a)
\end{align*}
Moreover, we have that $g - g_\epsilon \leq \epsilon$ because
\begin{align*}
 (\sum_{i=1}^k \beta_i \cdot \chrf{B_{i,\epsilon}})(a)
 &= \max \set{r_i}{ 1 \leq i \leq k \text{ and } f(a) > r_i - \epsilon} - r_0\\
 &\geq \max \set{r_i}{ 1 \leq i \leq k \text{ and } f(a) \geq r_i} - \epsilon\\
 &\geq \max \set{r_i}{ 1 \leq i \leq k \text{ and } g(a) = r_i} - \epsilon\\
 &= g(a) - \epsilon
\end{align*}
Thus it follows that
\begin{align*}
  \int g\,d\mu - \int g_\epsilon\,d\mu = \int(g-g_\epsilon)d\mu
  \leq \int \epsilon\, d\mu = \epsilon
\end{align*}
\end{proof}

\begin{proof}[Proof of Theorem~\ref{thm:int-continuous}]
(i) We prove the result first for $\SO$-simple functions. If $\mu\sqleq\nu$, then
for any $\SO$-simple function $g = \sum_{A} r_A\chrf{A}$,
\begin{align*}
\int g\,d\mu &= \int \sum_{A} r_A\chrf{A}\,d\mu
= \sum_{A} r_A\mu(A)\\
&\leq \sum_{A} r_A\nu(A)
= \int \sum_{A} r_A\chrf{A}\,d\nu
= \int g\,d\nu.
\end{align*}
Thus the map \eqref{eq:int-continuous} is monotone.
If $\DD$ is a directed set of measures with respect to $\sqleq$, then
\begin{align*}
\int g\,d(\tbigsqcup\DD) &= \int\sum_{A} r_A\chrf{A}\,d(\tbigsqcup\DD)
= \sum_{A} r_A(\tbigsqcup\DD)(A)\\
&= \sup_{\mu\in\DD}\sum_{A} r_A\mu(A)
= \sup_{\mu\in\DD}\int\sum_{A} r_A\chrf{A}\,d\mu\\
&= \sup_{\mu\in\DD}\int g\,d\mu.
\end{align*}

Now consider an arbitrary Scott-continuous function $f:\pH\to[0,\infty]$. Let $S_{\SO}$ be the family of $\SO$-simple functions. By Theorem \ref{thm:integration-strong}, if $\mu\sqleq\nu$, we have
\begin{align*}
\int f\,d\mu &= \sup_{\substack{g\in S_{\SO}\\g\leq f}} \int g\,d\mu
\leq \sup_{\substack{g\in S_{\SO}\\g\leq f}} \int g\,d\nu
= \int f\,d\nu,
\end{align*}
and if $\DD$ is a directed set of measures with respect to $\sqleq$, then
\begin{align*}
\int f\,d(\tbigsqcup\DD) &= \sup_{\substack{g\in S_{\SO}\\g\leq f}} \int g\,d(\tbigsqcup\DD)
= \sup_{\substack{g\in S_{\SO}\\g\leq f}} \sup_{\mu\in\DD}\int g\,d\mu\\
&= \sup_{\mu\in\DD}\sup_{\substack{g\in S_{\SO}\\g\leq f}} \int g\,d\mu
= \sup_{\mu\in\DD}\int f\,d\mu.
\end{align*}

(ii) This just the monotone convergence theorem for Lebesgue Integration.
\end{proof}

\subsection{Continuous Operations on Measures}
\label{apx:continuousmeasures}

In this section we show that certain operations on measures are continuous. These
properties will be lifted to kernels as required.

\begin{lemma}
\label{lem:compactopen}
For any probability measure $\mu$ on an algebraic DCPO and open set $B$, the value $\mu(B)$ is approximated arbitrarily closely from below by $\mu(C)$ for compact-open sets $C$.
\end{lemma}
\begin{proof}
Since the sets $\up{\{a\}}$ for finite $a$ form a base for the topology, and every compact-open set is a finite union of such sets, the set $\KK(B)$ of compact-open subsets of $B$ is a directed set whose union is $B$. Then
\begin{align*}
\mu(B) = \mu(\tbigcup\KK(B)) = \sup\set{\mu(C)}{C\in\KK(B)}.
\end{align*}
\end{proof}

\begin{lemma}
\label{lem:prod-measure-continuous}
The product operator on measures in algebraic DCPOs is Scott-continuous in each argument.
\end{lemma}
\begin{proof}
The difficult part of the argument is monotonicity. Once we have that, then for any $B,C\in \SO$, we have $(\mu\times\nu)(B\times C) = \mu(B)\cdot\nu(C)$. Thus for any directed set $D$ of measures,
\begin{align*}
\lefteqn{(\tbigsqcup D\times\nu)(B\times C)}\qquad\\
&= (\tbigsqcup D)(B)\cdot\nu(C)
= (\sup_{\mu\in D}\mu(B))\cdot\nu(C)\\
&= \sup_{\mu\in D}(\mu(B)\cdot\nu(C))
= \sup_{\mu\in D}((\mu\times\nu)(B\times C))\\
&= (\tbigsqcup_{\mu\in D}(\mu\times\nu))(B\times C).
\end{align*}
By Theorem \ref{thm:product}, the sets $B\times C$ for $B,C\in \SO$ form a basis for the Scott topology on the product space $\pH\times\pH$, thus $\tbigsqcup D\times\nu = \tbigsqcup_{\mu\in D}(\mu\times\nu)$.

To show monotonicity, we use approximability by compact-open sets (Lemma \ref{lem:compactopen}). We wish to show that if $\mu_1\sqleq\mu_2$, then $\mu_1\times\nu\sqleq\mu_2\times\nu$. By Lemma \ref{lem:compactopen}, it suffices to show that
\begin{align*}
(\mu_1\times\nu)(\bigcup_n B_n\times C_n) \leq (\mu_2\times\nu)(\bigcup_n B_n\times C_n),
\end{align*}
where the index $n$ ranges over a finite set, and $B_n$ and $C_n$ are open sets of the component spaces. Consider the collection of all atoms $A$ of the Boolean algebra generated by the $C_n$. For each such atom $A$, let
\begin{align*}
N(A) &= \set n{\text{$C_n$ occurs positively in $A$}}.
\end{align*}
Then
\begin{align*}
\bigcup_n B_n\times C_n &= \bigcup_{A} (\bigcup_{n\in N(A)} B_n)\times A.
\end{align*}
The right-hand side is a disjoint union, since the $A$ are pairwise disjoint. Then
\begin{align*}
(\mu_1\times\nu)(\bigcup_n B_n\times C_n)
&= (\mu_1\times\nu)(\bigcup_{A} (\bigcup_{n\in N(A)} B_n)\times A)\\
&= \sum_A (\mu_1\times\nu)((\bigcup_{n\in N(A)} B_n)\times A)\\
&= \sum_A \mu_1(\bigcup_{n\in N(A)} B_n)\cdot\nu(A)\\
&\leq \sum_A \mu_2(\bigcup_{n\in N(A)} B_n)\cdot\nu(A)\\
&= (\mu_2\times\nu)(\bigcup_n B_n\times C_n).
\end{align*}
\end{proof}

Let $S$ and $T$ be measurable spaces and $f:S\to T$ a measurable function. For a measure $\mu$ on $S$, the \emph{push-forward measure} $f_*(\mu)$ is the measure $\mu\circ f^{-1}$ on $T$.

\begin{lemma}
\label{lem:pushforward}
If $f:(\pH)^\kappa\to\pH$ is Scott-continuous with respect to the subset order, then the push-forward operator $f_*:\MM((\pH)^\kappa)\to\MM(\pH)$ is Scott-continuous with respect to $\sqleq$.
\end{lemma}
\begin{proof}
Let $\mu,\nu\in\MM((\pH)^\kappa)$, $\mu\sqleq\nu$. If $B\in \SO$, then $f^{-1}(B)$ is Scott-open in $(\pH)^\kappa$, so $f_*(\mu)(B) = \mu(f^{-1}(B))\leq \nu(f^{-1}(B)) = f_*(\nu)(B)$. As $B\in \SO$ was arbitrary, $f_*(\mu) \sqleq f_*(\nu)$. Similarly, if $D$ is any $\sqleq$-directed set in $\MM((\pH)^\kappa)$, then so is $\set{f_*(\mu)}{\mu\in D}$, and
\begin{align*}
f_*(\tbigsqcup D)(B) &= (\tbigsqcup D)(f^{-1}(B)) = \sup_{\mu\in D}\mu(f^{-1}(B))\\
&= \sup_{\mu\in D}f_*(\mu)(B) = (\tbigsqcup_{\mu\in D}f_*(\mu))(B)
\end{align*}
for any $B\in \SO$, thus $f_*(\tbigsqcup D) = \tbigsqcup_{\mu\in D}f_*(\mu)$.
\end{proof}

\begin{lemma}
\label{lem:parcomp-measure-continuous}
Parallel composition of measures {\upshape(}$\pcomp${\upshape)} is Scott-continuous in each argument.
\end{lemma}
\begin{proof}
By definition, $\mu\pcomp\nu = (\mu\times\nu)\cmp{\tbigcup}^{-1}$, where $\tbigcup:\pH\times\pH\to\pH$ is the set union operator. The set union operator is easily shown to be continuous with respect to the Scott topologies on $\pH\times\pH$ and the $\pH$. By Lemma \ref{lem:pushforward}, the push-forward operator with respect to union is Scott-continuous with respect to $\sqleq$. By Lemma \ref{lem:prod-measure-continuous}, the product operator is Scott-continuous in each argument with respect to $\sqleq$. The operator $\pcomp$ is the composition of these two Scott continuous operators, therefore is itself Scott-continuous.
\end{proof}

\subsection{Continuous Kernels}
\label{apx:continuouskernels}

\begin{lemma}
\label{lem:det-preserves}
The deterministic kernel associated with any Scott-continuous function $f:D\to E$ is a continuous kernel.
\end{lemma}
\begin{proof}
Recall from \cite{\pnkpaper} that deterministic kernels are those whose output measures are Dirac measures
(point masses). Any measurable function $f:D\to E$ uniquely determines
a deterministic kernel $P_f$ such that $P_f(a,-) = \dirac{f(a)}$ 
(or equivalently, $P = \unitop \circ f$)
and vice versa (this was shown in \cite{\pnkpaper} for $D=E=\pH$). We show that if in addition $f$ is Scott-continuous, then the kernel $P_f$ is continuous.

Let $f:D\to E$ be Scott-continuous. For any open $B$, if $a\sqleq b$, then $f(a)\sqleq f(b)$ since $f$ is monotone. Since $B$ is up-closed, if $f(a)\in B$, then $f(b)\in B$. Thus
\begin{align*}
P_f(a,B) = [f(a)\in B] \leq [f(b)\in B] = P_f(b,B).
\end{align*}
If $A\subs D$ is a directed set, then $f(\bigsqcup A) = \bigsqcup_{a\in A} f(a)$. Since $B$ is open,
$\bigsqcup_{a\in A} f(a)\in B$ iff there exists $a\in A$ such that $f(a)\in B$. Then
\begin{align*}
P_f(\tbigsqcup A,B) &= [f(\tbigsqcup A)\in B] = [\tbigsqcup_{a\in A} f(a)\in B]\\
&= \sup_{a\in A}[f(a)\in B] = \sup_{a\in A}P_f(a,B).
\end{align*}
\end{proof}

\begin{lemma}
\label{lem:atomic-programs-cont}
All atomic \PNK\ programs (including predicates) denote deterministic and Scott-continuous kernels.
\end{lemma}
\begin{proof}
By Lemma~\ref{lem:nk-prim-charact}, all atomic programs denote kernels of the form
$a\mapsto \unit{\set{f(\h)}{\h\in a}}$, where $f$ is a partial function $H\pfun H$.
Hence they are deterministic. Using Lemma~\ref{lem:det-preserves}, we see that they are also Scott-continuous:
\begin{itemize}
\item
If $a\subs b$, then $\set{f(\h)}{\h\in a}\subs\set{f(\h)}{\h\in b}$; and
\item
If $D\subs\pH$ is a directed set, then $\set{f(\h)}{\h\in\tbigcup D}=\bigcup_{a\in D}\set{f(\h)}{\h\in a}$.
\end{itemize}
\end{proof}

\begin{lemma}
\label{lem:int-preserves}
Let $P$ be a continuous Markov kernel and $f:\pH\to\R_+$ a Scott-continuous function. Then the map
\begin{align}
a \mapsto \int_{c\in\pH} f(c)\cdot P(a,dc)\label{eq:int-preserves}
\end{align}
is Scott-continuous.
\end{lemma}
\begin{proof}
The map \eqref{eq:int-preserves} is the composition of the maps
\begin{align*}
a &\mapsto P(a,-) & P(a,-) \mapsto \int_{c\in\pH} P(a,dc)\cdot f(c),
\end{align*}
which are Scott-continuous by Lemmas \ref{lem:curry-continuous} and \ref{thm:int-continuous}, respectively, and the composition of Scott-continuous maps is Scott-continuous.
\end{proof}

\begin{lemma}
\label{lem:prod-preserves}
Product preserves continuity of Markov kernels: If $P$ and $Q$ are continuous, then so is $P\times Q$.
\end{lemma}
\begin{proof}
We wish to show that if $a\subs b$, then $(P\times Q)(a,-)\sqleq(P\times Q)(b,-)$, and if $A$ is a directed subset of $\pH$, then $(P\times Q)(\bigcup A) = \sup_{a\in A} (P\times Q)(a,-)$. For the first statement, using Lemma \ref{lem:prod-measure-continuous} twice,
\begin{align*}
(P\times Q)(a,-) &= P(a,-)\times Q(a,-)
\sqleq P(b,-)\times Q(a,-)\\
&\sqleq P(b,-)\times Q(b,-)
= (P\times Q)(b,-).
\end{align*}
For the second statement, for $A$ a directed subset of $\pH$,
\begin{align*}
(P\times Q)(\tbigsqcup A,-) &= P(\tbigsqcup A,-)\times Q(\tbigsqcup A,-)\\
&= (\tbigsqcup_{a\in A}P(a,-))\times (\tbigsqcup_{b\in A}Q(b,-))\\
&= \tbigsqcup_{a\in A}\tbigsqcup_{b\in A}P(a,-)\times Q(b,-)\\
&= \tbigsqcup_{a\in A}P(a,-)\times Q(a,-)\\
&= \tbigsqcup_{a\in A}(P\times Q)(a,-).
\end{align*}
\end{proof}

\begin{lemma}
\label{lem:seqcomp-preserves}
Sequential composition preserves continuity of Markov kernels: If $P$ and $Q$ are continuous, then so is $P\cmp Q$.
\end{lemma}
\begin{proof}
We have
\begin{align*}
(P\cmp Q)(a,A) &= \int_{c\in\pH} P(a,dc)\cdot Q(c,A).
\end{align*}
Since $Q$ is a continuous kernel, it is Scott-continuous in its first argument, thus so is $P\cmp Q$ by Lemma \ref{lem:int-preserves}.
\end{proof}

\begin{lemma}
\label{lem:parcomp-preserves}
Parallel composition preserves continuity of Markov kernels: If $P$ and $Q$ are continuous, then so is $P\pcomp Q$.
\end{lemma}
\begin{proof}
Suppose $P$ and $Q$ are continuous.
By definition, $P\pcomp Q = (P\times Q)\cmp\tbigcup$. By Lemma \ref{lem:prod-preserves}, $P\times Q$ is continuous, and $\tbigcup:\pH\times\pH\to\pH$ is continuous. Thus their composition is continuous by Lemma \ref{lem:seqcomp-preserves}.
\end{proof}

\begin{lemma}
\label{lem:choice-preserves}
The probabilistic choice operator {\upshape($\opr$)} preserves continuity of kernels.
\end{lemma}
\begin{proof}
If $P$ and $Q$ are continuous, then $P\opr Q = rP + (1-r)Q$. If $a\subs b$, then
\begin{align*}
(P\opr Q)(a,-) &= rP(a,-) + (1-r)Q(a,-)\\
&\leq rP(b,-) + (1-r)Q(b,-)\\
&= (P\opr Q)(b,-).
\end{align*}
If $A\subs\pH$ is a directed set, then
\begin{align*}
(P\opr Q)(\tbigcup A,-) &= rP(\tbigcup A,-) + (1-r)Q(\tbigcup A,-)\\
&= \tbigsqcup_{a\in A} (rP(a,-) + (1-r)Q(a,-))\\
&= \tbigsqcup_{a\in A} (P\opr Q)(a,-).
\end{align*}
\end{proof}

\begin{lemma}
\label{lem:star-preserves}
The iteration operator {\upshape(*)} preserves continuity of kernels.
\end{lemma}
\begin{proof}
Suppose $P$ is continuous. It follows inductively using Lemmas \ref{lem:parcomp-preserves} and \ref{lem:seqcomp-preserves} that $\pp n$ is continuous. Since $P\star = \bigsqcup_n \pp n$ and since the supremum of a directed set of continuous kernels is continuous by Theorem \ref{thm:kernelDCPO}, $P\star$ is continuous.
\end{proof}

\begin{proof}[Proof of Theorem \ref{thm:continuouskernels}]
The result follows from Lemmas \ref{lem:det-preserves}, \ref{lem:int-preserves}, \ref{lem:prod-preserves}, \ref{lem:seqcomp-preserves}, \ref{lem:parcomp-preserves}, \ref{lem:choice-preserves}, and \ref{lem:star-preserves}.
\end{proof}

\begin{proof}[Proof of Corollary \ref{cor:continuouskernels}]
This follows from Theorem \ref{thm:continuouskernels}. All primitive programs are deterministic, thus
give continuous kernels, and continuity is preserved by all the program operators.
\end{proof}

\subsection{Continuous Operations on Kernels}
\label{apx:continuousopsonkernels}

\begin{lemma}
\label{lem:prod-continuous}
The product operation on kernels {\upshape(}$\times${\upshape)} is Scott-continuous in each argument.
\end{lemma}
\begin{proof}
We use Lemma \ref{lem:prod-measure-continuous}. If $P_1\sqleq P_2$, then for all $a\in\pH$,
\begin{align*}
(P_1\times Q)(a,-) &= P_1(a,-)\times Q(a,-)\\
&\sqleq P_2(a,-)\times Q(a,-) = (P_2\times Q)(a,-).
\end{align*}
Since $a$ was arbitrary, $P_1\times Q\sqleq P_2\times Q$. For a directed set $\DD$ of kernels,
\begin{align*}
(\tbigsqcup\DD\times Q)(a,-) &= (\tbigsqcup\DD)(a,-)\times Q(a,-)\\
&= \tbigsqcup_{P\in\DD}P(a,-)\times Q(a,-)\\
&= \tbigsqcup_{P\in\DD}(P(a,-)\times Q(a,-))\\
&= \tbigsqcup_{P\in\DD}(P\times Q)(a,-)\\
&= (\tbigsqcup_{P\in\DD}(P\times Q))(a,-).
\end{align*}
Since $a$ was arbitrary, $\tbigsqcup\DD\times Q = \tbigsqcup_{P\in\DD}(P\times Q)$.
\end{proof}

\begin{lemma}
\label{lem:parcomp-continuous}
Parallel composition of kernels {\upshape(}$\pcomp${\upshape)} is Scott-continuous in each argument.
\end{lemma}
\begin{proof}
By definition, $P\pcomp Q = (P\times Q)\cmp\tbigcup$. By Lemmas \ref{lem:prod-continuous} and \ref{lem:seqcomp-continuous}, the product operation and sequential composition are continuous in both arguments, thus their composition is.
\end{proof}

\begin{lemma}
\label{lem:curry-continuous}
Let $P$ be a continuous Markov kernel. The map $\curry P$ is Scott-continuous with respect to the subset order on $\pH$ and the order $\sqleq$ on $\MM(\pH)$.
\end{lemma}
\begin{proof}
We have $(\curry P)(a) = P(a,-)$.
Since $P$ is monotone in its first argument, if $a\subs b$ and $B\in \SO$, then $P(a,B) \leq P(b,B)$. As $B\in \SO$ was arbitrary,
\begin{align*}
(\curry P)(a) &= P(a,-) \sqleq P(b,-) = (\curry P)(b). 
\end{align*}
This shows that $\curry P$ is monotone.

Let $D\subs\pH$ be a directed set. By the monotonicity of $\curry P$, so is the set 
$\set{(\curry P)(a)}{a\in D}$. Then for any $B\in \SO$,
\begin{align*}
(\curry P)(\tbigcup D)(B) &= P(\tbigcup D,B)
= \sup_{a\in D}P(a,B)\\
&= \sup_{a\in D}(\curry P)(a)(B)\\
&= (\tbigsqcup_{a\in D}(\curry P)(a))(B),
\end{align*}
thus $(\curry P)(\tbigcup D) = \tbigsqcup_{a\in D}(\curry P)(a)$.
\end{proof}

\begin{lemma}
\label{lem:seqcomp-continuous}
Sequential composition of kernels is Scott-continuous in each argument.
\end{lemma}
\begin{proof}
To show that $\cmp$ is continuous in its first argument, we wish to show that if $P_1,P_2,Q$ are any continuous kernels with $P_1\sqleq P_2$, and if $\DD$ is any directed set of continuous kernels, then
\begin{align*}
P_1\cmp Q &\leq P_2\cmp Q
&
(\tbigsqcup\DD)\cmp Q &= \tbigsqcup_{P\in\DD} (P\cmp Q).
\end{align*}
We must show that for all $a\in\pH$ and $B\SO$,
\begin{align*}
\int_c P_1(a,dc)\cdot Q(c,B) &\leq \int_c P_2(a,dc)\cdot Q(c,B)\\
\int_c (\tbigsqcup\DD)(a,dc)\cdot Q(c,B) &= \sup_{P\in\DD} \int_c P(a,dc)\cdot Q(c,B).
\end{align*}
By Lemma \ref{lem:kernelorder}, for all $a\in\pH$, $P_1(a,-)\sqleq P_2(a,-)$ and $(\tbigsqcup\DD)(a,-) = \tbigsqcup_{P\in\DD} P(a,-)$, and $Q(-,B)$ is a Scott-continuous function by assumption. The result follows from Lemma \ref{thm:int-continuous}(i).

The argument that $\cmp$ is continuous in its second argument is similar, using Lemma \ref{thm:int-continuous}(ii). We wish to show that if $P,Q_1,Q_2$ are any continuous kernels with $Q_1\sqleq Q_2$, and if $\DD$ is any directed set of continuous kernels, then
\begin{align*}
P\cmp Q_1 &\leq P\cmp Q_2
&
P\cmp\tbigsqcup\DD &= \tbigsqcup_{Q\in\DD} (P\cmp Q).
\end{align*}
We must show that for all $a\in\pH$ and $B\in \SO$,
\begin{align*}
\int_c P(a,dc)\cdot Q_1(c,B) &\leq \int_c P(a,dc)\cdot Q_2(c,B)\\
\int_c P(a,dc)\cdot (\tbigsqcup\DD)(c,B) &= \sup_{Q\in\DD} \int_c P(a,dc)\cdot Q(c,B).
\end{align*}
By Lemma \ref{lem:kernelorder}, for all $B\in \SO$, $Q_1(-,B)\sqleq Q_2(-,B)$ and $(\tbigsqcup\DD)(-,B) = \tbigsqcup_{Q\in\DD} Q(-,B)$. The result follows from Lemma \ref{thm:int-continuous}(ii).
\end{proof}

\begin{lemma}
\label{lem:choice-continuous}
The probabilistic choice operator applied to kernels {\upshape($\opr$)} is continuous in each argument.
\end{lemma}
\begin{proof}
If $P$ and $Q$ are continuous, then $P\opr Q = rP + (1-r)Q$. If $P_1\sqleq P_2$, then for any $a\in\pH$ and $B\in \SO$,
\begin{align*}
(P_1\opr Q)(a,B) &= rP_1(a,B) + (1-r)Q(a,B)\\
&\leq rP_2(a,B) + (1-r)Q(a,B)\\
&= (P_2\opr Q)(a,B),
\end{align*}
so $P_1\opr Q\sqleq P_2\opr Q$. If $\DD$ is a directed set of kernels and $B\SO$, then
\begin{align*}
(\tbigsqcup\DD\opr Q)(a,B) &= r(\tbigsqcup\DD)(a,B) + (1-r)Q(a,B)\\
&= \sup_{P\in\DD} (rP(a,B) + (1-r)Q(a,B))\\
&= \sup_{P\in\DD} (P\opr Q)(a,B).
\end{align*}
\end{proof}

\begin{lemma}
\label{lem:star-approx-monotone-p}
If $P\sqleq Q$ then $\ksn{P}{n} \sqleq \ksn{Q}{n}$.
\end{lemma}
\begin{proof}
By induction on $n \in \N$. The claim is trivial for $n=0$. For $n>0$,
we assume that $\ksn{P}{n-1} \sqleq \ksn{Q}{n-1}$ and deduce
\begin{align*}
\ksn{P}{n} = \skp \pcomp P \cmp \ksn{P}{n-1}
\sqleq \skp \pcomp Q \cmp \ksn{Q}{n-1} =\ksn{Q}{n}
\end{align*}
by monotonicity of sequential and parallel composition (Lemmas~\ref{lem:seqcomp-continuous} and \ref{lem:parcomp-continuous}, respectively).
\end{proof}

\begin{lemma}
\label{lem:star-approx-monotone-n}
If $m\leq n$ then $\pp m \sqleq \pp n$.
\end{lemma}
\begin{proof}
We have $\pp 0 \sqleq \pp 1$ by Lemmas~\ref{lem:upperbound} and \ref{lem:kernelorder}.
Proceeding by induction using Lemma \ref{lem:star-approx-monotone-p}, we have $\pp n \sqleq \pp{n+1}$ for all $n$. The result follows from transitivity.
\end{proof}

\begin{lemma}
\label{lem:star-continuous}
The iteration operator applied to kernels {\upshape(*)} is continuous.
\end{lemma}
\begin{proof}
It is a straightforward consequence of Lemma~\ref{lem:star-approx-monotone-p} and Theorem \ref{thm:star} that if $P\sqleq Q$, then $P\star\sqleq Q\star$. Now let $\DD$ be a directed set of kernels. It follows by induction using Lemmas \ref{lem:parcomp-continuous} and \ref{lem:seqcomp-continuous} that the operator $P\mapsto\pp n$ is continuous, thus
\begin{align*}
(\tbigsqcup\DD)\star &= \tbigsqcup_n \ksn{(\tbigsqcup\DD)}n
= \tbigsqcup_n\tbigsqcup_{P\in\DD}\ksn Pn\\
&= \tbigsqcup_{P\in\DD}\tbigsqcup_n\ksn Pn
= \tbigsqcup_{P\in\DD}P\star.
\end{align*}
\end{proof}

\begin{proof}[Proof of Theorem \ref{thm:continuouspnkops}]
The result follows from Lemmas \ref{lem:prod-continuous}, \ref{lem:parcomp-continuous}, \ref{lem:curry-continuous}, \ref{lem:seqcomp-continuous}, \ref{lem:choice-continuous}, and \ref{lem:star-continuous}. 
\end{proof}

\subsection{Iteration as Least Fixpoint}
\label{apx:relation}

In this section we show that the semantics of iteration presented in \cite{\pnkpaper}, defined in terms of an infinite process, coincides with the least fixpoint semantics presented here.

In this section, we use the notation $P\star$ refers to the semantics of \cite{\pnkpaper}.
For the iterate introduced here, we use $\bigsqcup_n \pp n$.

Recall from \cite{\pnkpaper} the approximants
\begin{align*}
\pp 0 &= \skp & \pp{m+1} &= \skp \pcomp P\cmp\pp m.
\end{align*}
It was shown in \cite{\pnkpaper} that for any $c\in\pH$, the measures $\pp m(c,-)$ converge weakly to $P\star(c,-)$; that is, for any bounded (Cantor-)continuous real-valued function $f$ on $\pH$, the expected values of $f$ with respect to the measures $\pp m(c,-)$ converge to the expected value of $f$ with respect to $P\star(c,-)$:
\begin{align*}
\lim_{m\to\infty}\int_{a\in\pH} f(a)\cdot\pp m(c,da)\ =\ \int_{a\in\pH} f(a)\cdot P\star(c,da).
\end{align*}

\begin{theorem}
\label{thm:Steffensproof}
The kernel $Q = \bigsqcup_{n \in \N} P^{(n)}$ is the unique fixpoint of
$(\lambda Q.~ \skp \pcomp P \cmp Q)$ such that $P^{(n)}(a)$ weakly converges
to $Q(a)$ (with respect to the Cantor topology) for all $a \in \pH$. 
\end{theorem}
\begin{proof}
Let $P\star$ denote any fixpoint of $(\lambda Q.~ \skp \pcomp P \cmp Q)$ such
that the measure $\mu_n = P^{(n)}(a)$  weakly converges to the measure $\mu = P\star(a)$,
\ie such that for all (Cantor-)continuous bounded functions $f:\pH \to \R$ \[
  \lim_{n \to \infty} \int f d\mu_n =
  \int f d\mu
\] for all $a \in \pH$. Let $\nu = Q(a)$.
Fix an arbitrary Scott-open set $V$. Since $\pH$ is a Polish space
under the Cantor topology, there exists an increasing chain of compact sets\[
  C_1 \subseteq C_2 \subseteq \dots \subseteq V
  \quad \text{ such that } \quad
  \sup_{n \in \N} \mu(C_n) = \mu(V).
\]
By Urysohn's lemma (see \cite{KolmogorovFomin70,Rao87}), there exist continuous
functions $f_n:\pH\to[0,1]$ such that $f_n(x) = 1$ for $x\in C_n$
and $f(x) = 0$ for $x\in \setcompl V$.
We thus have
\begin{align*}
\mu(C_n) 
&= \int \chrf{C_n} d\mu\\
&\leq \int f_n d\mu                         &&\text{by monotonicity of $\int$}\\
&= \lim_{m \to \infty} \int f_n d\mu_m           &&\text{by weak convergence}\\
&\leq \lim_{m \to \infty} \int \chrf{V} d\mu_m   &&\text{by monotonicity of $\int$}\\
&= \lim_{m \to \infty} \mu_m(V)           \\
&= \nu(V)                                   &&\text{by pointwise convergence on $\SO$}
\end{align*}
Taking the supremum over $n$, we get that $\mu(V) \leq \nu(V)$.
Since $\nu$ is the $\sqleq$-\emph{least} fixpoint, the measures must therefore agree
on $V$, which implies that they are equal by Theorem~\ref{thm:extension}.
Thus, any fixpoint of $(\lambda Q.~ \skp \pcomp P \cmp Q)$ with the weak convergence
property must be equal to $Q$. But the fixpoint $P\star$ defined in previous work
\emph{does} enjoy the weak convergence property, and therefore so does $Q=P\star$.
\end{proof}

\begin{proof}[Proof of Lemma~\ref{lem:separation}]
Let $A$ be a Borel set. Since we are in a Polish space, $\mu(A)$ is approximated arbitrarily closely from below by $\mu(C)$ for compact sets $C\subs A$ and from above by $\mu(U)$ for open sets $U\supseteq A$. By Urysohn's lemma (see \cite{KolmogorovFomin70,Rao87}), there exists a continuous function $f:D\to[0,1]$ such that $f(a) = 1$ for all $a\in C$ and $f(a) = 0$ for all $a\not\in U$. We thus have
\begin{align*}
\mu(C) &= \int_{a\in C} f(a)\cdot\mu(da) \leq \int_{a\in D} f(a)\cdot\mu(da)\\
&= \int_{a\in U} f(a)\cdot\mu(da) \leq \mu(U),\\
\mu(C) &\leq \mu(A) \leq \mu(U),
\end{align*}
thus
\begin{align*}
\left|\mu(A) - \int_{a\in D} f(a)\cdot\mu(da)\right| &\leq \mu(U)-\mu(C),
\end{align*}
and the right-hand side can be made arbitrarily small.
\end{proof}

By Lemma \ref{lem:separation}, if $P,Q$ are two Markov kernels and 
\begin{align*}
\int_{a\in\pH} f(a)\cdot P(c,da) &= \int_{a\in\pH} f(a)\cdot Q(c,da)
\end{align*}
for all Cantor-continuous $f:\pH\to[0,1]$, then $P(c,-)=Q(c,-)$. If this holds for all $c\in\pH$, then $P=Q$.

\begin{proof}[Proof of Theorem \ref{thm:approx}]
Let $\eps>0$. Since all continuous functions on a compact space are uniformly continuous, for sufficiently large finite $b$ and for all $a\subs b$, the value of $f$ does not vary by more than $\eps$ on $\atm ab$; that is, $\sup_{c\in\atm ab} f(c) - \inf_{c\in\atm ab} f(c) < \eps$. Then for any $\mu$,
\begin{align*}
& \int_{c\in\atm ab} f(c)\cdot\mu(dc) - \int_{c\in\atm ab} \inf_{c\in\atm ab}f(c)\cdot\mu(dc)\\
&\leq \int_{c\in\atm ab} (\sup_{c\in\atm ab}f(c) - \inf_{c\in\atm ab}f(c))\cdot\mu(dc)
< \eps\cdot\mu(\atm ab).
\end{align*}
Moreover,
\begin{align*}
(\tbigsqcup A)(\atm ab) &= \sum_{a\subs c\subs b} (-1)^{\len{c-a}}({\tbigsqcup A})(B_c)\\
&= \sum_{a\subs c\subs b} (-1)^{\len{c-a}}\sup_{\mu\in A}\mu(B_c)\\
&= \lim_{\mu\in A}\sum_{a\subs c\subs b} (-1)^{\len{c-a}}\mu(B_c)
= \lim_{\mu\in A}\mu(\atm ab),
\end{align*}
so for sufficiently large $\mu\in A$, $\mu(\atm ab)$ does not differ from $(\bigsqcup A)(\atm ab)$ by more than $\eps\cdot 2^{-\len b}$. Then for any constant $r\in[0,1]$,
\begin{align*}
\lefteqn{\left|\int_{c\in\atm ab} r\cdot({\tbigsqcup A})(dc) - \int_{c\in\atm ab} r\cdot\mu(dc)\right|}\qquad\\
&= r\cdot\left|(\tbigsqcup A)(\atm ab) - \mu(\atm ab)\right|\\
&\leq \left|(\tbigsqcup A)(\atm ab) - \mu(\atm ab)\right| < \eps\cdot 2^{-\len b}.
\end{align*}
Combining these observations,
\begin{align*}
& \left|\int_{c\in\pH} f(c)\cdot({\tbigsqcup A})(dc) - \int_{c\in\pH} f(c)\cdot\mu(dc)\right|\\
&= \left|\sum_{a\subs b}\int_{c\in\atm ab} f(c)\cdot({\tbigsqcup A})(dc) - \sum_{a\subs b}\int_{c\in\atm ab} f(c)\cdot\mu(dc)\right|\\
&\leq \sum_{a\subs b}\left(\left|\int_{c\in\atm ab} f(c)\cdot({\tbigsqcup A})(dc) - \int_{c\in\atm ab} \inf_{c\in\atm ab}f(c)\cdot({\tbigsqcup A})(dc)\right|\right.\\
&\qquad + \left|\int_{c\in\atm ab} \inf_{c\in\atm ab}f(c)\cdot({\tbigsqcup A})(dc) - \int_{c\in\atm ab} \inf_{c\in\atm ab}f(c)\cdot\mu(dc)\right|\\
&\qquad + \left.\left|\int_{c\in\atm ab} \inf_{c\in\atm ab}f(c)\cdot\mu(dc) - \int_{c\in\atm ab} f(c)\cdot\mu(dc)\right|\right)\\
&\leq \sum_{a\subs b}\left(\eps\cdot({\tbigsqcup A})(\atm ab) + \eps\cdot 2^{-\len b} + \eps\cdot\mu(\atm ab)\right)\\
&= 3\eps.
\end{align*}
As $\eps>0$ was arbitrary,
\begin{align*}
\lim_{\mu\in A}\int_{c\in\pH} f(c)\cdot\mu(dc) &= \int_{c\in\pH} f(c)\cdot({\tbigsqcup A})(dc).
\end{align*}
\end{proof}

\begin{proof}[Proof of Theorem \ref{thm:star}]
Consider the continuous transformation
\begin{align*}
T_P(Q) &\defeq \skp \pcomp P\cmp Q
\end{align*}
on the DCPO of continuous Markov kernels. The continuity of $T_P$ follows from Lemmas \ref{lem:parcomp-continuous} and \ref{lem:seqcomp-continuous}.
The bottom element $\bot$ is $\drp$ in this space, and
\begin{align*}
T_P(\bot) &= \skp = \pp 0
&
T_P(\pp n) &= \skp \pcomp P\cmp \pp n = \pp{n+1},
\end{align*}
thus $T_P^{n+1}(\bot) = \pp n$, so $\bigsqcup T_P^n(\bot) = \bigsqcup_n\pp n$, and this is the least fixpoint of $T_P$. As shown in \cite{\pnkpaper}, $P\oldstar$ is also a fixpoint of $T_P$, so it remains to show that $P\oldstar = \bigsqcup_n\pp n$.

Let $c\in\pH$. As shown in \cite{\pnkpaper}, the measures $\pp n(c,-)$ converge weakly to $P\oldstar(c,-)$; that is, for any Cantor-continuous function $f:\pH\to[0,1]$, the expected values of $f$ relative to $\pp n$ converge to the expected value of $f$ relative to $P\oldstar$:
\begin{align*}
\lim_n\int f(a)\cdot\pp n(c,da) = \int f(a)\cdot P\oldstar(c,da).
\end{align*}
But by Theorem \ref{thm:approx}, we also have
\begin{align*}
\lim_n\int f(a)\cdot\pp n(c,da) = \int f(a)\cdot(\bigsqcup_n\pp n)(c,da),
\end{align*}
thus
\begin{align*}
\int f(a)\cdot P\oldstar(c,da) &= \int f(a)\cdot(\tbigsqcup_n\pp n)(c,da).
\end{align*}
As $f$ was arbitrary, we have $P\oldstar(c,-) = (\tbigsqcup_n\pp n)(c,-)$ by Lemma \ref{lem:separation},
and as $c$ was arbitrary, we have $P\oldstar = \tbigsqcup_n\pp n$.
\end{proof}

\section{Approximation and Discrete Measures}

This section contains the proofs of \S\ref{sec:approx}. We need the following auxiliary lemma to prove Theorem~\ref{thm:directed}.
\begin{lemma}\ 
\label{lem:zz}
\begin{enumerate}[{\upshape(i)}]
\item
For any Borel set $B$, $(\rest\mu b)(B) = \mu(\set c{c\cap b\in B})$.
\item
$\rest{(\rest\mu b)}d = \rest\mu{(b\cap d)}$.
\item
If $a,b\in\pfin H$ and $a\subs b$, then $\rest\mu a\sqleq\rest\mu b\sqleq\mu$.
\item
$\mu\sqleq\dirac b$ iff $\mu = \rest\mu b$.
\item
The function $\mu\mapsto\rest\mu b$ is continuous.
\end{enumerate}
\end{lemma}
\begin{proof}
(i)
\begin{align*}
(\rest\mu b)(B) &= \sum_{a\subs b} \mu(\atm ab)\dirac a(B)
= \sum_{a\subs b} \mu(\set c{c\cap b = a})[a\in B]\\
&= \sum_{\substack{a\subs b\\a\in B}} \mu(\set c{c\cap b = a})
= \mu(\bigcup_{\substack{a\subs b\\a\in B}} \set c{c\cap b = a})\\
&= \mu(\set c{c\cap b\in B}).
\end{align*}

(ii)
For any Borel set $B$,
\begin{align*}
(\rest{(\rest\mu b)}d)(B) &= (\rest\mu b)(\set c{c\cap d\in B})\\
&= \mu(\set c{c\cap b\in \set c{c\cap d\in B}})\\
&= \mu(\set c{c\cap b\cap d\in B})\\
&= (\rest\mu{(b\cap d)})(B).
\end{align*}

(iii)
If $a\subs b$, then for any up-closed Borel set $B$,
\begin{gather*}
\set c{c\cap a\in B} \subs \set c{c\cap b\in B} \subs B,\\
\mu(\set c{c\cap a\in B}) \leq \mu(\set c{c\cap b\in B}) \leq \mu(B),\\
(\rest\mu a)(B) \leq (\rest\mu b)(B) \leq \mu(B).
\end{gather*}
As this holds for all $B\in\SO$, we have $\rest\mu a \sqleq \rest\mu b \sqleq \mu$.

(iv)
First we show that $\rest\mu b\sqleq\dirac b$. For any up-closed Borel set $B$,
\begin{align*}
(\rest\mu b)(B)
&= \sum_{a\subs b} \mu(\atm ab)[a\in B]\\
&\leq \sum_{a\subs b} \mu(\atm ab)[b\in B]
= [b\in B]
= \dirac b(B).
\end{align*}
Now we show that if $\mu\sqleq\dirac b$, then $\mu=\rest\mu b$. From
\begin{align*}
d\subs b \wedge d\subs c &\Iff d\subs c\cap b
& c\in B_d &\Iff d\subs c
\end{align*}
we have
\begin{align*}
(\exists d\in F\ d\subs b \wedge c\in B_d) \Iff (\exists d\in F\ c\cap b\in B_d)
\end{align*}
\begin{align*}
c\in\bigcup_{\substack{d\in F\\d\subs b}} B_d \Iff c\cap b\in\bigcup_{d\in F} B_d
\end{align*}
\begin{align}
(\rest\mu b)(\bigcup_{d\in F} B_d) &= \mu(\set c{c\cap b\in\bigcup_{d\in F} B_d})
= \mu(\bigcup_{\substack{d\in F\\d\subs b}} B_d).\label{eq:zz}
\end{align}
Now if $\mu\sqleq\dirac b$, then
\begin{align*}
\mu(\bigcup_{\substack{d\in F\\d\not\subs b}} B_d) \leq \dirac b(\bigcup_{\substack{d\in F\\d\not\subs b}} B_d) = [b\in \bigcup_{\substack{d\in F\\d\not\subs b}} B_d] = 0,
\end{align*}
so
\begin{align*}
\mu(\bigcup_{d\in F} B_d) &\leq \mu(\bigcup_{\substack{d\in F\\d\subs b}} B_d) + \mu(\bigcup_{\substack{d\in F\\d\not\subs b}} B_d)
= \mu(\bigcup_{\substack{d\in F\\d\subs b}} B_d).
\end{align*}
Combining this with \eqref{eq:zz}, we have that $\mu$ and $\rest\mu b$ agree on all $B\in\SO$, therefore they agree everywhere.

(v)
If $\mu\sqleq\nu$, then for all $B\in\SO$,
\begin{align*}
(\rest\mu b)(B) &= \mu(\set c{c\cap b\in B})\\
&\leq \nu(\set c{c\cap b\in B}) = (\rest\nu b)(B).
\end{align*}
Also, for any directed set $D$ of measures and $B\in\SO$,
\begin{align*}
(\rest{(\tbigsqcup D)}b)(B) &= (\tbigsqcup D)(\set c{c\cap b\in B})\\
&= \sup_{\mu\in D}\mu(\set c{c\cap b\in B})
= \sup_{\mu\in D}(\rest\mu b)(B)\\
&= (\tbigsqcup_{\mu\in D}(\rest\mu b))(B),
\end{align*}
therefore $\rest{(\tbigsqcup D)}b=\tbigsqcup_{\mu\in D}(\rest\mu b)$.
\end{proof}

\begin{proof}[Proof of Theorem~\ref{thm:directed}]
The set $\set{\rest\mu b}{b\in\pfin H}$ is a directed set below $\mu$ by Lemma \ref{lem:zz}(iii), and for any up-closed Borel set $B$,
\begin{align*}
(\bigsqcup_{b\in\pfin H} \rest\mu b)(B) &= \sup_{b\in\pfin H}\mu(\set c{c\cap b\in B})\\
&= \mu(\bigcup_{b\in\pfin H}\set c{c\cap b\in B})
= \mu(B).
\end{align*}
An approximating set for $\mu$ is the set
\begin{align*}
L = \set{\sum_{a\subs b}r_a\delta_a}{b\in\pfin H,\ r_a < \mu(\atm ab)\text{ for all } a\neq\emptyset}.
\end{align*}
If $L$ is empty, then $\mu(\atm\emptyset b)=1$ for all finite $b$, in which case $\mu=\dirac\emptyset$ and there is nothing to prove. Otherwise, $L$ is a nonempty directed set whose supremum is $\mu$.

Now we show that $\nu\ll\mu$ for any $\nu\in L$. Suppose $D$ is a directed set and $\mu\sqleq\bigsqcup D$. By Lemma \ref{lem:zz}(iii) and (v),
\begin{align*}
\rest\mu b\sqleq \rest{(\bigsqcup D)}b = \bigsqcup_{\rho\in D}\rest\rho b.
\end{align*}
Moreover, for any $B\in\SO$, $B\neq B_\emptyset$, and $\sum_{a\subs b}r_a\delta_a\in L$,
\begin{align*}
(\rest\nu b)(B) &= \sum_{a\in B}\nu(\atm ab)[a\in B]\\
&< \sum_{a\in B}\mu(\atm ab)[a\in B] = (\rest\mu b)(B). 
\end{align*}
Then $\nu(B_\emptyset)=\rho(B_\emptyset)=1$ for all $\rho\in D$, and for any $B\in\SO$, $B\neq B_\emptyset$,
\begin{align}
(\rest\nu b)(B) < (\rest\mu b)(B) \leq \sup_{\rho\in D}(\rest\rho b)(B)\label{eq:directed}
\end{align}
so there exists $\rho\in D$ such that $(\rest\nu b)(B) \leq (\rest\rho b)(B)$. But since $B$ can intersect $\pH$ in only finitely many ways and $D$ is directed, a single $\rho\in D$ can be found such that \eqref{eq:directed} holds uniformly for all $B\in\SO$, $B\neq B_\emptyset$.
Then $\rest\nu b \sqleq \rho \in D$.
\end{proof}

\begin{proof}[Proof of Corollary \ref{cor:guarded}]
Let $f:\pH\to\pH$ map $a$ to $a\cap b$. This is a continuous function that gives rise to a deterministic kernel. Then for any $B\in\SO$,
\begin{align*}
(P\cmp b)(a,B)
&= P(a,f^{-1}(B))
= P(a,\set c{c\cap b\in B})\\
&= (\rest{P(a,-)}b)(B). 
\end{align*}
\end{proof}

\fi

\end{document}